\documentclass[a4paper,10pt,openright]{article} 
\usepackage[english]{babel} 
\usepackage[latin1]{inputenc} 
\usepackage{amsmath,amssymb,amsthm,amsfonts,amscd,color,enumerate, eucal,latexsym,mathrsfs,mathtools,cancel,tikz}
\usepackage{epsfig}  
\usepackage{simplewick}
\usepackage{hyperref}
\usepackage[all,cmtip]{xy}
\usetikzlibrary{matrix,calc}

\definecolor{NiColor}{RGB}{77,77,255}
\definecolor{NiColoRed}{RGB}{255,77,77}
\definecolor{NiCitation}{RGB}{0,181,26}
\definecolor{ClaColor}{RGB}{255,0,0}
\definecolor{RubColor}{RGB}{255,140,0}

\newtheoremstyle{TheoremStyle}
        {3pt}
        {3pt}
        {\slshape}
        {}
        {\bf}
        {:}
        {.5em}
        {}
\newtheoremstyle{ExampleAndRemarkStyle}
        {3pt}
        {3pt}
        {\slshape}
        {}
        {\bf}
        {:}
        {.5em}
        {}
\newtheoremstyle{ProofStyle}
        {3pt}
        {3pt}
        {}
        {}
        {\bf}
        {:}
        {.5em}
        {}

\theoremstyle{TheoremStyle}
\newtheorem{theorem}{Theorem}
\newtheorem{corollary}[theorem]{Corollary}
\newtheorem{proposition}[theorem]{Proposition}
\newtheorem{lemma}[theorem]{Lemma}
\newtheorem{assumption}[theorem]{Assumption}
\newtheorem{Definition}[theorem]{Definition}
\theoremstyle{ExampleAndRemarkStyle}
\newtheorem{remark}[theorem]{Remark}
\newtheorem{Example}[theorem]{Example} 

\title{%
	On Maxwell's Equations on Globally Hyperbolic Spacetimes with Timelike Boundary
}

\author{%
	Claudio Dappiaggi$^{1,2,3,a}$, Nicol\`o Drago$^{4,5,b}$ 
	and Rubens Longhi$^{1,c}$\vspace{4mm}\\
	{\small $^1$ Dipartimento di Fisica -- Universit{\`a} di Pavia, Via Bassi 6, I-27100 Pavia, Italy.}\vspace{1mm}\\
	{\small $^2$ INFN, Sezione di Pavia -- Via Bassi 6, I-27100 Pavia, Italy.}\vspace{1mm}\\
	{\small $^3$ Istituto Nazionale di Alta Matematica -- Sezione di Pavia, Via Ferrata, 5, 27100 Pavia, Italy.}\vspace{1mm}\\
	{\small $^4$ Dipartimento di Matematica -- Universit{\`a} di Trento, via Sommarive 15, I-38123 Povo (Trento), Italy.}\vspace{1mm}\\
	{\small $^5$ INFN, TIFPA -- via Sommarive 15, I-38123 Povo (Trento), Italy.}\vspace{4mm}\\
	{\footnotesize  ~$^a$ claudio.dappiaggi@unipv.it~,~$^b$ 
		nicolo.drago@unitn.it~,~$^c$ rubens.longhi01@universitadipavia.it}
}

\date{\today}

\begin{document}
	
	\maketitle
	
\begin{abstract}
	We study Maxwell's equation as a theory for smooth $k$-forms on globally hyperbolic spacetimes with timelike boundary as defined by Ak\'e, Flores and Sanchez \cite{Ake-Flores-Sanchez-18}.
	In particular we start by investigating on these backgrounds the D'Alembert - de Rham wave operator $\Box_k$ and we highlight the boundary conditions which yield a Green's formula for $\Box_k$.
	Subsequently, we characterize the space of solutions of the associated initial and boundary value problem under the assumption that advanced and retarded Green operators do exist.
	This hypothesis is proven to be verified by a large class of boundary conditions using the method of boundary triples and under the additional assumption that the underlying spacetime is ultrastatic.
	Subsequently we focus on the Maxwell operator.
	First we construct the boundary conditions which entail a Green's formula for such operator and then we highlight two distinguished cases, dubbed $\delta\mathrm{d}$-tangential and $\delta\mathrm{d}$-normal boundary conditions.
	Associated to these we introduce two different notions of gauge equivalence and we prove that in both cases, every equivalence class admits a representative abiding to the Lorenz gauge.
	We use this property and the analysis of the operator $\Box_k$ to construct and to classify the space of gauge equivalence classes of solutions of the Maxwell's equations with the prescribed boundary conditions.
	As a last step and in the spirit of future applications in the framework of algebraic quantum field theory, we construct the associated unital $*$-algebras of observables proving in particular that, as in the case of the Maxwell operator on globally hyperbolic spacetimes with empty boundary, they possess a non-trivial center.
\end{abstract}



\section{Introduction}\label{Sec: Introduction}

Electromagnetism and the associated Maxwell's equations, written both in terms of the Faraday tensor or of the vector potential, represent one of the most studied models in mathematical physics. On the one hand they are of indisputable practical relevance, while, on the other hand, they are the prototypical example of a gauge theory, which can be still thoroughly and explicitly investigated thanks to the Abelian nature of the underlying gauge group. 

On curved backgrounds the study of this model has attracted a lot of attention not only from the classical viewpoint but also in relation to its quantization. Starting from the early work of Dimock \cite{Dimock:1992ff}, the investigation of Maxwell's equations, generally seen as a theory for differential forms, has been thorough especially in the framework of algebraic quantum field, {\it e.g.} \cite{Fewster:2003ey,Pfenning:2009nx}. One of the key reasons for such interest is related to the fact that electromagnetism has turned out to be one of the simplest examples where the principle of general local covariance, introduced in \cite{Brunetti-Fredenhagen-Verch-03}, does not hold true on account of topological obstructions -- see for example \cite{Benini:2013ita,Benini:2013tra,Dappiaggi:2011zs,Sanders:2012sf}. 

A closer look at all these references and more generally to the algebraic approach unveils that most of the analyses rest on two key data: the choice of a gauge group and of an underlying globally hyperbolic background of arbitrary dimension. While the first one is related to the interpretation of electromagnetism as a theory for the connections of a principal $U(1)$-bundle, the second one plays a key r\^ole in the characterization of the space of classical solutions of Maxwell's equations and in the associated construction of a unital $*$-algebra of observables.
More precisely, every solution of Maxwell's equations identifies via the action of the gauge group an equivalence class of differential forms.
Each of these classes admits a representative, namely a coclosed form which solves a normally hyperbolic partial differential equation, ruled by the D'Alembert - de Rham operator. {Such representative is not unique due to a residual gauge freedom.}
Most notably, since the underlying spacetime is globally hyperbolic, one can rely on classical results, see for example \cite{BGP}, to infer that the D'Alembert - de Rham operator admits unique advanced and retarded fundamental solutions. Not only these can be used to characterize the kernel of such operator, but they also allow both to translate the requirement of considering only coclosed form as a constraint on the admissible initial data and to give an explicit representation for the space of the gauge equivalence classes of solutions of Maxwell's equations. At a quantum level, instead, the fundamental solutions represent the building block to implement the canonical commutation relations within the $*$-algebra of observables, {\it cf.} \cite{Dimock:1992ff}.

Completely different is the situation when we drop the assumption of the underlying background being globally hyperbolic since especially the existence and uniqueness results for the fundamental solutions are no longer valid.
In this paper we will be working in this framework, assuming in particular that the underlying manifold $(M,g)$ is globally hyperbolic and it possesses a timelike boundary, {in the sense of \cite{Ake-Flores-Sanchez-18}, where this concept has been formalized. In other words, to $(M,g)$ one can associate $(\partial M, \iota_M^*g)$, where $\iota_M:\partial M\hookrightarrow M$, is a Lorentzian smooth submanifold. This class of spacetimes contains several notable examples, such as anti-de Sitter (AdS) or asymptotically AdS spacetimes.}
These play a key r\^ole in several models that have recently attracted a lot of attention especially for the study of the properties of Green-hyperbolic operators -- like the wave, the Klein-Gordon or the Dirac operator -- , see for example \cite{Bachelot,Grosse-Murro-18,Holzegel,Wrochna:2016ruq,Vasy}.
From a classical point of view, in order to construct the solutions for any of these equations, initial data assigned on a Cauchy surface are no longer sufficient and it is necessary to supplement them with the choice of a boundary condition. This particular feature prompts the question whether these systems still admit fundamental solutions and, if so, whether they are unique and whether they share the same structural properties of their counterparts in a globally hyperbolic spacetime with empty boundary.  

For the wave operator acting on real scalar functions a complete answer to this question has been given in \cite{Dappiaggi-Drago-Ferreira-19} for static globally hyperbolic spacetimes with a timelike boundary combining spectral calculus with boundary triples, introduced by Grubb in \cite{Grubb68}. 

In this work we will be concerned instead with the study of Maxwell's equations acting on generic $k$-forms, with $0\leq k< m=\dim M$, see \cite{Holzegel:2015swa} for an analysis in terms of the Faraday tensor on an anti-de Sitter spacetime. In comparison with the scalar scenario, the situation is rather different. First of all the dynamics is ruled by the operator $\delta_{k+1}\mathrm{d}_k$ where $\mathrm{d}_k$ is the differential acting on $\Omega^k(M)$, the space of smooth $k$-forms while $\delta_{k+1}$ is the codifferential acting on $(k+1)$-forms. To start with, one can observe that this operator is not formally self-adjoint and thus boundary conditions must be imposed. The admissible ones are established by a direct inspection of the Green's formula for the Maxwell-operator. In between the plethora of all possibilities we highlight two distinguished choices, dubbed $\delta\mathrm{d}$-tangential and $\delta\mathrm{d}$-normal boundary conditions which, in the case $k=0$, reduce to the more common Dirichlet and Neumann boundary conditions. 

As second step we recognize that a notion of gauge group has to be introduced. While a more geometric approach based on interpreting Maxwell's equations in terms of connections on a principal $U(1)$-bundle might be the most desirable approach, we decided to investigate this viewpoint in a future work. We focus instead only on Maxwell's equations as encoding the dynamics of a theory for differential $k$-forms. If the underlying manifold $(M,g)$ would have no boundary the gauge group would be chosen as $\mathrm{d}\Omega^{k-1}(M)$. While at first glance one might wish to keep the same choice, it is immediate to realize that this is possible only for the $\delta\mathrm{d}$-normal boundary condition which is insensitive to any shift of a form by an element of the gauge group. On the contrary, in the other cases, one needs to reduce the admissible gauge transformations so to ensure compatibility with the boundary conditions. 

The next step in our analysis mimics the counterpart when the underlying globally hyperbolic spacetime $(M,g)$ has empty boundary, namely we construct the space of gauge equivalence classes of solutions for Maxwell's equations and we prove that each class admits a non unique representative which is a coclosed $k$-form $\omega$, such that $\Box_k\omega$=0. Using a standard nomenclature, we say that we consider a Lorenz gauge fixing. Here $\Box_k=\mathrm{d}_{k-1}\delta_k+\delta_{k+1}\mathrm{d}_k$ is the D'Alembert - de Rham wave operator which is known to be normally hyperbolic, see e.g. \cite{Pfenning:2009nx}. On the one hand we observe that the mentioned non uniqueness is related to a residual gauge freedom which can be fully accounted for. On the other hand we have reduced the characterization of the equivalence classes of solutions of Maxwell's equations to studying the d'Alembert - de Rham wave operator. 

This can be analyzed similarly to the wave operator acting on scalar functions as one could imagine since the two operators coincide for $k=0$. Therefore we study $\Box_k$ independently, first identifying via a Green's formula a collection of admissible boundary conditions, Subsequently, under the assumption that advanced and retarded Green's operators exist, we characterize completely the space of solutions of the equation $\Box_k\omega=0$, $\omega\in\Omega^k(M)$ with prescribed boundary condition. At this stage we highlight the main technical obstruction which forces us to consider only two distinguished boundary conditions for the Maxwell operator. As a matter of facts, we show that, although at an algebraic level it holds always $\delta_k\circ\Box_k=\Box_{k-1}\circ\delta_k$, the counterpart at the level of fundamental solutions is verified only for specific choices of the boundary condition. This leads to an obstruction in translating the Lorenz gauge condition of working only with coclosed $k$-form to a constraint in the admissible initial data. This failure does not imply that the Lorenz gauge is ill-defined, but only that, for a large class of boundary conditions, one needs to envisage a strategy different from the one used on globally hyperbolic spacetimes with empty boundary in order to study the underlying problem. 

It it important to mention that it is beyond our current knowledge verifying whether our assumption on the existence of fundamental solutions is always true. We expect that a rather promising avenue consists of adapting to the case in hand the techniques and the ideas discussed in \cite{Derezinski:2017hvk} and in \cite{Gannot:2018jkg}, but this is certainly a challenging task, which we leave for future work. On the contrary we test our assumption in the special case of ultrastatic, globally hyperbolic spacetimes with timelike boundary. In this scenario we adopt the techniques used in \cite{Dappiaggi-Drago-Ferreira-19} proving that advanced and retarded fundamental solutions do exist for a large class of boundary conditions, including all those of interest for our analysis. 

Finally we give an application of our result inspired by the quantization of Maxwell's equations in the algebraic approach to quantum field theory. While this framework has been extremely successful on a generic globally hyperbolic spacetime with empty boundary, only recently the case with a timelike boundary has been considered, see {\it e.g.} \cite{Benini:2017dfw,Dappiaggi:2017wvj,Dybalski:2018egv,MSTW19,Zahn:2015due}. In particular we focus on the construction of a unital $*$-algebra of observables for Maxwell's equations both with $\delta\mathrm{d}$-tangential and $\delta\mathrm{d}$-normal boundary condition and we prove that in both cases one can always find a non trivial Abelian ideal. This is the signature that, also in presence of a timelike boundary, one cannot expect that the principle of general local covariance holds true in its original form. 

{To conclude, we emphasize that, from a physical viewpoint our analysis has two key goals. The first, more evident, is to discuss a rigorous quantization scheme for Maxwell equations in presence of a class of boundaries which are expected to play a r\^ole in many applications. A rather natural example is the Casimir effect for photons, which can be modeled as a free electromagnetic field confined between two perfectly isolating and infinitely extended parallel plates in Minkowski spacetime. The second instead consists of highlighting the necessity in presence of boundaries of revising the structure of the underlying gauge group, which changes according to the boundary conditions chosen. In the prospect of analyzing in future works interacting models, this has far reaching consequences.}

\vskip .3cm

\noindent The paper is organized as follows: In Section \ref{Sec: Geometric Data} we introduce the notion of globally hyperbolic spacetime with timelike boundary as well as all the relevant space of differential $k$-forms. In addition we recall the basic definitions of differential and codifferential operator and we introduce two distinguished maps between bulk and boundary forms. Section \ref{Sec: the algebra for the vector potential with Dirichlet boundary conditions} contains the core of this paper. For clarity purposes, we start in Subsection \ref{Sec: on the D'Alembert--de Rham wave operator} from the analysis of the D'Alembert - de Rham wave operator $\Box_k$. To begin with we study a class of boundary conditions which implement the Green's formula, hence making $\Box_k$ a formally self-adjoint operator. Subsequently we assume that, for a given boundary condition, advanced and retarded Green's operators exist and we codify the information of the space of classical solutions of the underlying dynamics in terms of a short exact sequence, similar to the standard one when the underlying globally hyperbolic spacetime has no boundary, {\it cf.} \cite{BGP}. In addition we discuss the interplay between the fundamental solutions and the differential/codifferential operator. In Subsection \ref{Sec: on the Maxwell operator} we focus instead on Maxwell's equations. First we investigate which boundary conditions can be imposed so that the operator ruling the dynamics is formally self-adjoint. Subsequently we introduce the $\delta\mathrm{d}$-tangential and the $\delta\mathrm{d}$-normal boundary conditions together with an associated gauge group. Using these data, we prove that the equivalence classes of solutions of Maxwell's equations, always admit a representative in the Lorenz gauge, which obeys an equation of motion ruled by $\Box_k$. Such solution is non unique in the sense that a residual gauge freedom exists. Nonetheless, using the fundamental solutions of the D'Alembert - de Rham wave operator, we are able to characterize the above equivalence classes in terms of suitable initial data. To conclude, in Section \ref{Sec: Algebra of observables for Sol(M)} we use the results from the previous parts to construct a unital $*$-algebra of observables associated to Maxwell's equations with $\delta\mathrm{d}$-tangential and $\delta\mathrm{d}$-normal boundary conditions. In particular we prove that in all cases there exists an Abelian $*$-ideal. In Appendix \ref{App: Existence of fundamental solutions on ultrastatic spacetimes} we prove that our assumption on the existence of fundamental solutions is verified whenever the underlying spacetime is ultrastatic. Finally in Appendix \ref{App: an explicit decomposition} it is proven an explicit decomposition for $k$-forms on globally hyperbolic spacetimes, which plays a key r\^ole is some proofs in the main body of the paper. In Appendix \ref{App: Poincare duality for manifold with boundary} we recall the basic notion of relative cohomology for manifolds with boundaries as well as the associated Poincar\'e-Lefschetz duality.

\section{Geometric Data}\label{Sec: Geometric Data}

In this subsection, our goal is to fix notations and conventions, as well as to summarize the main geometric data, which play a key r\^ole in our analysis. Following the standard definition, see for example \cite[Ch. 1]{Lee}, $M$ indicates a smooth, second-countable, connected, oriented manifold of dimension $m\geq 2$ , with smooth boundary $\partial M$, assumed for simplicity to be connected. We assume also that $M$ admits a finite good cover. A point $p\in M$ such that there exists an open neighbourhood $U$ containing $p$, diffeomorphic to an open subset of $\mathbb{R}^m$, is called an {\em interior point} and the collection of these points is indicated with $\mathrm{Int}(M)\equiv\mathring{M}$. As a consequence $\partial M\doteq M\setminus\mathring{M}$, if non-empty, can be read as an embedded submanifold $(\partial M,\iota_{\partial M})$ of dimension $m-1$ with $\iota_{\partial M}\in C^\infty(\partial M; M)$.

In addition we endow $M$ with a smooth Lorentzian metric $g$ of signature $(-,+,...,+)$ and consider only those cases in which $\iota_{\partial M}^*g$ identifies a Lorentzian metric on $\partial M$ and $(M,g)$ is time oriented. As a consequence $(\partial M,\iota^*_{\partial M}g)$ acquires the induced time orientation and we say that $(M,g)$ has a {\em timelike boundary}. 

Since we will be interested particularly in the construction of advanced and retarded fundamental solutions for normally hyperbolic operators, we focus our attention on a specific class of Lorentzian manifolds with timelike boundary, namely those which are globally hyperbolic. While, in the case of $\partial M=\emptyset$ this is a standard concept, in presence of a timelike boundary it has been properly defined and studied recently in \cite{Ake-Flores-Sanchez-18}. Summarizing part of their constructions and results, we say that a time-oriented, Lorentzian manifold with timelike boundary $(M,g)$ is {\em causal} if it possesses no closed, causal curve, while it is {\em globally hyperbolic} if it is causal and, for all $p,q\in M$, $J^+(p)\cap J^-(q)$ is either empty or {compact. Here $J^\pm$ denote the causal future and past, {\it cf.} \cite[Sec. 1.3]{BGP} }. These conditions entail the following consequences, see \cite[Th. 1.1 \& 3.14]{Ake-Flores-Sanchez-18}:

\begin{theorem}\label{Thm: globally hyperbolic spacetime with time-like boundary}
	Let $(M,g)$ be a time-oriented Lorentzian manifold with timelike boundary of dimension $\dim M=m\geq 2$.
	Then the following conditions are equivalent:
	\begin{enumerate}
		\item
		$(M,g)$ is globally hyperbolic ;
		\item
		$(M,g)$ possesses a Cauchy surface, namely an achronal subset of $M$ which is intersected only once by every inextensible timelike curve ;
		\item 
		$(M,g)$ is isometric to $\mathbb{R}\times\Sigma$ endowed with the line-element
		\begin{equation}\label{eq:line_element}
			\mathrm{d}s^2=-\beta \mathrm{d}\tau^2+h_\tau\,,
		\end{equation}
	where $\tau:M\to\mathbb{R}$ is a Cauchy temporal function\footnote{Given a generic time oriented Lorentzian manifold $(N,\tilde{g})$, a Cauchy temporal function is a map $\tau:M\to\mathbb{R}$ such that its gradient is timelike and past-directed, while its level surfaces are Cauchy hypersurfaces.}, whose gradient is tangent to $\partial M$, $\beta\in C^\infty(\mathbb{R}\times\Sigma;(0,\infty))$ while $\mathbb{R}\ni \tau\to (\{\tau\}\times\Sigma,h_\tau)$ identifies a one-parameter family of $(m-1)-$dimensional spacelike, Riemannian manifolds with boundaries. Each $\{\tau\}\times\Sigma$ is a Cauchy surface for $(M,g)$.
	\end{enumerate}
\end{theorem} 

Henceforth we will be tacitly assuming that, when referring to a globally hyperbolic spacetime with timelike boundary $(M,g)$, we work directly with \eqref{eq:line_element} and we shall refer to $\tau$ as the time coordinate.
Furthermore each Cauchy surface $\Sigma_\tau\doteq\{\tau\}\times\Sigma$ acquires an orientation induced from that of $M$.
In addition we shall say that $(M,g)$ is {\em static} if it possesses an {irrotational} timelike Killing vector field $\chi\in\Gamma(TM)$ whose restriction to $\partial M$ is tangent to the boundary, {\it i.e.} $g_p(\chi,\nu)=0$ for all $p\in\partial M$ where $\nu$ is the outward pointing, unit vector, normal to the boundary at $p$.
With reference to \eqref{eq:line_element} and for simplicity, we identify $\chi$ with $\partial_\tau$.
Thus the condition of being static translates into the constraint that both $\beta$ and $h_\tau$ are independent from $\tau$. 
If in addition $\beta=1$ we call $(M,g)$ {\em ultrastatic}.

\vskip .2cm

On a Lorentzian spacetime $(M,g)$ with timelike boundary we consider $\Omega^k(M)$, $0\leq k\leq\dim M$, the space of real valued smooth $k$-forms.
A particular r\^ole will be played by the support of the forms that we consider. In the following definition we introduce the different possibilities that we will consider, which are a generalization of the counterpart used for scalar fields which corresponds in our scenario to $k=0$, \textit{cf.} \cite{Baer-15}.
\begin{Definition}\label{Def: space of forms}
	Let $(M,g)$ be a Lorentzian spacetime with timelike boundary. We denote with 
	\begin{enumerate}
		\item 	$\Omega_{\mathrm{c}}^k(M)$ the space of smooth $k$-forms with compact support in $M$ while {we denote} with $\Omega_{\mathrm{c}}^k(\mathring{M})\subset\Omega^k_{\mathrm{c}}(M)$ the collection of smooth and compactly supported $k$-forms $\omega$ such that $\textrm{supp}(\omega)\cap\partial M=\emptyset$.
		\item
		$\Omega_{\mathrm{spc}}^k(M)$ (\textit{resp}. $\Omega_{\mathrm{sfc}}^k(M)$) the space of strictly past compact (\textit{resp.} strictly future compact) $k$-forms, that is the collection of $\omega\in\Omega^k(M)$ such that there exists a compact set $K\subseteq M$ for which $J^+(\textrm{supp}(\omega))\subseteq J^+(K)$ (\textit{resp.} $J^-(\textrm{supp}(\omega))\subseteq J^-(K)$), where $J^\pm$ denotes the causal future and the causal past in $M$.  Notice that $\Omega_{\mathrm{sfc}}^k(M)\cap\Omega_{\mathrm{spc}}^k(M)=\Omega_{\mathrm{c}}^k(M)$.
		\item
		$\Omega_{\mathrm{pc}}^k(M)$ (\textit{resp}. $\Omega_{\mathrm{fc}}^k(M)$) the space of {past} compact (\textit{resp.} {future} compact) $k$-forms, that is, $\omega\in\Omega^k(M)$ for which
		${\rm supp}(\omega)\cap J^-(K)$ (\textit{resp.} ${\rm supp}(\omega)\cap J^+(K)$) is compact for all compact $K\subset M$.
		\item $\Omega_{\mathrm{tc}}^k(M)\doteq\Omega_{\mathrm{fc}}^k(M)\cap\Omega_{\mathrm{pc}}^k(M)$, the space of timelike compact $k$-forms.
		\item $\Omega_{\mathrm{sc}}^k(M)\doteq\Omega_{\mathrm{sfc}}^k(M)\cup\Omega_{\mathrm{spc}}^k(M)$, the space of spacelike compact $k$-forms.
	\end{enumerate}
\end{Definition}

We indicate with $\mathrm{d}_k\colon\Omega^k(M)\to\Omega^{k+1}(M)$ the exterior derivative and, being $(M,g)$ oriented, we can identify a unique, metric-induced, Hodge operator $\star_k:\Omega^k(M)\to\Omega^{m-k}(M)$, $m=\dim M$ such that, for all $\alpha,\beta\in\Omega^k(M)$,
{
$\alpha\wedge\star_k\beta=(\alpha,\beta)_{k,g}\mu_g$, where $\wedge$ is the exterior product of forms, $\mu_g$ the metric induced volume form, while $(\,,\,)_{k,g}$ is the pairing between $k$-forms induced by the metric $g$.
}
In addition one can define a pairing between $k$-forms as
\begin{align}\label{Eq: pairing between k-forms}
	(\alpha,\beta)_k
	\doteq\int\alpha\wedge\star_k\beta\,,
\end{align}
where $\alpha,\beta\in\Omega^k(M)$ are such that $\operatorname{supp}(\alpha)\cap\operatorname{supp}(\beta)$ is compact.
Since $M$ is endowed with a Lorentzian metric it holds that, when acting on smooth $k$-forms, $\star_k^{-1}=(-1)^{k(m-k)-1}\star_{m-k}$.
Combining these data first we define the {\em codifferential} operator $\delta_k:\Omega^{k+1}(M)\to\Omega^k(M)$ as $\delta_k\doteq(-1)^k\star_{k-1}^{-1}\circ\,\mathrm{d}_{m-k}\circ\star_k$.
Secondly we introduce the {\em D'Alembert-de Rham} wave operator $\Box_k:\Omega^k(M)\to\Omega^k(M)$ such that $\Box_k\doteq \mathrm{d}_{k-1}\delta_k+\delta_{k+1} \mathrm{d}_k$, as well as the {\em Maxwell} operator $\delta_{k+1}\mathrm{d}_k:\Omega^k(M)\to\Omega^k(M)$.
Observe, furthermore, that $\Box_k$ differs {from the D'Alembert wave operator $g^{ab}\nabla_a\nabla_b$} acting on $k$-forms by $0$-order term built out of the metric and whose explicit form depends from the value of $k$, see for example \cite[Sec. II]{Pfenning:2009nx}.

\begin{remark}
	For notational convenience, in the following we shall drop all subscripts $_k$ since the relevant value will be clear case by case from the context. Hence, unless stated otherwise, all statements of this paper apply to all $k$ such that $0\leq k \leq m=\dim M$.
\end{remark}

To conclude the section, we focus on the boundary $\partial M$ and on the interplay with $k$-forms lying in $\Omega^k(M)$. The first step consists of defining two notable maps. These relate $k$-forms defined on the whole $M$ with suitable counterparts living on $\partial M$ and,
{in the special case of $k=1$, they boil down to the restriction to the boundary either of the tangent component of a $1$-form or of its component conormal to $\partial M$.} For later convenience we consider in the following definition a slightly more general scenario, namely a codimension $1$ smoothly embedded submanifold $N\hookrightarrow M$.

\begin{remark}
Since we feel that some confusion might arise, we denote the paring between forms on $\partial M$ with $(\;,\;)_\partial$.
\end{remark}

\begin{Definition}\label{Def: tangential and normal component}
	Let $(M,g_M)$ be a smooth Lorentzian manifold and let $\iota_N\colon N\to M$ be a codimension $1$ smoothly embedded submanifold of $M$ with induced metric $g_N:=\iota_N^*g_M$.
	We define the \textit{tangential} and \textit{normal} components relative to $N$ as 
	\begin{subequations}\label{Eq: tangential and normal maps}
		\begin{align}
			&\mathrm{t}_N\colon\Omega^k(M)\to\Omega^k(N)\,,\qquad\quad\omega\mapsto
			\mathrm{t}_N\omega:=\iota_N^*\omega\,,\\
			&\mathrm{n}_N\colon\Omega^k(M)\to\Omega^{k-1}(N)\,,\qquad\omega\mapsto
			\mathrm{n}_N\omega:=\star_N^{-1}\mathrm{t}_N\star_M\omega\,,
		\end{align}
	\end{subequations}
	where $\star_M,\star_N$ denote the Hodge dual over $M,N$ respectively.
	In particular, for all $k\in\mathbb{N}\cup\{0\}$ we define
	\begin{align}\label{Eq: k-forms with vanishing tangential or normal component}
		\Omega_{\mathrm{t}_N}^k(M)\doteq\lbrace\omega\in\Omega^k(M)\;|\;\mathrm{t}_N\omega=0\rbrace\,,\qquad
		\Omega_{\mathrm{n}_N}^k(M)\doteq\lbrace\omega\in\Omega^k(M)\;|\;\mathrm{n}_N\omega=0\rbrace\,.
	\end{align}
	Similarly we will use the symbols $\Omega_{\mathrm{c,t_N}}^k(M)$ and $\Omega_{\mathrm{c,n_N}}^k(M)$ when we consider only smooth, compactly supported $k$-forms.
\end{Definition}

\begin{remark}
	In this paper the r\^ole of $N$ will be played often by $\partial M$. In this case, we shall drop the subscript form Equation \eqref{Eq: tangential and normal maps}, namely $\mathrm{t}\equiv\mathrm{t}_{\partial M}$ and $\mathrm{n}\equiv\mathrm{n}_{\partial M}$. Furthermore, the differential and the codifferential operators on $\partial M$ will be denoted, respectively, as $\mathrm{d}_\partial$, $\delta_\partial$. 
\end{remark}

\noindent As last step, we observe that \eqref{Eq: tangential and normal maps} together with \eqref{Eq: k-forms with vanishing tangential or normal component} entail the following series of identities on $\Omega^k(M)$ for all $k\in\mathbb{N}\cup\{0\}$.
\begin{subequations}\label{Eq: relations between d,delta,t,n}
	\begin{gather}
	\label{Eq: relations-bulk}
	\star\delta=(-1)^k\mathrm{d}\star\,,\quad
	\delta\star=(-1)^{k+1}\star\mathrm{d}\,,\\
	\label{Eq: relations-bulk-to-boundary}
	\star_\partial\mathrm{n}=\mathrm{t}\star\,,\quad
	\star_\partial\mathrm{t}=(-1)^{(m-k)}\mathrm{n}\star\,,\quad
	\mathrm{d}_\partial\mathrm{t}=\mathrm{t}\mathrm{d}\,,\quad
	\delta_\partial\mathrm{n}=-\mathrm{n}\delta\,.
	\end{gather}
\end{subequations}
A notable consequence of \eqref{Eq: relations-bulk-to-boundary} is that, while on globally hyperbolic spacetimes with empty boundary, the operators $\mathrm{d}$ and $\delta$ {are formal adjoints of each other}, in the case in hand, the situation is different.
A direct application of Stokes' theorem yields that 
\begin{align}\label{Eq: boundary terms for delta and d}
(\mathrm{d}\alpha,\beta)-(\alpha,\delta\beta)=
(\mathrm{t}\alpha,\mathrm{n}\beta)_\partial\,,
\end{align}
where the pairing in the right-hand side is the one associated to forms living on $\partial M$ and where $\alpha\in\Omega^k(M)$ and $\beta\in\Omega^{k+1}(M)$ are arbitrary, though such that $\textrm{supp}(\alpha)\cap\textrm{supp}(\beta)$ is compact. 
In connection to the operators $\mathrm{d}$ and $\delta$ we shall employ the notation
\begin{equation}\label{Eq: kernel of d and delta}
\Omega^k_{\mathrm{d}}(M)=\{\omega\in\Omega^k(M)\;|\;\mathrm{d}\omega=0\}\,,\qquad
\Omega^k_\delta(M)=\{\omega\in\Omega^k(M)\;|\;\delta\omega=0\}\,,
\end{equation}
where $k\in\mathbb{N}$. Similarly we shall indicate with $\Omega^k_{\sharp,\delta}(M)\doteq\Omega^k_{\sharp}(M)\cap\Omega^k_\delta(M)$ and $\Omega^k_{\sharp,\mathrm{d}}(M)\doteq\Omega^k_{\sharp}(M)\cap\Omega^k_{\mathrm{d}}(M)$ where $\sharp\in\{\mathrm{c,sc,pc,fc,tc}\}$.

\begin{remark}\label{Rmk: surjectivity of t,n,tdelta,nd}
	With reference to Definition \ref{Def: tangential and normal component}, observe that the following linear map is surjective:
	\begin{align*}
	\Omega^k(M)\ni\omega\mapsto(\mathrm{n}\omega,\mathrm{t}\omega,\mathrm{t}\delta\omega,\mathrm{nd}\omega)\in
	\Omega^{k-1}(\partial M)\times
	\Omega^k(\partial M)\times
	\Omega^{k-1}(\partial M)\times
	\Omega^k(\partial M)\,.
	\end{align*}
	The proof of this claim is based on a local computation similar to the one in the proof of Lemma \ref{Lem: equivalence between M-boundary conditions and Sigma-boundary conditions}.
	For all relatively compact open subset $U_{\partial M}\subset\partial M$ we consider a open neighbourhood $U\subseteq M$ of the form $U=[0,\epsilon)\times U_{\partial M}$ built out the exponential map $\exp_M$. In addition we can fine tune $U$ in such a way, that calling $x$, the coordinate built via $\exp_M$ out of the outward pointing, normal vector field at each point $p\in U$, the smooth function $N=g(\partial_x,\partial_x)$ is strictly positive.
	Let $U_x\doteq\{x\}\times U_{\partial M}$ for $x\in [0,\epsilon)$ and let $\mathrm{t}_{U_x}$, $\mathrm{n}_{U_x}$ be the corresponding tangential and normal maps -- \textit{cf.} Definition \ref{Def: tangential and normal component}.
	Therefore, we can split $\omega\in\Omega^k(M)$ as
	\begin{align*}
		\omega|_{U_x}
		=\mathrm{t}_{U_x}\omega
		+\mathrm{n}_{U_x}\omega\wedge N^{\frac{1}{2}}\mathrm{d}x\,.
	\end{align*}
	It descends from Definition \ref{Def: tangential and normal component} that
	\begin{align*}
		\mathrm{t}\omega|_{U_{\partial M}}=
		\mathrm{t}_{U_x}\omega|_{x=0}\,,\qquad
		\mathrm{n}\omega|_{U_{\partial M}}=
		\mathrm{n}_{U_x}\omega|_{x=0}\,.
	\end{align*}
	Applying the differential $\mathrm{d}$ to the local splitting of $\omega$ yields
	\begin{align*}
		\mathrm{nd}\omega|_{U_{\partial M}}=
		N^{-\frac{1}{2}}\partial_x\mathrm{t}_{U_x}\omega|_{x=0}
		+N^{-\frac{1}{2}}\mathrm{d}_{\partial U}(N^{\frac{1}{2}}\mathrm{n}_{U_x}\omega)|_{x=0}\,.
	\end{align*}
	Moreover, the Hodge dual $\star_U\omega$ can be computed as
	\begin{align*}
		\star_U\omega|_{U_x}
		=\star_{U_x}\mathrm{n}_{U_x}\omega
		+\star_{U_x}\mathrm{t}_{U_x}\omega\wedge N^{\frac{1}{2}}\mathrm{d}x\,,
	\end{align*}
	where $\star_{U_x}\colon\Omega^\bullet(U_x)\to\Omega^{m-1-\bullet}(U_x)$ denotes the Hodge dual on $U_x$.
	Taking into account equations \eqref{Eq: relations between d,delta,t,n} we find
	\begin{align*}
		\mathrm{t}\delta\omega|_{U_{\partial M}}
		&=(-1)^{(k+1)(m-k)-1}\star_{\partial U}\mathrm{nd}\star_U\omega|_{U_{\partial M}}
		\\&=(-1)^{(k+1)(m-k)-1}\star_{\partial U}\bigg[
		N^{-\frac{1}{2}}\partial_x\star_{U_x}\mathrm{n}_{U_x}\omega|_{x=0}
		+N^{-\frac{1}{2}}\mathrm{d}_{\partial U}(N^{\frac{1}{2}}\star_{U_x}\mathrm{t}_{U_x}\omega)|_{x=0}
		\bigg]\,.
	\end{align*}
	The claim follows from the fact that $\mathrm{t}_{U_x}\omega|_{x=0}$, $\mathrm{n}_{U_x}\omega|_{x=0}$, $\partial_x\mathrm{t}_{U_x}\omega|_{x=0}$ and $\partial_x\mathrm{n}_{U_x}\omega|_{x=0}$ are functionally independent.
\end{remark}

\begin{remark}
	The normal map $\mathrm{n}:\Omega^k(M)\to\Omega^{k-1}(\partial M)$ can be equivalently read as $\nu\operatorname{\lrcorner}\omega$, the contraction on $\partial M$ between $\omega\in\Omega^k(M)$ and the vector field $\nu\in\Gamma(TM)|_{\partial M}$ which corresponds at each point $p\in\partial M$ to the outward pointing unit vector, normal to $\partial M$.
\end{remark}

\section{Maxwell's Equations and Boundary Conditions}\label{Sec: the algebra for the vector potential with Dirichlet boundary conditions}

In this section we analyze the space of solutions of Maxwell's equations for arbitrary $k$-forms on a globally hyperbolic spacetime with timelike boundary $(M,g)$. We proceed in two separate steps.
First we focus our attention on the D'Alembert - de Rham wave operator $\Box=\delta \mathrm{d}+\mathrm{d}\delta$ acting on $\Omega^k(M)$.
We identify a class of boundary conditions which correspond to imposing that the underlying system is closed ({\it i.e.} the symplectic flux across $\partial M$ vanishes) and we characterize the kernel of the operator in terms of its advanced and retarded fundamental solutions.
These are assumed to exist and, following the same strategy employed in \cite{Dappiaggi-Drago-Ferreira-19} for the scalar wave equation, we prove that this is indeed the case whenever $(M,g)$ is an ultrastatic spacetime, \textit{cf.} Appendix \ref{App: Existence of fundamental solutions on ultrastatic spacetimes}.

In the second part of the section we focus instead on the Maxwell operator $\delta\mathrm{d}\colon\Omega^k(M)\to\Omega^k(M)$.
In order to characterize its kernel we will need to discuss the interplay between the choice of boundary condition and that of gauge fixing.
This represents the core of this part of our work.

\subsection{On the D'Alembert--de Rham wave operator}\label{Sec: on the D'Alembert--de Rham wave operator}

Consider the operator $\Box:\Omega^k(M)\to\Omega^k(M)$, where $(M,g)$ is a globally hyperbolic spacetime with timelike boundary of dimension $\dim M=m\geq 2$. Then, for any pair $\alpha,\beta\in\Omega^k(M)$ such that $\textrm{supp}(\alpha)\cap\textrm{supp}(\beta)$ is compact, the following Green's formula holds true:
\begin{align}\label{Eq: boundary terms for wave operator}
	(\Box\alpha,\beta)-(\alpha,\Box\beta)=
	(\mathrm{t}\delta\alpha,\mathrm{n}\beta)_\partial-
	(\mathrm{n}\alpha,\mathrm{t}\delta\beta)_\partial-
	(\mathrm{n}\mathrm{d}\alpha,\mathrm{t}\beta)_\partial+
	(\mathrm{t}\alpha,\mathrm{n}\mathrm{d}\beta)_\partial\,,
\end{align}
where $\mathrm{t,n}$ are the maps introduced in Definition \ref{Def: tangential and normal component}, while $(,)$ and $(,)_\partial$ are the standard, metric induced pairing between $k$-forms respectively on $M$ and on $\partial M$.
In view of Definition \ref{Def: tangential and normal component}, it descends that the right-hand side of \eqref{Eq: boundary terms for wave operator} vanishes automatically if we restrict our attention to $\alpha\in\Omega_{\mathrm{c}}^k(\mathring{M})$ or $\beta\in\Omega^k_{\mathrm{c}}(\mathring{M})$, but boundary conditions must be imposed for the same property to hold true on a larger set of $k$-forms. From a physical viewpoint this requirement is tantamount to imposing that the system described by $k$-forms obeying the D'Alembert--de Rham wave equation is closed.

\begin{lemma}\label{Lemma: boundary condition}
	Let $f,f^\prime\in C^\infty(\partial M)$ and let 
	\begin{equation}\label{Eq: f,f' boundary condition}
	\Omega^k_{f,f^\prime}(M)\doteq\{\omega\in\Omega^k(M)\;|\;\mathrm{nd}\omega=f\mathrm{t}\omega\,,\;\mathrm{t}\delta\omega=f^\prime \mathrm{n}\omega \}.
	\end{equation}
	Then, $\forall\alpha,\beta\in\Omega^k_{f,f^\prime}(M)$, $0\leq k\leq m=\dim M$ such that $\textrm{supp}(\alpha)\cap\textrm{supp}(\beta)$ is compact, it holds 
	$$(\Box\alpha,\beta)-(\alpha,\Box\beta)=0.$$
\end{lemma}

\begin{proof}
	This is a direct consequence of \eqref{Eq: boundary terms for wave operator} together with the property that, for every $f\in C^\infty(\partial M)$ and for every $\alpha\in\Omega^k(\partial M)$, $\star_\partial(f\alpha)=f(\star_\partial\alpha)$. In addition observe that the assumption on the support of $\alpha$ and $\beta$ descends also to the forms present in each of the pairing in the right hand side of \eqref{Eq: boundary terms for wave operator}.
\end{proof}

\begin{remark}\label{Rem: extreme cases}
	In Lemma \ref{Lemma: boundary condition} two cases are quite peculiar. As a matter of fact, if $k=m=\dim M$ the first condition becomes empty since $\mathrm{d}\omega = \mathrm{t}\omega=0$ for all $\omega\in\Omega^m(M)$.
	Similarly, if $k=0$, the second condition does not bring any constraint since $\delta\omega=\mathrm{n}\omega=0$ for all $\omega\in\Omega^0(M)$.
	In this case equation \eqref{Eq: f,f' boundary condition} reduces to Robin boundary conditions, which were studied in \cite{Dappiaggi-Drago-Ferreira-19}.
\end{remark}

\begin{remark}
	It is important to stress that the boundary conditions defined in Lemma \ref{Lemma: boundary condition} are not the largest class which makes the right hand side \eqref{Eq: boundary terms for wave operator} vanish.
	As a matter of fact one can think of additional possibilities similar to the so-called Wentzell boundary conditions, which were considered in the scalar scenario, see {\it e.g.} \cite{Dappiaggi-Drago-Ferreira-19,Dappiaggi:2018pju,Zahn:2015due}.
\end{remark}

Lemma \eqref{Lemma: boundary condition} individuates therefore a class of boundary conditions which makes the operator $\Box$ formally self-adjoint.
In between all these possibilities we highlight those which are of particular interest to our analysis -- \textit{cf.} Theorem \ref{Thm: assumption theorem}.

\begin{Definition}\label{Def: Dirichlet, Box-tangential, Box-normal, Robin Box-tangential, Robin Box-normal boundary conditions}
Let $(M,g)$ be a globally hyperbolic spacetime with timelike boundary and let $f\in C^\infty(\partial M)$.
We call
\begin{enumerate}
	\item space of $k$-forms with {\em Dirichlet} boundary condition
	\begin{align}\label{Eq: Dirichlet k-forms}
		\Omega^k_{\mathrm{D}}(M)\doteq\{\omega\in\Omega^k(M)\;|\;\mathrm{t}\omega=0\;,\;\mathrm{n}\omega=0\}\,,
	\end{align}
	\item space of $k$-forms with {\em $\Box$-tangential} boundary condition
	\begin{align}\label{Eq: Box-tangential k-forms}
		\Omega^k_\parallel(M)\doteq\{\omega\in\Omega^k(M)\;|\;\mathrm{t}\omega=0\;,\;\mathrm{t}\delta\omega=0\}\,,
	\end{align}
	\item space of $k$-forms with {\em $\Box$-normal} boundary condition
	\begin{align}\label{Eq: Box-normal k-forms}
		\Omega^k_\perp(M)\doteq\{\omega\in\Omega^k(M)\;|\;\mathrm{n}\omega=0\;,\;\mathrm{nd}\omega=0\}\,.
	\end{align}
	\item space of $k$-forms with {\em Robin $\Box$-tangential} boundary condition
		\begin{align}\label{Eq: Robin Box-tangential k-forms}
			\Omega^k_{f_\parallel}(M)\doteq\{\omega\in\Omega^k(M)\;|\;\mathrm{t}\delta\omega=f\mathrm{n}\omega\;,\;\mathrm{t}\omega=0\}\,,
		\end{align}
	\item space of $k$-forms with {\em Robin $\Box$-normal} boundary condition
		\begin{align}\label{Eq: Robin Box-normal k-forms}
		\Omega^k_{f_\perp}(M)\doteq\{\omega\in\Omega^k(M)\;|\;\mathrm{nd}\omega=f\mathrm{t}\omega\;,\;\mathrm{n}\omega=0\}\,,
		\end{align}
\end{enumerate}
Whenever the domain of the operator $\Box$ is restricted to one of these spaces, we shall indicate it with symbol $\Box_\sharp$ where $\sharp\in\{\mathrm{D},\parallel,\perp,f_\parallel,f_\perp\}$.
\end{Definition}

\begin{remark}\label{Rmk: on nomenclature for Dirichlet and Neumann boundary conditions}
	Since per definition $\delta\Omega^0(M)=\{0\}=\mathrm{n}\Omega^0(M)$, we observe that $\Omega^0_{\mathrm{D}}(M)=\Omega^0_\parallel(M)$.
	In particular we have
	\begin{align*}
		\Omega^0_{\mathrm{D}}(M)\doteq\{\omega\in C^\infty(M)\;|\;\mathrm{t}\omega=\omega|_{\partial M}=0\},\qquad
		\Omega^0_{\perp}(M)\doteq\{\omega\in C^\infty(M)\;|\;\mathrm{nd}\omega=\nu(\mathrm{d}\omega)|_{\partial M}=0\}\,,
	\end{align*}
	where, for all $p\in\partial M$, $\nu_p$ coincides with the outward pointing unit vector, normal to the boundary.
	These two options coincide with the standard Dirichlet and Neumann boundary conditions for scalar functions.
	Moreover for $f=0$ we have $\Omega^k_{f_\parallel}(M)=\Omega^k_{\parallel}(M)$ as well as $\Omega^k_{f_\perp}(M)=\Omega^k_\perp(M)$.
	\\
	Finally it is worth mentioning that, for a static spacetime $(M,g)$, the boundary conditions $1$-$3$, introduced in Definition \ref{Def: Dirichlet, Box-tangential, Box-normal, Robin Box-tangential, Robin Box-normal boundary conditions}, are themselves static, that is they do not depend explicitly on the time coordinate $\tau$.
	A similar statement holds true for $f_\perp$, $f_\parallel$ boundary conditions provided that $f\in C^\infty(\partial M)$ and $\partial_\tau f=0$.
	This will play a key r\^ole when we will verify that Assumption \ref{Thm: assumption theorem} is valid on ultrastatic spacetimes -- \textit{cf.} Proposition \ref{Prop: self-adjoint relation for parallel-, perp- and f,0- boundary conditions} in Appendix \ref{App: Existence of fundamental solutions on ultrastatic spacetimes}.
\end{remark}

\begin{remark}\label{Rmk: duality of bc under Hodge action}
	It is interesting to observe that different boundary conditions can be related via the action of the Hodge operator.
	In particular, using Equation \eqref{Eq: relations between d,delta,t,n} and \eqref{Eq: f,f' boundary condition}, one can infer that, for any $f,f^\prime\in C^\infty(\partial M)$ it holds that 
	$$\star\Omega^k_{f,f^\prime}(M)=\Omega^{m-k}_{-f^\prime,-f}(M).$$
	At the same time, with reference, to the space of $k$-forms in Definition \ref{Def: Dirichlet, Box-tangential, Box-normal, Robin Box-tangential, Robin Box-normal boundary conditions} it holds
		\begin{align}\label{Eq: duality between Dirichlet-Neumann boundary conditions}
			\star\Omega^k_{\mathrm{D}}(M)=\Omega^{m-k}_{\mathrm{D}}(M)\,,\qquad
			\star\Omega^k_\parallel(M)=\Omega^{m-k}_\perp(M)\,,\qquad
			\star\Omega^k_{f_\parallel}(M)=\Omega^{m-k}_{-f_\perp}(M)\,.
		\end{align}
\end{remark}

\noindent For later convenience we prove the following lemma.
\begin{lemma}\label{Lem: on boundary conditions preserving splitting}
	Let $\sharp\in\lbrace\mathrm{D},\parallel,\perp,f_\parallel,f_\perp\rbrace$, with $f\in C^\infty(\partial M)$.
	The following statements hold true:
	\begin{enumerate}
		\item
		for all $\omega\in\Omega_{\mathrm{sc}}^k(M)\cap\Omega_\sharp^k(M)$ there exists $\omega^+\in\Omega_{\mathrm{spc}}^k(M)\cap\Omega_\sharp^k(M)$ and $\omega^-\in\Omega_{\mathrm{sfc}}^k(M)\cap\Omega_{\sharp}^k(M)$ such that $\omega=\omega^++\omega^-$.
		\item 
		for all $\omega\in\Omega_\sharp^k(M)$ there exists $\omega^+\in\Omega_{\mathrm{pc}}^k(M)\cap\Omega_\sharp^k(M)$ and $\omega^-\in\Omega_{\mathrm{fc}}^k(M)\cap\Omega_{\sharp}^k(M)$ such that $\omega=\omega^++\omega^-$.
	\end{enumerate}
\end{lemma}
\begin{proof}
	We show the result in the first case, the second one can be proved in complete analogy.
	Let $\omega\in\Omega_{\mathrm{sc}}^k(M)\cap\Omega_\sharp^k(M)$.
	Consider $\Sigma_1,\Sigma_2$, two Cauchy surfaces on $M$ -- \textit{cf.} \cite[Def. 3.10]{Ake-Flores-Sanchez-18} -- such that $J^+(\Sigma_1)\subset J^+(\Sigma_2)$.
	Moreover, let $\varphi_+\in \Omega_{\mathrm{pc}}^0(M)$ be such that $\varphi_+|_{J^+(\Sigma_2)}=1$ and $\varphi_+|_{J^-(\Sigma_1)}=0$.
	We define $\varphi_-:=1-\varphi_+\in \Omega_{\mathrm{fc}}^0(M)$.
	Notice that we can always choose $\varphi$ so that, for all $x\in M$, $\varphi(x)$ depends only on the value $\tau(x)$, where $\tau$ is the global time function defined in Theorem \ref{Thm: globally hyperbolic spacetime with time-like boundary}.
	We set $\omega_\pm\doteq\varphi_\pm\omega$ so that $\omega^+\in\Omega_{\mathrm{spc}}^k(M)\cap\Omega_\sharp^k(M)$ while $\omega^-\in\Omega_{\mathrm{sfc}}^k(M)\cap\Omega_{\sharp}^k(M)$.
	This is automatic for $\sharp=\mathrm{D}$ on account of the equalities
	\begin{align*}
		\mathrm{t}\omega^\pm=\varphi_\pm\mathrm{t}\omega=0\,,\qquad
		\mathrm{n}\omega^\pm=\varphi_\pm\mathrm{n}\omega=0\,.
	\end{align*}
	We check that $\omega^\pm\in\Omega^k_\sharp(M)$ for $\sharp=\perp$. The proof for the remaining boundary conditions $\perp,f_\parallel,f_\perp$ follows from a similar computation -- or by duality \textit{cf.} Remark \ref{Rmk: duality of bc under Hodge action}. It holds
	\begin{align*}
		\mathrm{n}\omega_\pm=\varphi_\pm|_{\partial M}\mathrm{n}\omega=0\,,\qquad
		\mathrm{n}\mathrm{d}\omega_\pm
		=\mathrm{n}(\mathrm{d}{\varphi_\pm}\wedge\omega)
		=\partial_\tau{\varphi_\pm}\,\mathrm{n}_{\partial\Sigma_\tau}\mathrm{t}_{\Sigma_\tau}\omega=0\,.
	\end{align*}
	In the last equality
	$\mathrm{t}_{\Sigma_\tau}\colon\Omega^k(M)\to\Omega^k(\Sigma_\tau)$ and $\mathrm{n}_{\partial\Sigma_\tau}\colon\Omega^k(\Sigma_\tau)\to\Omega^{k-1}(\partial\Sigma_\tau)$ are the maps from Definition \ref{Def: tangential and normal component} with $N\equiv\Sigma_\tau\doteq\{\tau\}\times\Sigma$, where $M=\mathbb{R}\times\Sigma$. The last identity follows because the condition $\mathrm{n}\omega=0$ is equivalent to $\mathrm{n}_{\partial\Sigma_\tau}\mathrm{t}_{\Sigma_\tau}\omega=0$ and $\mathrm{n}_{\partial\Sigma_\tau}\mathrm{n}_{\Sigma_\tau}\omega=0$ for all $\tau\in\mathbb{R}$ -- \textit{cf.} Lemma \ref{Lem: equivalence between M-boundary conditions and Sigma-boundary conditions} in Appendix \ref{App: an explicit decomposition}.
\end{proof}

\noindent In the following we shall make a key assumption on the existence of distinguished fundamental solutions for the operator $\Box_\sharp$ for $\sharp\in\lbrace\mathrm{D},\parallel,\perp,f_\parallel,f_\perp\rbrace$.
Subsequently we shall prove that such hypothesis holds true whenever the underlying globally hyperbolic spacetime with timelike boundary is ultrastatic and $f\in C^\infty(\partial \Sigma)$ has definite sign -- \textit{cf.} Appendix \ref{App: Existence of fundamental solutions on ultrastatic spacetimes}.
Recalling both Definition \ref{Def: space of forms} and Definition \ref{Def: Dirichlet, Box-tangential, Box-normal, Robin Box-tangential, Robin Box-normal boundary conditions} we require the following:

\begin{assumption}\label{Thm: assumption theorem}
	For all $f\in C^\infty(\partial M)$ and for all $k$ such that $0\leq k\leq m=\dim M$, there exist advanced $(-)$ and retarded $(+)$ fundamental solutions for the d'Alembert-de Rham wave operator $\Box_\sharp$ where $\sharp\in\lbrace\mathrm{D},\parallel,\perp,f_\parallel,f_\perp\rbrace$. In other words there exist continuous maps $G^\pm_\sharp\colon\Omega_{\mathrm{c}}^k(M)\to\Omega_{\mathrm{sc},\sharp}^k(M)\doteq\Omega_{\mathrm{sc}}^k(M)\cap\Omega_\sharp^k(M)$ such that
	\begin{align}\label{Eq: properties of advanced and retarded propagators}
		\Box\circ G^\pm_\sharp = \operatorname{Id}_{\Omega_{\mathrm{c}}^k(M)}\,,\qquad
		G^\pm_\sharp\circ\Box_{\mathrm{c},\sharp}=\operatorname{Id}_{\Omega_{\mathrm{c},\sharp}^k(M)}\,,\qquad
		{\rm supp}(G^\pm_\sharp\omega)\subseteq J^\pm({\rm supp}(\omega))\,,
	\end{align}
	for all $\omega\in\Omega_{\mathrm{c}}^k(M)$ where
	$\Box_{\mathrm{c},\sharp}$ indicates that the domain of $\Box$ is restricted to $\Omega_{\mathrm{c},\sharp}^k(M)$.
\end{assumption}

\vskip.2cm

\begin{remark}\label{Rmk: on the definition of advanced and retarded propagators}
	Notice that domain of $G^\pm_\sharp$ is not restricted to $\Omega^k_{\mathrm{c},\sharp}(M)$.
	Furthermore the second identity in \eqref{Eq: properties of advanced and retarded propagators} cannot be extended to  $G^\pm_\sharp\circ\Box=\operatorname{Id}_{\Omega_{\mathrm{c}}^k(M)}$ since it would entail $G^\pm_\sharp\Box\omega=\omega$ for all $\omega\in\Omega_{\mathrm{c}}^k(M)$.
	Yet the left hand side also entails that $\omega\in\Omega^k_{\mathrm{c},\sharp}$, which is manifestly a contradiction. 	
\end{remark}

\vskip.2cm

\begin{corollary}\label{Cor: uniqueness}
	Under the same hypotheses of Assumption \ref{Thm: assumption theorem}, if the fundamental solutions $G^\pm_\sharp$ exist, they are unique.
\end{corollary}

\begin{proof}
	Suppose that, beside $G^-_\sharp$, there exists a second map
	$\widetilde{G}^-_\sharp\colon\Omega_{\mathrm{c}}^k(M)\to\Omega_{\mathrm{sc},\sharp}^k(M)$ enjoying the properties of equation \eqref{Eq: properties of advanced and retarded propagators}. Then, for any but fixed $\alpha\in\Omega^k_{\mathrm{c}}(M)$ it holds
	$$(\alpha,G^+_\sharp\beta)=(\Box G^-_\sharp\alpha,G^+_\sharp\beta)=(G^-_\sharp\alpha,\Box G^+_\sharp\beta)=(G^-_\sharp\alpha,\beta),\quad\forall\beta\in\Omega^k_{\mathrm{c}}(M)\,,$$
	where we used both the support properties of the fundamental solutions and Lemma \ref{Lemma: boundary condition} which guarantees that $\Box$ is formally self-adjoint on $\Omega^k_\sharp(M)$.
	Similarly, replacing $G^-_\sharp$ with $\widetilde{G}^-_\sharp$, it holds $(\alpha,G^+_{\sharp}\beta)=(\widetilde{G}^-_\sharp\alpha,\beta)$.
	It descends that $((\widetilde{G}^-_\sharp-G^-_\sharp)\alpha,\beta)=0$, which entails $\widetilde{G}^-_\sharp\alpha=G^-_\sharp\alpha$ being the pairing between $\Omega^k(M)$ and $\Omega^k_{\mathrm{c}}(M)$ separating.
	A similar result holds for the retarded fundamental solution.
\end{proof}

\noindent This corollary can be also read as a consequence of the property that, for all $\omega\in\Omega^k_{\mathrm{c}}(M)$,
$G^\pm_\sharp\omega\in\Omega_{\mathrm{sc},\sharp}^k(M)$ can be characterized as the unique solution to the Cauchy problem
\begin{align}\label{Eq: Cauchy problem for propagators with boundary conditions}
\Box\psi=\omega\,,\qquad
{\rm supp}(\psi)\cap M\setminus J^\pm({\rm supp}(\omega))=\emptyset\,,\qquad
\psi\in\Omega^k_\sharp(M)\,.
\end{align}

\begin{remark}\label{Rmk: Cauchy problem with non-compact source}
	The fundamental solution $G_\sharp^+$ (\textit{resp.} $G_\sharp^-$) can be extended to  $G_\sharp^+\colon\Omega_{\mathrm{pc}}^k(M)\to\Omega_{\mathrm{pc}}^k(M)\cap\Omega^k_\sharp(M)$ (\textit{resp.} {$G_\sharp^-\colon\Omega_{\mathrm{fc}}^k(M)\to\Omega_{\mathrm{fc}}^k(M)\cap\Omega^k_\sharp(M)$}) -- \textit{cf.} \cite[Thm. 3.8]{Baer-15}.
	As a consequence the problem $\Box\psi=\omega$ with $\omega\in\Omega^k(M)$ always admits a solution lying in $\Omega^k_\sharp(M)$.
	As a matter of facts, consider any smooth function $\eta\equiv\eta(\tau)$, where $\tau\in\mathbb{R}$, {\it cf.} equation \eqref{eq:line_element}, such that $\eta(\tau)=1$ for all $\tau>\tau_1$ and $\eta(\tau)=0$ for all $\tau<\tau_0$. Then calling $\omega^+\doteq\eta\omega$ and $\omega^-=(1-\eta)\omega$, it holds $\omega^+\in\Omega_{\mathrm{pc}}^k(M)$ while $\omega^-\in\Omega_{\mathrm{fc}}^k(M)$. Hence  $\psi=G_\sharp^+\omega^++G_\sharp^-\omega^-\in\Omega_\sharp^k(M)$ is {a}
	solution.
\end{remark}

We prove the main result of this section, which characterizes the kernel of $\Box_\sharp$ on the space of smooth $k$-forms with prescribed boundary condition $\sharp\in\lbrace\mathrm{D},\parallel,\perp,f_\parallel,f_\perp\rbrace$.

\begin{proposition}\label{Prop: exact sequence and duality relations}
	Whenever Assumption \ref{Thm: assumption theorem} is fulfilled, then, for all $\sharp\in\lbrace\mathrm{D},\parallel,\perp,f_\parallel,f_\perp\rbrace$, setting $G_\sharp\doteq G_\sharp^+-G_\sharp^-:\Omega^k_{\mathrm{c}}(M)\to\Omega^k_{\mathrm{sc},\sharp}(M)$, the following statements hold true:
	\begin{enumerate}
		\item
		for all $f\in C^\infty(\partial M)$ the following duality relations hold true:
		\begin{align}\label{Eq: duality between propagators}
			\star G^\pm_{\mathrm{D}}=G^\pm_{\mathrm{D}}\star\,,\qquad
			\star G^\pm_\parallel=
			G^\pm_\perp\star\,,\qquad
			\star G^\pm_{f_\parallel}=
			G^\pm_{f_\perp}\star\,.
		\end{align}
		\item 
		for all $\alpha,\beta\in\Omega_{\mathrm{c}}^k(M)$ it holds
		\begin{align}\label{Eq: adjont of propagators}
			(\alpha,G^\pm_\sharp\beta)=(G_\sharp^\mp\alpha,\beta)\,.
		\end{align}
		\item
		the interplay between $G_\sharp$ and $\Box_\sharp$ is encoded in the exact sequence:
		\begin{align}\label{Eq: short exact sequence_aa}
			0\to\Omega^k_{\mathrm{c},\sharp}(M)\stackrel{\Box_\sharp}{\longrightarrow}
			\Omega^k_{\mathrm{c}}(M)\stackrel{G_\sharp}{\longrightarrow}
			\Omega^k_{\mathrm{sc},\sharp}(M)\stackrel{\Box_\sharp}{\longrightarrow}
			\Omega^k_{\mathrm{sc}}(M)\to 0\,,
		\end{align}
	where $\Omega^k_{\mathrm{c},\sharp}(M)\doteq\Omega_{\mathrm{c}}^k(M)\cap\Omega_\sharp^k(M)$.
	\end{enumerate}
\end{proposition}

\begin{proof}
	We prove the different items separately.
	Starting from {\em 1.}, we observe that $\star\Box=\Box\star$.
	This entails that, for all $\alpha\in\Omega^k_{\mathrm{c}}(M)$, 
	\begin{align*}  
		{\Box\star^{-1} G^\pm_\sharp\star\alpha=\star^{-1}\Box G^\pm_\sharp\star\alpha=\alpha\,.}
	\end{align*}
	{
	Remark \ref{Rmk: duality of bc under Hodge action} entails that $\star^{-1}G^\pm_\sharp\star\alpha$ satisfies the necessary boundary conditions, so to apply Corollary \ref{Cor: uniqueness}.
	This yields that $\star^{-1} G^\pm_\sharp\star=G^\pm_{\star\sharp}$, where $G^\pm_{\star\sharp}$ indicates the advanced/retarded propagator for $\star\sharp$-boundary conditions being $\sharp\in\{\mathrm{D},\parallel,\perp,f_\parallel,f_\perp\}$.
	As a consequence Equation \eqref{Eq: duality between propagators} descends. 
	}
	
	\vskip .2cm
	
	\noindent{\em 2.} Equation \eqref{Eq: adjont of propagators} is a consequence of the following chain of identities valid for all $\alpha,\beta\in\Omega_{\mathrm{c}}^k(M)$ 
		\begin{align*}
		(\alpha,G_\sharp^\pm\beta)=
		(\Box_\sharp G_\sharp^\mp\alpha,G_\sharp^\pm\beta)=
		(G_\sharp^\mp\alpha,\Box_\sharp G_\sharp^\pm\beta)=
		(G_\sharp^\mp\alpha,\beta)\,,
		\end{align*}
		where we used both the support properties of the fundamental solutions and Lemma \ref{Lemma: boundary condition}.
	
	\vskip .2cm
	
	\noindent{\em 3.} The exactness of the series is proven using the properties already established for the fundamental solutions $G^\pm_\sharp$.
	The left exactness of the sequence is a consequence of the second identity in equation \eqref{Eq: properties of advanced and retarded propagators} which ensures that $\Box_\sharp\alpha=0$, $\alpha\in\Omega^k_{\mathrm{c},\sharp}(M)$, entails $\alpha=G_\sharp^+\Box_\sharp\alpha=0$.
	In order to prove that $\ker G_\sharp=\Box_\sharp\Omega^k_{\mathrm{c},\sharp}$, we first observe that $G_\sharp\Box_\sharp\Omega^k_{\mathrm{c},\sharp}(M)=\{0\}$ on account of equation \eqref{Eq: properties of advanced and retarded propagators}.
	Moreover, if $\beta\in\Omega^k_{\mathrm{c}}(M)$ is such that $G_\sharp\beta=0$, then $G^+_\sharp\beta=G^-_\sharp\beta$.
	Hence, in view of the support properties of the fundamental solutions $G^+_\sharp\beta\in\Omega^k_{\mathrm{c},\sharp}(M)$ and $\beta=\Box_\sharp G^+_\sharp\beta$.
	Subsequently we need to verify that {$\ker\Box_\sharp=G_\sharp\Omega^k_{\mathrm{c}}(M)$}.
	Once more $\Box_\sharp G_\sharp\Omega^k_{\mathrm{c}}(M)=\{0\}$ follows from equation \eqref{Eq: properties of advanced and retarded propagators}.
	Conversely, let $\omega\in\Omega^k_{\mathrm{sc},\sharp}(M)$ be such that $\Box_\sharp\omega=0$.
	On account of Lemma \ref{Lem: on boundary conditions preserving splitting} we can split $\omega=\omega^++\omega^-$ where $\omega^+\in\Omega^k_{\mathrm{spc},\sharp}(M)$.
	Then $\Box_\sharp\omega^+=-\Box_\sharp\omega^-\in\Omega^k_{\mathrm{c}}(M)$ and
	\begin{align*}
		G_\sharp\Box_\sharp\omega^+=
		G_\sharp^+\Box_\sharp\omega^++
		G_\sharp^-\Box_\sharp\omega^-=\omega\,.
	\end{align*}
	To conclude we need to establish the right exactness of the sequence.
	Consider any $\alpha\in\Omega^k_{\mathrm{sc}}(M)$ and the equation $\Box_\sharp\omega=\alpha$.
	Consider the function $\eta(\tau)$ as in Remark \ref{Rmk: Cauchy problem with non-compact source} and let $\omega\doteq G^+_\sharp(\eta\alpha)+G^-_\sharp((1-\eta)\alpha)$.
	In view of Remark \ref{Rmk: Cauchy problem with non-compact source} and of the support properties of the fundamental solutions, $\omega\in\Omega^k_{\mathrm{sc},\sharp}(M)$ and $\Box_\sharp\omega=\alpha$.
\end{proof}

\begin{remark}\label{Rmk: extension of short exact sequence}
	Following the same reasoning as in \cite{Baer-15} together with minor adaptations of the proofs of \cite{Dappiaggi-Drago-Ferreira-19}, one may extend $G_\sharp$ to an operator $G_\sharp\colon\Omega_{\mathrm{tc}}^k(M)\to\Omega^k_\sharp(M)$ for all $\sharp\in\{\mathrm{D},\parallel,\perp,f_\parallel,f_\perp\}$.
	As a consequence the exact sequence of Proposition \ref{Prop: exact sequence and duality relations} generalizes as
	\begin{align}\label{Eq: short exact sequence for timelike k-forms}
		0\to\Omega_{\mathrm{tc},\sharp}^k(M)\stackrel{\Box_\sharp}{\longrightarrow}
		\Omega_{\mathrm{tc}}^k(M)\stackrel{G_\sharp}{\longrightarrow}
		\Omega_\sharp^k(M)\stackrel{\Box_\sharp}{\longrightarrow}
		\Omega^k(M)\to 0\,.
	\end{align}
\end{remark}

\begin{remark}\label{Rmk: compactly supported solutions of the wave operator}
	Proposition \ref{Prop: exact sequence and duality relations} and Remark \ref{Rmk: extension of short exact sequence} ensure that $\ker_{\mathrm{c}}\Box_\sharp\subseteq\ker_{\mathrm{tc}}\Box_\sharp=\lbrace 0\rbrace$. In other words, there are no timelike compact solutions to the equation $\Box\omega=0$ with $\sharp$-boundary conditions.
	More generally it can be shown that $\ker_{\mathrm{c}}\Box\subseteq\ker_{\mathrm{tc}}\Box=\lbrace 0 \rbrace$, namely there are no timelike compact solutions regardless of the boundary condition.
	This follows by standard arguments using a suitable energy functional defined on the solution space -- \textit{cf.} \cite[Thm. 30]{Dappiaggi-Drago-Ferreira-19} for the proof for $k=0$.
\end{remark}

\noindent
In view of the applications to the Maxwell operator, it is worth focusing specifically on the boundary conditions $\perp$, $\parallel$ individuated in Definition \ref{Def: Dirichlet, Box-tangential, Box-normal, Robin Box-tangential, Robin Box-normal boundary conditions} since it is possible to prove a useful relation between the associated propagators and the operators $\mathrm{d}$,$\delta$.

\begin{lemma}\label{Lem: relations between delta,d and advanced-retarded propagators}
	Under the hypotheses of Assumption \ref{Thm: assumption theorem} it holds that
	\begin{align}
		\label{Eq: relations between delta,d and Box-tangential advanced-retarded propagators}
		G_\parallel^\pm\circ\mathrm{d}
		&=\mathrm{d}\circ G_\parallel^\pm
		\qquad\mathrm{on}\;\Omega_{\mathrm{t}}^k(M)\cap\Omega^k_{\mathrm{pc/fc}}(M)\,,\qquad
		G_\parallel^\pm\circ\delta=\delta\circ G_\parallel^\pm
		\qquad\mathrm{on}\;\Omega_{\mathrm{pc/fc}}^k(M)\,,\\
		\label{Eq: relations between delta,d and Box-normal advanced-retarded propagators}
		G_\perp^\pm\circ\delta
		&=\delta\circ G^\pm_{\perp}
		\qquad\mathrm{on}\;\Omega^k_{\mathrm{n}}(M)\cap\Omega^k_{\mathrm{pc/fc}}(M)\,,\qquad
		G_\perp^\pm\circ\mathrm{d}=\mathrm{d}\circ G_\perp^\pm
		\qquad\mathrm{on}\;\Omega_{\mathrm{pc/fc}}^k(M)\,.
	\end{align}
\end{lemma}
\begin{proof}
	From equation \eqref{Eq: duality between propagators} it follows that equations (\ref{Eq: relations between delta,d and Box-tangential advanced-retarded propagators}-\ref{Eq: relations between delta,d and Box-normal advanced-retarded propagators}) are dual to each other via the Hodge operator.
	Hence we shall only focus on equation \eqref{Eq: relations between delta,d and Box-tangential advanced-retarded propagators}.
	
	For every $\alpha\in\Omega^k_{\mathrm{c}}(M)\cap\Omega^k_{\mathrm{t}}(M)$, $G^\pm_\parallel \mathrm{d}\alpha$ and $\mathrm{d}G^\pm_\parallel\alpha$ lie both in $\Omega^k_\parallel(M)$.
	In particular, using equation \eqref{Eq: relations-bulk-to-boundary}, $\mathrm{t}\delta \mathrm{d} G^\pm_\parallel\alpha=\mathrm{t}(\Box_\parallel-\mathrm{d}\delta)G^\pm_\parallel(\alpha)=\mathrm{t}\alpha=0$ while the second boundary condition is automatically satisfied since $\mathrm{td}G^\pm_\parallel=\mathrm{dt}G^\pm_\parallel=0$.
	Hence, considering $\beta=G^\pm_\parallel\mathrm{d}\alpha-\mathrm{d}G^\pm_\parallel\alpha$, it holds that $\Box\beta=0$ and $\beta\in\Omega^k_\parallel\cap\Omega^k_{\mathrm{pc/fc}}(M)$.
	In view of Remark \ref{Rmk: Cauchy problem with non-compact source}, this entails $\beta=0$.
\end{proof}

We conclude this section with a corollary to Lemma \ref{Lem: relations between delta,d and advanced-retarded propagators} which shows that, when considering the difference between the advanced and the retarded fundamental solutions, the support restrictions present in equations (\ref{Eq: relations between delta,d and Box-tangential advanced-retarded propagators}-\ref{Eq: relations between delta,d and Box-normal advanced-retarded propagators}) disappear.

\begin{corollary}\label{Cor: G commutes with d, delta}
	Under the hypotheses of Assumption \ref{Thm: assumption theorem} it holds that:
	\begin{enumerate}[(i)]
		\item
		for all $\alpha\in\Omega^k_{\mathrm{tc}}(M)$ there exists $\beta_\parallel\in\Omega^{k+1}_{\mathrm{tc}}(M)$ such that $\mathrm{t}\beta_\parallel=0$, $\mathrm{t}\delta\beta_\parallel=\mathrm{t}\alpha$ and
		\begin{align}\label{Eq: relation between d,delta with parallel-propagator}
			\delta G_\parallel\alpha
			=G_\parallel\delta\alpha\,,\qquad
			\mathrm{d}G_\parallel\alpha
			=G_\parallel(\mathrm{d}\alpha-\Box\beta_\parallel)\,.
		\end{align}
		\item
		for all $\alpha\in\Omega^k_{\mathrm{tc}}(M)$ there exists $\beta_\perp\in\Omega^{k-1}_{\mathrm{tc}}(M)$ such that $\mathrm{n}\beta_\perp=0$, $\mathrm{nd}\beta_\perp=\mathrm{n}\alpha$ and
		\begin{align}\label{Eq: relation between d,delta with perp-propagator}
			\delta G_\perp\alpha
			=G_\perp(\delta\alpha
			-\Box\beta_\perp)\,,\qquad
			\mathrm{d}G_\perp\alpha
			=G_\perp\mathrm{d}\alpha\,.
		\end{align}
	\end{enumerate}
\end{corollary}

\begin{proof}
	 As starting point, we observe that the existence of $\beta_\parallel,\beta_\perp$ is guaranteed by Remark \ref{Rmk: surjectivity of t,n,tdelta,nd}.
	\paragraph{Proof for $\parallel$ boundary conditions.}
	On account of Lemma \ref{Lem: relations between delta,d and advanced-retarded propagators}, it holds
	\begin{align}\label{Eq: relation between d,delta with parallel-propagator on domains}
		\delta G_\parallel
		=G_\parallel\delta
		\quad\mathrm{on}\;\Omega^k_{\mathrm{tc}}(M)\,,\qquad
		\mathrm{d}G_\parallel
		=G_\parallel\mathrm{d}
		\quad\mathrm{on}\;\Omega^k_{\mathrm{tc},\mathrm{t}}(M)\,.
	\end{align}
In addition
	\begin{align*}
		\mathrm{d} G_\parallel\alpha
		=\mathrm{d} G_\parallel(\alpha-\delta\beta_\parallel)
		+\mathrm{d} G_\parallel\delta\beta_\parallel
		=G_\parallel\mathrm{d}\alpha
		-G_\parallel\mathrm{d}\delta\beta_\parallel
		+\mathrm{d}G_\parallel\delta\beta_\parallel\,,
	\end{align*}
	where in the second equality we used Equation \eqref{Eq: relation between d,delta with parallel-propagator on domains} together with the boundary condition $\mathrm{t}\delta\beta_\parallel=\mathrm{t}\alpha$.
	Due to both Equation \eqref{Eq: relation between d,delta with parallel-propagator on domains} and the condition $\mathrm{t}\beta_\parallel=0$ it holds
	\begin{align*}
		\mathrm{d}G_\parallel\delta\beta_\parallel
		=\mathrm{d}\delta G_\parallel\beta_\parallel
		=\Box G_\parallel\beta_\parallel
		-\delta\mathrm{d}G_\parallel\beta_\parallel
		=-G_\parallel\delta\mathrm{d}\beta_\parallel
		=-G_\parallel\Box\beta_\parallel
		+G_\parallel\mathrm{d}\delta\beta_\parallel\,.
	\end{align*}
	Putting together these data we find
	\begin{align*}
		\mathrm{d} G_\parallel\alpha
		=G_\parallel\mathrm{d}\alpha
		-G_\parallel\mathrm{d}\delta\beta_\parallel
		+\mathrm{d}G_\parallel\delta\beta_\parallel
		=G_\parallel\mathrm{d}\alpha
		-G_\parallel\Box\beta_\parallel\,,
	\end{align*}
	as claimed.
	\paragraph{Proof for $\perp$ boundary conditions.}
	The proof is similar to the previous one mutatis mutandis.
	In particular Lemma \ref{Lem: relations between delta,d and advanced-retarded propagators} entails
	\begin{align}\label{Eq: relation between d,delta with perp-propagator on domains}
		\delta G_\perp
		=G_\perp\delta
		\quad\mathrm{on}\;\Omega^k_{\mathrm{tc},\mathrm{n}}(M)\,,\qquad
		\mathrm{d}G_\perp
		=G_\perp\mathrm{d}
		\quad\mathrm{on}\;\Omega^k_{\mathrm{tc}}(M)\,,
	\end{align}
	At the same time
	\begin{align*}
		\delta G_\perp\alpha
		=\delta G_\perp(\alpha-\mathrm{d}\beta_\perp)
		+\delta G_\perp\mathrm{d}\beta_\perp
		=G_\perp\delta\alpha
		-G_\perp\delta\mathrm{d}\beta_\perp
		+\delta G_\perp\mathrm{d}\beta_\perp
		=G_\perp(\delta\alpha-\Box\beta_\perp)\,,
	\end{align*}
	where we have used the boundary conditions of $\beta_\perp$ together with
	\begin{align*}
		\delta G_\perp\mathrm{d}\beta_\perp
		=\delta\mathrm{d}G_\perp\beta_\perp
		=\Box G_\perp\beta_\perp
		-\mathrm{d}\delta G_\perp\beta_\perp
		=-G_\perp\mathrm{d}\delta\beta_\perp
		=-G_\perp\Box\beta_\perp
		+G_\perp\delta\mathrm{d}\beta_\perp\,.
	\end{align*}
\end{proof}

\begin{remark}
	Notice that Equations \eqref{Eq: relation between d,delta with parallel-propagator}-\eqref{Eq: relation between d,delta with perp-propagator} do not depend on the particular choice of $\beta_\parallel,\beta_\perp$.
	In particular, let assume $\hat{\beta}_\perp\in\Omega^{k-1}_{\mathrm{tc}}(M)$ is another $(k-1)$-form such that $\mathrm{n}\hat{\beta}_\perp=0$ while $\mathrm{nd}\hat{\beta}_\perp=\mathrm{n}\alpha$.
	If follows that $(\beta_\perp-\hat{\beta}_\perp)\in\Omega^{k-1}_{\mathrm{tc},\perp}(M)$ and therefore $G_\perp\Box(\beta_\perp-\hat{\beta}_\perp)=0$ on account of Proposition \ref{Prop: exact sequence and duality relations}.
\end{remark}

\subsection{On the Maxwell operator}\label{Sec: on the Maxwell operator}

In this section we focus our attention on the Maxwell operator $\delta\mathrm{d}:\Omega^k(M)\to\Omega^k(M)$ studying its kernel in connection both to the D'Alembert - de Rham wave operator $\Box$ and to the identification of suitable boundary conditions.
We shall keep the assumption that $(M,g)$ is a globally hyperbolic spacetime with timelike boundary of dimension $\dim M = m\geq 2$ -- \textit{cf.} Theorem \ref{Thm: globally hyperbolic spacetime with time-like boundary}.
Notice that, if $k=m$, then the Maxwell operator becomes trivial, while, if $k=0$, is coincides with the D'Alembert - de Rham operator $\Box$. 
Hence this case falls in the one studied in the preceding section and in \cite{Dappiaggi-Drago-Ferreira-19}. Therefore, unless stated otherwise, henceforth we shall consider only $0<k<m=\dim M$.

In complete analogy to the analysis of $\Box$, we observe that, for any pair $\alpha,\beta\in\Omega^k(M)$ such that $\textrm{supp}(\alpha)\cap\textrm{supp}(\beta)$ is compact, the following Green's formula holds true:
\begin{align}\label{Eq: boundary terms for delta d operator}
	(\delta\mathrm{d}\alpha,\beta)-(\alpha,\delta\mathrm{d}\beta)=
	(\mathrm{t}\alpha,\mathrm{n}\mathrm{d}\beta)_\partial
	-(\mathrm{n}\mathrm{d}\alpha,\mathrm{t}\beta)_\partial\,.
\end{align}

In the same spirit of Lemma \ref{Lemma: boundary condition}, the operator $\delta\mathrm{d}$ becomes formally self-adjoint if we restrict its domain to 
\begin{equation}\label{Eq: Robin-like bc for Maxwell operator}
	\Omega^k_f(M)\doteq\{\omega\in\Omega^k(M)\;|\;\mathrm{nd}\omega=f\mathrm{t}\omega\}\,,
\end{equation}
where $f\in C^\infty(\partial M)$ is arbitrary but fixed.
In what follows we will consider two particular boundary conditions which are directly related to the $\Box$-tangential and to the $\Box$-normal boundary conditions for the D'Alembert - de Rham operator -- \textit{cf.} Definition \ref{Def: Dirichlet, Box-tangential, Box-normal, Robin Box-tangential, Robin Box-normal boundary conditions}.

The discussion of the general case is related to the Robin $\Box$-tangential / Robin $\Box$-normal boundary conditions. However, in these cases, it is not clear whether a generalization of Lemma \ref{Lem: relations between delta,d and advanced-retarded propagators} holds true. This is an important obstruction to adapt our analysis to these cases.

\begin{Definition}\label{Def: g-Dirichelet, g-Neumann boundary conditions and solution spaces}
	Let $(M,g)$ be a globally hyperbolic spacetime with timelike boundary and let $0<k<\dim M$.
	We call
	\begin{enumerate}
		\item space of $k$-forms with $\delta\mathrm{d}$-tangential boundary condition, $\Omega_{\mathrm{t}}^k(M)$ as in equation \eqref{Eq: k-forms with vanishing tangential or normal component} with $N=\partial M$.
		\item space of $k$-forms with $\delta\mathrm{d}$-normal boundary condition
		\begin{align}\label{Eq: deltad-normal bc}
			\Omega_{\mathrm{nd}}^k(M)\doteq\lbrace\omega\in\Omega^k(M)|\;\mathrm{n}\mathrm{d}\omega=0\rbrace\,.
		\end{align}
	\end{enumerate} 
\end{Definition}

\begin{remark}
	It is worth observing that the the $\delta d$-normal boundary condition appears to play a distinguished r\^ole in relation to standard electromagnetism, seen as a theory for $1$-forms.
	More precisely, if we consider a globally hyperbolic spacetime $(M,g)$ with timelike boundary, the underlying action reads 
	$$S[A]=\frac{1}{2}(\mathrm{d}A,\mathrm{d}A)\,.$$
	By considering an arbitrary variation with respect to $\alpha\in\Omega^1_{\mathrm{c}}(M)$, it descends that
	$$\frac{\mathrm{d}}{\mathrm{d}\lambda} S[A+\lambda\alpha]\big|_{\lambda=0}=(\alpha,\delta \mathrm{d}A)+(\mathrm{t}\alpha,\mathrm{nd}A)_\partial\,.$$
	The arbitrariness of $\alpha$ leads to the equation of motion $\delta \mathrm{d}A=0$, together with the boundary condition $\mathrm{nd}A=0$.
	This indication of the preferred r\^ole of the $\delta \mathrm{d}$-normal boundary condition will be strengthened by the following discussion -- \textit{cf.} Remark \ref{Rmk: on bc-dependent gauge groups}.
\end{remark}

In the following our first goal is to characterize the kernel of the Maxwell operator with a prescribed boundary condition, {\it cf.} Equation \eqref{Eq: Robin-like bc for Maxwell operator}. To this end we need to focus on the {\em gauge invariance} of the underlying theory. In the case in hand this translates in the following characterization.

\vskip .2cm

\begin{Definition}\label{Def: configuration space with deltad-tangential and deltad-normal bc}
	Let $(M,g)$ be a globally hyperbolic spacetime with timelike boundary and let $\delta\mathrm{d}$ be the Maxwell operator acting on $\Omega^k(M)$, $0< k<\dim M$.
	We say that 
	\begin{enumerate}
		\item
		$A\in\Omega^k_{\mathrm{t}}(M)$, is {\em gauge equivalent} to $A^\prime\in\Omega^k_{\mathrm{t}}(M)$
		if $A-A^\prime\in \mathrm{d}\Omega^{k-1}_{\mathrm{t}}(M)$, namely if there exists $\chi\in \Omega^{k-1}_{\mathrm{t}}(M)$ such that $A^\prime=A+\mathrm{d}\chi$.
		The space of solutions with $\delta\mathrm{d}$-tangential boundary conditions is denoted by
		\begin{align}\label{Eq: gauge sol with deltad-tangential bc}
			\operatorname{Sol}_{\mathrm{t}}(M)\doteq
			\frac{\lbrace A\in\Omega^k(M)|\;\delta\mathrm{d}A=0\,,\mathrm{t}A=0\rbrace}{\mathrm{d}\Omega^{k-1}_{\mathrm{t}}(M)}\,.
		\end{align}
		\item
		$A\in\Omega^k_{\mathrm{nd}}(M)$, is {\em gauge equivalent} to $A^\prime\in\Omega^k_{\mathrm{nd}}(M)$ if there exists $\chi\in \Omega^{k-1}(M)$ such that $A^\prime=A+\mathrm{d}\chi$.
		The space of solutions with $\delta\mathrm{d}$-normal boundary conditions is denoted by
		\begin{align}\label{Eq: gauge sol with nd-boundary condition}
			\operatorname{Sol}_{\mathrm{nd}}(M)\doteq
			\frac{\lbrace A\in\Omega^k(M)|\;\delta\mathrm{d}A=0\,,\mathrm{nd}A=0\rbrace}{\mathrm{d}\Omega^{k-1}(M)}\,.
		\end{align}
		Similarly the space of spacelike supported solutions with $\delta\mathrm{d}$-tangential (resp. $\delta\mathrm{d}$-normal) boundary conditions are 
		\begin{align}\label{Eq: sc-gauge sol with deltad-tangential or deltad-normal bc}
			\operatorname{Sol}_{\mathrm{t}}^{\mathrm{sc}}(M)\doteq
			\frac{\lbrace A\in\Omega^k_{\mathrm{sc}}(M)|\;\delta\mathrm{d}A=0\,,\mathrm{t}A=0\rbrace}{\mathrm{d}\Omega^{k-1}_{\mathrm{t,sc}}(M)}\,,\quad
			\operatorname{Sol}_{\mathrm{nd}}^{\mathrm{sc}}(M)\doteq
			\frac{\lbrace A\in\Omega^k_{\mathrm{sc}}(M)|\;\delta\mathrm{d}A=0\,,\mathrm{nd}A=0\rbrace}{\mathrm{d}\Omega^{k-1}_{\mathrm{sc}}(M)}\,.
		\end{align}
	\end{enumerate}	
\end{Definition}

\begin{remark}\label{Rmk: on bc-dependent gauge groups}
	Notice that in Definition \ref{Def: configuration space with deltad-tangential and deltad-normal bc} we have employed two different notions of gauge equivalence in the construction of $\operatorname{Sol}_{\mathrm{t}}(M)$ and of $\operatorname{Sol}_{\mathrm{nd}}(M)$, which are related to the different choices of boundary conditions.
	It is worth observing that, at the level of solution space, the boundary condition $\mathrm{nd}\omega=0$ involves a constraint on a quantity, {\it e.g.} the Faraday tensor when working with $k=1$, which is gauge invariant with respect to the standard gauge group of Maxwell theory on a globally hyperbolic spacetime without boundary. Therefore this reverberates in the lack of any necessity to restrict the underlying gauge group in the case in hand.
	For this reason such scenario is certainly distinguished.
	As a matter of fact, when working with $\operatorname{Sol}_{\mathrm{t}}(M)$, the boundary condition does not involve a quantity which is gauge invariant under the action of the standard gauge group of Maxwell theory on a globally hyperbolic spacetime without boundary. Hence, in this case, one must introduce a reduced gauge group. The latter can be chosen in different ways and, to avoid such quandary, one should resort to a more geometrical formulation of Maxwell's equations, namely as originating from a theory for the connections of a principal $U(1)$-bundle over the underlying globally hyperbolic spacetime with timelike boundary, {\it cf.} \cite{Benini:2013ita,Benini:2013tra} for the case with empty boundary.
	Since this analysis would require a whole paper on its own we postpone it to future work.
\end{remark}

The following propositions discuss the existence of a representative fulfilling the Lorenz gauge condition of an equivalence classes $[A]\in\operatorname{Sol}_{\mathrm{t}}(M)$ (\textit{resp.} $[A]\in\operatorname{Sol}_{\mathrm{nd}}(M)$) -- \textit{cf.} \cite[Lem. 7.2]{Benini-16}.
In addition we provide a connection between $\delta\mathrm{d}$-tangential (\textit{resp.} $\delta\mathrm{d}$-normal) boundary conditions with $\Box$-tangential (\textit{resp.} $\Box$-normal) boundary conditions. Recalling Definition \ref{Def: Dirichlet, Box-tangential, Box-normal, Robin Box-tangential, Robin Box-normal boundary conditions} of the $\Box$-tangential boundary condition, the following holds true.

\begin{proposition}\label{Prop: Lorenz gauge for deltad-tangential bc}
	Let $(M,g)$ be a globally hyperbolic spacetime with timelike boundary.
	Then for all $[A]\in\operatorname{Sol}_{\mathrm{t}}(M)$ there exists a representative $A^\prime\in [A]$ such that
	\begin{align}\label{Eq: system of sins}
		\Box_\parallel A^\prime=0\,,\qquad
		\delta A^\prime=0\,.
	\end{align}
	Moreover, up to gauge transformation we have $A^\prime=G_\parallel\alpha$ with $\alpha\in\Omega^k_{\mathrm{tc},\delta}(M)$.
	Finally, the same result holds true for $[A]\in\operatorname{Sol}_{\mathrm{t}}^{\mathrm{sc}}(M)$ -- in particular in this case $A^\prime=G_\perp\alpha$ for $\alpha\in\Omega^k_{\mathrm{c},\delta}(M)$.
\end{proposition}

\begin{proof}
	We focus only on the first statement, the proof of the second one being similar.
	Let $A\in[A]\in\operatorname{Sol}_{\mathrm{t}}(M)$, that is, $A\in\Omega^k(M)$, $\delta\mathrm{d}A=0$ and $\mathrm{t}A=0$.
	Consider any $\chi\in\Omega^{k-1}_{\mathrm{t}}(M)$ such that 
	\begin{equation}\label{Eq: gauge fixing}
		\Box\chi=-\delta A,\qquad\delta\chi=0,\qquad \mathrm{t}\chi=0\,.
	\end{equation}
	In view of Assumption \ref{Thm: assumption theorem} and of Remark \ref{Rmk: Cauchy problem with non-compact source}, we can fix $\chi=-\sum_\pm G_\parallel^\pm\delta A^\pm$, where $A^\pm$ is defined as in Remark \ref{Rmk: Cauchy problem with non-compact source}.
	Per definition of $G_\parallel^\pm$, $\mathrm{t}\chi=0$ while, on account of Lemma \ref{Lem: relations between delta,d and advanced-retarded propagators}, $\delta\chi=-\sum_\pm\delta G_\parallel^\pm\delta A^\pm=0$.
	Hence $A^\prime$ is gauge equivalent to $A$ as per Definition \ref{Def: configuration space with deltad-tangential and deltad-normal bc}.
	
	Proposition \ref{Prop: exact sequence and duality relations} and Remark \ref{Rmk: extension of short exact sequence} entail that there exists $\alpha\in\Omega^k_{\mathrm{tc}}(M)$ such that $A^\prime=G_\parallel\alpha$.
	Equation \eqref{Eq: relation between d,delta with parallel-propagator on domains} implies that
	\begin{align*}
		0=\delta A^\prime
		=\delta G_\parallel\alpha
		=G_\parallel\delta\alpha\,,
	\end{align*}
	that is, $\delta\alpha\in\ker G_\parallel$.
	This implies that there exists $\beta\in\Omega^{k-1}_{\mathrm{tc},\parallel}(M)$ such that $\delta\alpha=\Box_\parallel\beta$.
	It follows that
	\begin{align*}
		0=\delta^2\alpha
		=\delta\Box_\parallel\beta
		=\Box\delta\beta\,,
	\end{align*}
	which entails, on account of Remark \ref{Rmk: compactly supported solutions of the wave operator}, $\delta\beta=0$.
	It follows that $\delta\alpha=\Box_\parallel\beta=\delta\mathrm{d}\beta$ and therefore
	\begin{align*}
		[A^\prime]
		=[G_\parallel\alpha-\mathrm{d}G_\parallel\beta]
		=[G_\parallel(\alpha-\mathrm{d}\beta)]\,,
	\end{align*}
	where we used Equation \eqref{Eq: relation between d,delta with parallel-propagator on domains}.
	Since $\alpha-\mathrm{d}\beta\in\Omega^k_{\mathrm{tc},\delta}(M)$ we have obtained the sought result.
\end{proof}

\noindent
The proof of the analogous result for $\Omega^k_{\mathrm{nd}}(M)$ is slightly different and, thus, we discuss it separately. Recalling Definition \ref{Def: Dirichlet, Box-tangential, Box-normal, Robin Box-tangential, Robin Box-normal boundary conditions} of the $\Box$-normal boundary conditions, the following statement holds true.

\begin{proposition}\label{Prop: Lorenz gauge for deltad-normal bc}
	Let $(M,g)$ be a globally hyperbolic spacetime with timelike boundary. Then for all $[A]\in\operatorname{Sol}_{\mathrm{nd}}(M)$ there exists a representative $A^\prime\in[A]$ such that
	\begin{align}\label{Eq: system of sins 2}
		\Box_\perp A^\prime=0\,,\qquad
		\delta A^\prime=0\,.
	\end{align}
	Moreover, up to gauge transformation we have $A^\prime=G_\perp\alpha$ with $\alpha\in\Omega^k_{\mathrm{tc,n},\delta}(M)$.
	Finally, the same result holds true for $[A]\in\operatorname{Sol}_{\mathrm{nd}}^{\mathrm{sc}}(M)$ -- in particular in this case $A^\prime=G_\perp\alpha$ for $\alpha\in\Omega^k_{\mathrm{c,n},\delta}(M)$.
\end{proposition}

\begin{proof}
	As in the previous proposition, we can focus only on the first point, the second following suit.
	Let $A$ be a representative of $[A]\in\operatorname{Sol}_{\mathrm{nd}}(M)$. Hence $A\in\Omega^k(M)$ so that $\delta\mathrm{d}A=0$ and $\mathrm{nd}A=0$.
	Consider first $\chi_0\in\Omega^{k-1}(M)$ such that $\mathrm{nd}\chi_0=-\mathrm{n}A$.
	The existence is guaranteed since the map $\mathrm{nd}$ is surjective -- \textit{cf.} Remark \ref{Rmk: surjectivity of t,n,tdelta,nd}.
	As a consequence we can exploit the residual gauge freedom to select $\chi_1\in\Omega^{k-1}(M)$ such that
	\begin{align}\label{Eq: gauge fixing partial}
		\Box \chi_1=-\delta \widetilde{A}\,,\qquad
	\delta\chi_1=0\,,\qquad
	\mathrm{nd}\chi_1 =0\,\qquad
	\mathrm{n}\chi_1=0\,,
	\end{align}
	where $\widetilde{A}=A+\mathrm{d}\chi_0$.
	Let $\eta\equiv\eta(\tau)$ be a smooth function such that $\eta=0$ if $\tau<\tau_0$ while $\eta=1$ if $\tau>\tau_1$, \textit{cf.} Remark \ref{Rmk: Cauchy problem with non-compact source}.
	Since $\mathrm{n}\widetilde{A}=0$ we can fine tune $\eta$ in such a way that both $\widetilde{A}^+\doteq\eta\widetilde{A}$ and $\widetilde{A}^-\doteq(1-\eta)\widetilde{A}$ satisfy $\mathrm{n}\widetilde{A}^\pm=0$.
	Equation \eqref{Eq: relations-bulk-to-boundary} entails that $\mathrm{n}\delta A^\pm=-\delta\mathrm{n}A^\pm=0$.
	Hence we can apply Lemma \ref{Lem: relations between delta,d and advanced-retarded propagators} setting $\chi_1=-\sum_\pm G^\pm_\perp\delta\widetilde{A}^+$.
	Calling $A^\prime=A+\mathrm{d}(\chi_0+\chi_1)$ we obtained the desired result.
	
	Proposition \ref{Prop: exact sequence and duality relations} and Remark \ref{Rmk: extension of short exact sequence} imply that there exists $\alpha\in\Omega^k_{\mathrm{tc}}(M)$ such that $A^\prime=G_\perp\alpha$.
	On account of Corollary \ref{Cor: G commutes with d, delta} there exists $\beta_\perp\in\Omega^{k-1}_{\mathrm{tc}}(M)$ such that
	\begin{align*}
		0=\delta A'
		=\delta G_\perp\alpha
		=G_\perp(\delta\alpha-\Box\beta_\perp)\,,
	\end{align*}
	together with $\mathrm{nd}\beta_\perp=\mathrm{n}\alpha$ and $\mathrm{n}\beta_\perp=0$.
	It follows that $\delta\alpha-\Box\beta_\perp=\Box_\perp\eta$ for $\eta\in\Omega^{k-1}_{\mathrm{tc},\perp}(M)$.
	Application of $\delta$ to the above identity leads to
	\begin{align*}
		0=
		\delta(\delta\alpha-\Box\beta_\perp-\Box_\perp\eta)
		=-\Box(\delta\beta_\perp+\delta\eta)\,.
	\end{align*}
	Remark \ref{Rmk: compactly supported solutions of the wave operator} entails that $\delta\beta_\perp+\delta\eta=0$, that is,
	\begin{align*}
		\delta\alpha
		=\Box_\perp\eta
		+\Box\beta_\perp
		=\delta\mathrm{d}(\eta+\beta_\perp)\,.
	\end{align*}
	It follows that
	\begin{align*}
		[A^\prime]
		=[G_\perp\alpha-\mathrm{d}G_\perp(\eta+\beta_\perp)]
		=[G_\perp(\alpha-\mathrm{d}(\eta+\beta_\perp))]\,.
	\end{align*}
	where we used equation \eqref{Eq: relation between d,delta with perp-propagator on domains}.
	Since $\alpha-\mathrm{d}(\eta+\beta_\perp)\in\Omega^k_{\mathrm{tc,n},\delta}(M)$ the proof is complete.
\end{proof}

\begin{remark}\label{Rmk: isomorphism for configuration space}
	A direct inspection of \eqref{Eq: gauge fixing} and of \eqref{Eq: system of sins 2} unveils that choosing a solution to these equations does not fix completely the gauge and a residual freedom is left.
	This amount either to 
	\begin{align*}
		\mathcal{G}_{\mathrm{t}}(M)\doteq
		\{\chi\in\Omega^{k-1}(M)\;|\;\delta\mathrm{d}\chi=0\,,\;\mathrm{t}\chi=0\}\,,
	\end{align*}
	or, in the case of a $\delta\mathrm{d}$-normal boundary condition, to
	\begin{align*}
		\mathcal{G}_{\mathrm{nd}}(M)\doteq
		\{\chi\in\Omega^{k-1}(M)\;|\;\delta\mathrm{d}\chi=0,\;\mathrm{n}\chi=0\;,\;\mathrm{n}\mathrm{d}\chi=0\}\,.
	\end{align*}
	Observe that, in the definition of $\mathcal{G}_{\mathrm{nd}}(M)$, we require $\chi$ to be in the kernel of $\delta\mathrm{d}$.
	Nonetheless, since the actual reduced gauge group is $\mathrm{d}\mathcal{G}_{\mathrm{nd}}(M)$ we can work with $\chi_0\in\Omega^{k-1}(M)$ such that $\Box\chi_0=0$.
	As a matter of fact for all $\chi\in\mathcal{G}_{\mathrm{nd}}$ we can set $\chi_0\doteq\chi+\mathrm{d}\lambda$ where $\lambda\in\Omega^{k-2}(M)$ is such that $\Box\lambda=-\delta\chi$ and $\mathrm{n}\lambda=\mathrm{nd}\lambda=0$ -- \textit{cf.} Proposition \ref{Prop: Lorenz gauge for deltad-normal bc}.
	In addition $\mathrm{d}\chi=\mathrm{d}\chi_0$.

	To better codify the results of the preceding discussion, it is also convenient to introduce the following linear spaces:
	\begin{align}
		\label{Eq: gauge fixed solutions with wave-tangential bc}
		\mathcal{S}^\Box_{\mathrm{t}}(M)&\doteq\{A\in\Omega^k(M)\;|\;
		\Box A=0\;,\;\delta A=0\;,\;\mathrm{t}A=0\}\,,\\
		\label{Eq: gauge fixed solutions with wave-normal bc}
		\mathcal{S}^\Box_{\mathrm{nd}}(M)&\doteq\{A\in\Omega^k(M)\;|\;
		\Box A=0\;,\;\delta A=0\;,\;\mathrm{n}A=0,\;\;\mathrm{nd}A=0\}\,.
	\end{align}
	where $f\in C^\infty(\partial M)$.
	Hence Propositions \ref{Prop: Lorenz gauge for deltad-tangential bc}-\ref{Prop: Lorenz gauge for deltad-normal bc} can be summarized as stating the existence of the following isomorphisms: 
	\begin{equation}\label{Eq: Isomorphisms}
		\mathcal{S}_{\mathcal{G}_{\mathrm{t}},k}(M)\doteq
		\frac{\mathcal{S}^\Box_{\mathrm{t}}(M)}{\mathrm{d}\mathcal{G}_{\mathrm{t}}(M)}\simeq
		\operatorname{Sol}_{\mathrm{t}}(M)\,,\qquad
		\mathcal{S}_{\mathcal{G}_{\mathrm{nd}},k}(M)\doteq
		\frac{\mathcal{S}^\Box_{\mathrm{nd}}(M)}{\mathrm{d}\mathcal{G}_{\mathrm{nd}}(M)}\simeq
		\operatorname{Sol}_{\mathrm{nd}}(M)\,.
	\end{equation}
\end{remark}

It is noteworthy that both $\operatorname{Sol}_{\mathrm{t}}^{\mathrm{sc}}(M), \operatorname{Sol}^{\mathrm{sc}}_{\mathrm{nd}}(M)$ can be endowed with a presymplectic form -- \textit{cf.} \cite[Prop. 5.1]{Hack-Schenkel-13}.

\begin{proposition}\label{Prop: presymplectic structure on spacelike solutions with gauge boundary conditions}
	Let $(M,g)$ be a globally hyperbolic spacetime with timelike boundary.
	Let $[A_1],[A_2]\in\operatorname{Sol}_{\mathrm{t}}^{\mathrm{sc}}(M)$ and, for $A_1\in[A_1]$, let $A_1=A_1^++A_1^-$ be any decomposition such that $A^+\in\Omega_{\mathrm{spc,t}}^k(M)$ while $A^-\in\Omega_{\mathrm{sfc,t}}^k(M)$
	-- \textit{cf.} Lemma \ref{Lem: on boundary conditions preserving splitting}.
	Then the following map $\sigma_{\mathrm{t}}\colon\operatorname{Sol}_{\mathrm{t}}^{\mathrm{sc}}(M)^{\times 2}\to\mathbb{R}$ is a presymplectic form:
	\begin{align}\label{Eq: presymplectic structure on solutions with gauge boundary conditions}
	\sigma_{\mathrm{t}}([A_1],[A_2])=
	(\delta\mathrm{d}A_1^+,A_2)\,,\qquad
	\forall [A_1],[A_2]\in\operatorname{Sol}_{\mathrm{t}}^{\mathrm{sc}}(M)\,.
	\end{align}
	A similar result holds for $\operatorname{Sol}_{\mathrm{nd}}^{\mathrm{sc}}(M)$ and we denote the associated presymplectic form $\sigma_{\mathrm{nd}}$.
	In particular for all $[A_1],[A_2]\in\operatorname{Sol}_{\mathrm{nd}}^{\mathrm{sc}}(M)$ we have $\sigma_{\mathrm{nd}}([A_1],[A_2])\doteq(\delta\mathrm{d}A_1^+,A_2)$ where $A_1\in[A_1]$ is such that $A\in\Omega^k_{\mathrm{sc},\perp}(M)$.
\end{proposition}
\begin{proof}
	We shall prove the result for $\sigma_{\mathrm{nd}}$, the proof for $\sigma_{\mathrm{t}}$ being the same mutatis mutandis.
	
	First of all notice that for all $[A]\in\operatorname{Sol}_{\mathrm{nd}}^{\mathrm{sc}}(M)$ there exists $A'\in[A]$ such that $A'\in\Omega_\perp^k(M)$.
	This is realized by picking an arbitrary $A\in[A]$ and defining $A'\doteq A+\mathrm{d}\chi$ where $\chi\in\Omega_{\mathrm{sc}}^{k-1}(M)$ is such that $\mathrm{nd}\chi=-\mathrm{n}A$ -- \textit{cf.} Remark \ref{Rmk: surjectivity of t,n,tdelta,nd}.
	We can thus apply Lemma \ref{Lem: on boundary conditions preserving splitting} in order to split $A'=A'_++A'_-$ where $A'_+\in\Omega_{\mathrm{spc,nd}}^k(M)$ and $A'_-\in\Omega_{\mathrm{sfc,nd}}^k(M)$. Notice that this procedure is not necessary for $\delta\mathrm{d}$-tangential boundary condition since we can always split $A\in\Omega_{\mathrm{sc,t}}^k(M)$ as
	$A=A^++A^-$ with $A_+\in\Omega_{\mathrm{spc,t}}^k(M)$ and $A_-\in\Omega_{\mathrm{sfc,t}}^k(M)$ without invoking Lemma \ref{Lem: on boundary conditions preserving splitting}.
	
	After these preliminary observations consider the map
	\begin{align*}
		\sigma_{\mathrm{nd}}\colon(\ker\delta\mathrm{d}\cap\Omega_{\mathrm{sc},\perp}^k(M))^{\times 2}\ni(A_1,A_2)\mapsto(\delta\mathrm{d}A_1^+,A_2)\,,
	\end{align*}
	where we used Lemma \ref{Lem: on boundary conditions preserving splitting} and we split $A_1=A_1^++A_1^-$, with $A_1^+\in\Omega_{\mathrm{spc},\perp}^k(M)$ while $A_1^-\in\Omega_{\mathrm{sfc},\perp}^k(M)$.
	The pairing $(\delta\mathrm{d}A_1^+,A_2)$ is finite because $A_2$ is a spacelike compact $k$-form while $\delta\mathrm{d}A_1^+$ is compactly supported on account of $A_1$ being on-shell.
	Moreover, $(\delta\mathrm{d}A_1^+,A_2)$ is independent from the splitting $A_1=A_1^++A_1^-$ and thus $\sigma_{\mathrm{nd}}$ is well-defined.
	Indeed, let $A_1=\widetilde{A}_1^++\widetilde{A}_1^-$ be another splitting: it follows that $A_1^+-\widetilde{A}_1^+=-(A_1^--\widetilde{A}_1^-)\in\Omega_{\mathrm{c,nd}}^k(M)$. Therefore
	\begin{align*}
		(\delta\mathrm{d}\widetilde{A}_1^+,A_2)=
		(\delta\mathrm{d}A_1^+,A_2)+
		(\delta\mathrm{d}(\widetilde{A}_1^+-A_1^+),A_2)=
		(\delta\mathrm{d}A_1^+,A_2)\,,
	\end{align*}
	where in the last equality we used the self-adjointness of $\delta\mathrm{d}$ on $\Omega_{\mathrm{nd}}^k(M)$.
		
	We show that $\sigma_{\mathrm{nd}}(A_1,A_2)=-\sigma_{\mathrm{nd}}(A_2,A_1)$ for all $A_1,A_2\in\ker\delta\mathrm{d}\cap\Omega_{\mathrm{sc},\perp}^k(M)$.
	For that we have
	\begin{align*}
	\sigma_{\mathrm{nd}}(A_1,A_2)=
	(\delta\mathrm{d}A_1^+,A_2)&=
	(\delta\mathrm{d}A_1^+,A_2^+)+(\delta\mathrm{d}A_1^+,A_2^-)\\&=
	-(\delta\mathrm{d}A_1^-,A_2^+)+(\delta\mathrm{d}A_1^+,A_2^-)\\&=
	-(A_1^-,\delta\mathrm{d}A_2^+)+(A_1^+,\delta\mathrm{d}A_2^-)\\&=
	-(A_1^-,\delta\mathrm{d}A_2^+)-(A_1^+,\delta\mathrm{d}A_2^+)\\&=
	-(A_1,\delta\mathrm{d}A_2^+)=
	-\sigma_{\mathrm{nd}}(A_1,A_2)\,,
	\end{align*}
	where we exploited Lemma \ref{Lem: on boundary conditions preserving splitting} and $A_1^\pm,A_2^\pm\in\Omega_{\mathrm{sc,nd}}^k(M)$.
	
	Finally we prove that $\sigma_{\mathrm{nd}}(A_1,\mathrm{d}\chi)=0$ for all $\chi\in\Omega^k_{\mathrm{sc}}(M)$.
	Together with the antisymmetry shown before, this entails that $\sigma_{\mathrm{nd}}$ descends to a well-defined map $\sigma_{\mathrm{nd}}\colon\operatorname{Sol}_{\mathrm{nd}}^{\mathrm{sc}}(M)^{\times 2}\to\mathbb{R}$ which is bilinear and antisymmetric. Therefore it is a presymplectic form.
	To this end let $\chi\in\Omega^{k-1}_{\mathrm{sc}}(M)$: we have
	\begin{align*}
		\sigma_{\mathrm{nd}}(A,\mathrm{d}\chi)=
		(\delta\mathrm{d}A_1^+,\mathrm{d}\chi)=
		(\delta^2\mathrm{d}A_1^+,\chi)
		+(\mathrm{n}\delta\mathrm{d}A^+,\mathrm{t}\chi)=0\,,
	\end{align*}
	where we used equation \eqref{Eq: boundary terms for delta and d} as well as $\mathrm{n}\delta\mathrm{d}A^+=-\delta\mathrm{nd}A^+=0$.
\end{proof}

Working either with $\operatorname{Sol}_{\mathrm{t}}^{(\mathrm{sc})}(M)$ or $\operatorname{Sol}_{\mathrm{nd}}^{(\mathrm{sc})}(M)$
leads to the natural question whether it is possible to give an equivalent representation of these spaces in terms of compactly supported $k$-forms.
Using Assumption \ref{Thm: assumption theorem}, the following proposition holds true:
 
\begin{proposition}\label{Prop: characterization of solution space in terms of test-forms}
	Let $(M,g)$ be a globally hyperbolic spacetime with timelike boundary.
	Then the following linear maps are isomorphisms of vector spaces
	\begin{align}
		&G_\parallel\colon
		\frac{\Omega_{\mathrm{tc},\delta}^k(M)}{\delta\mathrm{d}\Omega_{\mathrm{tc},\mathrm{t}}^k(M)}\to
		\operatorname{Sol}_{\mathrm{t}}(M)\,,\qquad\quad\;
		G_\parallel\colon
		\frac{\Omega_{\mathrm{c},\delta}^k(M)}{\delta\mathrm{d}\Omega_{\mathrm{c},\mathrm{t}}^k(M)}\to
		\operatorname{Sol}_{\mathrm{t}}^\mathrm{sc}{}(M)\,,\\
		&G_\perp\colon
		\frac{\Omega_{\mathrm{tc,n},\delta}^k(M)}{\delta\mathrm{d}\Omega_{\mathrm{tc},\mathrm{nd}}^k(M)}\to
		\operatorname{Sol}_{\mathrm{nd}}(M)\,,\qquad
		G_\perp\colon
		\frac{\Omega_{\mathrm{c,n},\delta}^k(M)}{\delta\mathrm{d}\Omega_{\mathrm{c},\mathrm{nd}}^k(M)}\to
		\operatorname{Sol}_{\mathrm{nd}}^\mathrm{sc}{}(M)\,,
	\end{align}
\end{proposition} 

\begin{proof}
	Mutatis mutandis, the proof of the four isomorphisms is the same.
	Hence we focus only on the case of timelike supported $k$-forms discussing separately the statement for $\parallel$- and $\perp$- boundary conditions.
	
	\paragraph{Proof for $\delta\mathrm{d}$-tangential boundary conditions.}
	A direct computation shows that $G_\parallel\left[\Omega_{\mathrm{tc},\delta}^k(M)\right]\subseteq\mathcal{S}^\Box_{\mathrm{t},k}(M)$. The condition $\delta G_\parallel\omega=0$ follows from Corollary \ref{Cor: G commutes with d, delta} -- \textit{cf.} Equation \eqref{Eq: relation between d,delta with parallel-propagator on domains}.
	Moreover, $G_\parallel$ descends to the quotient since for all $\eta\in\Omega_{\mathrm{tc},\mathrm{t}}^k(M)$ we have
	\begin{align*}
		G_\parallel\delta\mathrm{d}\eta =-G_\parallel\delta\mathrm{d}\eta
		=-\delta\mathrm{d}G_\parallel\eta
		=-\mathrm{d}\delta G_\parallel\eta\in\mathrm{d}\Omega_{\mathrm{t}}^{k-1}(M)\,,
	\end{align*}
	where we used Equation \eqref{Eq: relation between d,delta with parallel-propagator on domains} and Proposition \ref{Prop: exact sequence and duality relations}.
	
	Proposition \ref{Prop: Lorenz gauge for deltad-tangential bc} entails that $G_\parallel$ is surjective.
	We show that $G_\parallel$ is injective: let $[\alpha]\in\frac{\Omega_{\mathrm{tc},\delta}^k(M)}{\delta\mathrm{d}\Omega_{\mathrm{tc},\mathrm{t}}^k(M)}$ be such that $[G_\parallel\alpha]=[0]$.
	This entails that there exists $\chi\in\Omega^{k-1}_{\mathrm{tc,t}}(M)$ such that $G_\parallel\alpha=\mathrm{d}\chi$.
	Corollary \ref{Cor: G commutes with d, delta} and $\alpha\in\Omega_{\mathrm{tc},\delta}^k(M)$ ensures that $\delta\mathrm{d}\chi=0$, therefore $\chi\in\operatorname{Sol}_{\mathrm{t}}(M)$.
	
	Proposition \ref{Prop: Lorenz gauge for deltad-tangential bc} ensures that $\mathrm{d}\chi=\mathrm{d}G_\parallel\beta$ with $\beta\in\Omega_{\mathrm{tc},\delta}^{k-1}(M)$ while Corollary \ref{Cor: G commutes with d, delta} implies that there exists $\eta_\parallel\in\Omega^k_{\mathrm{tc}}(M)$ such that
	\begin{align*}
		G_\parallel\alpha
		=\mathrm{d}\chi
		=\mathrm{d}G_\parallel\beta
		=G_\parallel(\mathrm{d}\beta-\Box\eta_\parallel)\,.
	\end{align*}
	In addition it holds $\mathrm{t}\eta_\parallel=0$ and $\mathrm{t}\delta\eta_\parallel=\mathrm{t}\chi$.
	
	It follows that $\alpha-\mathrm{d}\beta+\Box\eta_\parallel\in\ker G_\parallel$ and therefore $\alpha-\mathrm{d}\beta+\Box\eta_\parallel=\Box_\parallel\zeta$ for $\zeta\in\Omega_{\mathrm{tc},\parallel}^k(M)$ -- \textit{cf.} Remark \ref{Rmk: extension of short exact sequence}.
	Applying $\delta$ to the last equality we find
	\begin{align*}
		0=\delta\alpha
		-\delta\mathrm{d}\beta
		+\Box\delta\eta_\parallel
		-\Box\delta\zeta
		=\Box(\delta\eta_\parallel-\delta\zeta-\beta)\,,
	\end{align*}
	where used that $\delta\beta=0$.
	Remark \ref{Rmk: compactly supported solutions of the wave operator} entails that $\delta\eta_\parallel-\delta\zeta-\beta=0$, therefore,
	\begin{align*}
		\alpha
		=\Box_\parallel\zeta
		+\mathrm{d}\beta
		-\Box\eta_\parallel
		=\delta\mathrm{d}(\zeta-\eta_\parallel)
		\in\delta\mathrm{d}\Omega^k_{\mathrm{c,t}}(M)\,,
	\end{align*}
 	that is, $[\alpha]=[0]$.
 	
 	\paragraph{Proof for $\delta\mathrm{d}$-normal boundary conditions.}
 	By direct inspection we have that $G_\perp\alpha\in\mathcal{S}_{\mathrm{nd},k}^\Box(M)$ for all $\alpha\in\Omega^k_{\mathrm{tc,n},\delta}(M)$.
 	Furthermore, Equation \eqref{Eq: relation between d,delta with parallel-propagator on domains} entails that $\delta G_\perp\alpha=G_\perp\delta\alpha=0$.
 	The map $G_\perp$ also descends to the quotients since for all $\eta\in\Omega^k_{\mathrm{c,nd}}(M)$ if holds
 	\begin{align*}
 		G_\perp\delta \mathrm{d}\eta
 		=\delta\mathrm{d}G_\perp\eta
 		=-\mathrm{d}\delta G_\perp\eta\in\mathrm{d}\Omega^{k-1}_{\mathrm{tc}}(M)\,,
 	\end{align*}
 	where in the second equality we used Corollary \ref{Cor: G commutes with d, delta}.
 	
 	Surjectivity of $G_\perp$ is guaranteed by Proposition \ref{Prop: Lorenz gauge for deltad-normal bc}.
 	We show injectivity of $G_\perp$: let $[\alpha]\in\frac{\Omega^k_{\mathrm{tc,n},\delta}(M)}{\delta\mathrm{d}\Omega^k_{\mathrm{tc,nd}}(M)}$ be such that $[G_\perp\alpha]=[0]$.
 	By definition there exists $\chi\in\Omega^{k-1}_{\mathrm{tc}}(M)$ such that $G_\perp\alpha=\mathrm{d}\chi$.
 	Since $\mathrm{n}\alpha=0$, Corollary \ref{Cor: G commutes with d, delta} entails that $\delta\mathrm{d}\chi=0$. In addition it holds $\mathrm{nd}\chi=\mathrm{n}G_\perp\alpha=0$.
 	It follows that $[\chi]\in\operatorname{Sol}_{\mathrm{nd}}(M)$.
 	
 	Proposition \ref{Prop: Lorenz gauge for deltad-normal bc} implies that there exists $\beta\in\Omega^{k-1}_{\mathrm{tc,n},\delta}(M)$ such that
 	\begin{align*}
	 	G_\perp\alpha
	 	=\mathrm{d}\chi
	 	=\mathrm{d}G_\perp\beta
	 	=G_\perp\mathrm{d}\beta\,,
 	\end{align*}
 	where we used equation \eqref{Eq: relation between d,delta with perp-propagator on domains}.
 	It follows that $\alpha-\mathrm{d}\beta\in\ker G_\perp$, therefore $\alpha-\mathrm{d}\beta=\Box_\perp\eta$ for $\eta\in\Omega^k_{\mathrm{tc},\perp}(M)$ -- \textit{cf.} Proposition \ref{Prop: exact sequence and duality relations} and Remark \ref{Rmk: extension of short exact sequence}.
 	Application of $\delta$ entails
 	\begin{align*}
	 	0=\delta(\alpha-\mathrm{d}\beta-\Box_\perp\eta)
	 	=-\Box(\beta+\delta\eta)\,,
 	\end{align*}
 	where we also used that $\delta\beta=0$.
 	Remark \ref{Rmk: compactly supported solutions of the wave operator} implies that $\beta+\delta\eta=0$, hence
 	\begin{align*}
	 	\alpha
	 	=\Box_\perp\eta+\mathrm{d}\beta
	 	=\delta\mathrm{d}\eta
	 	\in\delta\mathrm{d}\Omega^k_{\mathrm{tc,nd}}(M)\,,
 	\end{align*}
 	which entails $[\alpha]=[0]$.
\end{proof} 

The following proposition shows that the isomorphisms introduced in Proposition \ref{Prop: characterization of solution space in terms of test-forms} for $\operatorname{Sol}_{\mathrm{t}}^{\mathrm{sc}}(M)$ and $\operatorname{Sol}_{\mathrm{nd}}^{\mathrm{sc}}(M)$ lift to isomorphisms of presymplectic spaces.

\begin{proposition}\label{Prop: presymplectomorphism for spacelike solution spaces}
	Let $(M,g)$ be a globally hyperbolic spacetime with timelike boundary.
	The following statements hold true:
	\begin{enumerate}
		\item
		$\frac{\Omega^k_{\mathrm{c},\delta}(M)}{\delta\mathrm{d}\Omega_{\mathrm{c},\mathrm{t}}^k(M)}$ is a pre-symplectic space if endowed with the bilinear map $\widetilde{G}_\parallel([\alpha],[\beta])\doteq(\alpha,G_\parallel\beta)$.
		
		Moreover $\bigg(\frac{\Omega^k_{\mathrm{c},\delta}(M)}{\delta\mathrm{d}\Omega_{\mathrm{c},\mathrm{t}}^k(M)},\widetilde{G}_\parallel\bigg)$ is pre-symplectomorphic to $(\operatorname{Sol}^{\mathrm{sc}}_{\mathrm{t}}(M),\sigma_{\mathrm{t}})$.
		\item
		$\frac{\Omega^k_{\mathrm{c,n},\delta}(M)}{\delta\mathrm{d}\Omega_{\mathrm{c,nd}}^k(M)}$ is a pre-symplectic space
		if endowed with the bilinear map $\widetilde{G}_\perp([\alpha],[\beta])\doteq(\alpha,G_\perp\beta)$.

		Moreover $\bigg(\frac{\Omega^k_{\mathrm{c,n},\delta}(M)}{\delta\mathrm{d}\Omega_{\mathrm{c},\mathrm{nd}}^k(M)},\widetilde{G}_\perp\bigg)$ is pre-symplectomorphic to $(\operatorname{Sol}^{\mathrm{sc}}_{\mathrm{nd}}(M),\sigma_{\mathrm{nd}})$.
\end{enumerate}  
\end{proposition}

\begin{proof}
	We discuss the two cases separately.
	\paragraph{Proof for $\delta\mathrm{d}$-tangential boundary conditions.}
	We observe that $\widetilde{G}_\parallel$ is well-defined.
	As a matter of fact, let $\alpha,\beta\in\Omega_{\mathrm{c},\delta}^k(M)$, then $G_\parallel\beta\in\Omega_{\mathrm{sc}}^k(M)$. Therefore the pairing $(\alpha,G_\parallel\beta)$ is finite. Furthermore, if $\eta\in\Omega_{\mathrm{c},\mathrm{t}}^k(M)$, it holds
	\begin{align*}
		(\delta\mathrm{d}\eta,G_\parallel\beta)
		&=(\eta,\delta\mathrm{d}G_\parallel\beta)
		=-(\eta,\mathrm{d}\delta G_\parallel\beta)
		=-(\eta,\mathrm{d}G_\parallel\delta\beta)
		=0\,,\\
		(\alpha,G_\parallel\delta\mathrm{d}\eta)
		&=(\alpha,\delta\mathrm{d}G_\parallel\eta)
		=-(\alpha,\mathrm{d}\delta G_\parallel\eta)
		=-(\alpha,\mathrm{d}G_\parallel\delta\eta)=0\,,
	\end{align*}
	where we used that $G_\parallel\beta,\eta\in\Omega_{\mathrm{c},\mathrm{t}}^k(M)$ -- \textit{cf.} Equation \eqref{Eq: boundary terms for delta d operator} -- as well as Equation \eqref{Eq: relation between d,delta with parallel-propagator on domains}.
	Therefore $\widetilde{G}_\parallel$ is well-defined: Moreover, it is per construction bilinear and antisymmetric, therefore it induces a pre-symplectic structure.
	
	We prove that the isomorphism $G_\parallel\colon\frac{\Omega_{\mathrm{c},\delta}^k(M)}{\delta\mathrm{d}\Omega_{\mathrm{c},\parallel}^k(M)}\to\operatorname{Sol}_{\mathrm{t}}(M)$ introduced in Proposition \ref{Prop: characterization of solution space in terms of test-forms} is a pre-symplectomorphism.
	Let $[\alpha],[\beta]\in\frac{\Omega^k_{\mathrm{c},\delta}(M)}{\delta\mathrm{d}\Omega_{\mathrm{c},\parallel}^k(M)}$.
	As a direct consequence of the properties of $G_\parallel=G_\parallel^+-G_\parallel^-$, calling $A_1=G_\parallel\alpha$ and $A_2=G_\parallel\beta$, we can consider $A_1^\pm=G_\parallel^\pm\alpha$ in Equation \eqref{Eq: presymplectic structure on solutions with gauge boundary conditions}.
	This yields
	\begin{align*}
		\sigma_{\mathrm{t}}([G_\parallel\alpha],[G_\parallel\beta])=
		(\delta\mathrm{d}G_\parallel^+\alpha,G_\parallel\beta)=
		(\Box G_\parallel^+\alpha-\mathrm{d}\delta G_\parallel^+\alpha,G_\parallel\beta)=
		(\alpha,G_\parallel\beta)=
		\widetilde{G}_\parallel([\alpha],[\beta])\,,
	\end{align*}
	where we used Lemma \ref{Lem: relations between delta,d and advanced-retarded propagators} so that $\mathrm{d}\delta G_\parallel^+\alpha=\mathrm{d}G_\parallel^+\delta\alpha=0$.
	
	\paragraph{Proof for $\delta\mathrm{d}$-normal boundary conditions.}
	We observe that $\widetilde{G}_\perp$ is well-defined: indeed for all $\alpha\in\Omega^k_{\mathrm{c,n},\delta}(M)$ and $\eta\in\Omega^k_{\mathrm{c,nd}}(M)$ we have
	\begin{align*}
		(\alpha,G_\perp\delta\mathrm{d}\eta)
		&=(\alpha,\delta\mathrm{d}G_\perp\eta)
		=-(\alpha,\mathrm{d}\delta G_\perp\eta)
		=-(\delta\alpha,\delta G_\perp\eta)
		-(\mathrm{n}\alpha,\mathrm{t}\delta G_\perp\eta)
		=0
		\\
		(\delta\mathrm{d}\eta,G_\perp\alpha)
		&=(\eta,\delta\mathrm{d}G_\perp\alpha)
		=-(\eta,\mathrm{d}\delta G_\perp\alpha)
		=-(\eta,\mathrm{d}G_\perp\delta\alpha)
		=0\,,
	\end{align*}
	where in the former chain of equalities we used Equation \eqref{Eq: relation between d,delta with perp-propagator on domains} as well as Equation \eqref{Eq: boundary terms for delta and d}, while in the latter we used Equation \eqref{Eq: boundary terms for delta d operator} and Equation \eqref{Eq: relation between d,delta with perp-propagator on domains}.
	
	Proposition \ref{Prop: exact sequence and duality relations} shows that $\widetilde{G}_\perp$ induces a presymplectic structure.
	We now prove that the isomorphism $G_\perp\colon\frac{\Omega^k_{\mathrm{c,n},\delta}(M)}{\delta\mathrm{d}\Omega_{\mathrm{c,nd}}^k(M)}\to\operatorname{Sol}_{\mathrm{nd}}^{\mathrm{sc}}(M)$ introduced in Proposition \ref{Prop: characterization of solution space in terms of test-forms} is in fact a symplectomorphism.
	Let $[\alpha],[\beta]\in\frac{\Omega^k_{\mathrm{c,n},\delta}(M)}{\delta\mathrm{d}\Omega_{\mathrm{c,nd}}^k(M)}$ and let $[A_\alpha]:=[G_\perp\alpha], [A_\beta]:=[G_\perp\beta]\in\operatorname{Sol}_{\mathrm{nd}}^{\mathrm{sc}}(M)$.
	Following \ref{Prop: presymplectic structure on spacelike solutions with gauge boundary conditions} we can choose $A_\alpha^\pm=G_\perp^\pm\alpha$ so that
	\begin{align*}
		\sigma_{\mathrm{nd}}([A_\alpha],[A_\beta])
		=(\delta\mathrm{d}G_\perp^+\alpha,G_\perp\beta)
		=(\Box G_\perp^+\alpha,G_\perp\beta)
		-(\mathrm{d}\delta G_\perp^+\alpha,G_\perp\beta)
		=(\alpha,G_\perp\beta)
		=\widetilde{G}_\perp([\alpha],[\beta])\,,
	\end{align*}
	where proposition \ref{Lem: relations between delta,d and advanced-retarded propagators} ensures that $\mathrm{d}\delta G_\perp^+\alpha=\mathrm{d}G_\perp^+\delta\alpha=0$.	
\end{proof}

\begin{remark}\label{Rmk: on degeneracy on presymplectic structure}
	Following \cite[Cor. 5.3]{Hack-Schenkel-13}, $\sigma_{\mathrm{t}}$ (\textit{resp.} $\sigma_{\mathrm{nd}}$) do not define in general a symplectic form on the space of spacelike compact solutions $\operatorname{Sol}_{\mathrm{t}}(M)$ (\textit{resp.} $\operatorname{Sol}_{\mathrm{nd}}(M)$).
	A direct characterization of this deficiency is best understood  by introducing the following quotients:
	\begin{align}
		\widehat{\operatorname{Sol}}_{\mathrm{t}}^{\mathrm{sc}}:=
		\frac{\lbrace A\in\Omega^k_{\mathrm{sc}}(M)\;|\;\delta\mathrm{d}A=0\,,\;\mathrm{t}A=0\rbrace}
		{\mathrm{d}\Omega_{\mathrm{t}}^{k-1}(M)\cap\Omega^k_{\mathrm{sc}}(M)}\,,\quad
		\widehat{\operatorname{Sol}}_{\mathrm{nd}}^{\mathrm{sc}}:=
		\frac{\lbrace A\in\Omega^k_{\mathrm{sc}}(M)\;|\;\delta\mathrm{d}A=0\,,\;\mathrm{nd}A=0\rbrace}
		{\mathrm{d}\Omega^{k-1}(M)\cap\Omega^k_{\mathrm{sc}}(M)}\,,
	\end{align}
	Focusing on $\delta\mathrm{d}$-normal boundary conditions, it follows that there is a natural surjective linear map $\operatorname{Sol}_{\mathrm{nd}}^{\mathrm{sc}}\to\widehat{\operatorname{Sol}}_{\mathrm{nd}}^{\mathrm{sc}}$.
	Moreover, $\widehat{\operatorname{Sol}}_{\mathrm{nd}}^{\mathrm{sc}}$ is symplectic with respect to the form $\sigma_{\mathrm{nd}}([A_1],[A_2])=(\delta\mathrm{d}A_1^+,A_2)$.
	This can be shown as follows: if $\sigma_{\mathrm{nd}}([A_1],[A_2])=0$ for all $[A_1]\in\widehat{\operatorname{Sol}}_{\mathrm{nd}}^{\mathrm{sc}}$ then, choosing $A_1=G_\perp\alpha$ with $\alpha\in\Omega^k_{\mathrm{c,n},\delta}(M)$ leads to $0=\sigma_\perp([G_\perp\alpha],[A_2])=(\alpha,A_2)$ -- \textit{cf.} Proposition \ref{Prop: presymplectomorphism for spacelike solution spaces}.
	This entails $\mathrm{d}A_2=0$ as well as $A_2=0\in H_{k,\mathrm{c,n}}(M)^*\simeq H^k(M)$ -- \textit{cf.} Appendix \ref{App: Poincare duality for manifold with boundary}.
	Therefore $A_2=\mathrm{d}\chi$ where $\chi\in\Omega^{k-1}(M)$ that is $[A_2]=[0]$ in $\widehat{\operatorname{Sol}}_{\mathrm{nd}}^{\mathrm{sc}}(M)$.
	A similar result holds, mutatis mutandis, for $\parallel$.
	
	The net result is that $(\operatorname{Sol}_{\mathrm{nd}}^{\mathrm{sc}}(M),\sigma_{\mathrm{nd}})$ (\textit{resp.} $(\operatorname{Sol}_{\mathrm{t}}^{\mathrm{sc}}(M),\sigma_{\mathrm{t}})$) is symplectic if and only if  $\mathrm{d}\Omega^{k-1}_{\mathrm{sc}}(M)=\Omega^k_{\mathrm{sc}}(M)\cap\mathrm{d}\Omega^{k-1}(M)$ (\textit{resp.} $\mathrm{d}\Omega^{k-1}_{\mathrm{sc,t}}(M)=\Omega^k_{\mathrm{sc}}(M)\cap\mathrm{d}\Omega_{\mathrm{t}}^{k-1}(M)$).
	This is in agreement with the analysis in \cite{Benini:2013tra} for the case of globally hyperbolic spacetimes $(M,g)$ with $\partial M=\emptyset$.
\end{remark}

\begin{Example}
	We construct an example of a globally hyperbolic spacetime with timelike boundary $(M,g)$ such that $\mathrm{d}\Omega_{\mathrm{sc}}^{k-1}(M)$ is properly included in $\Omega_{\mathrm{sc}}^k(M)\cap\mathrm{d}\Omega^{k-1}(M)$ -- \textit{cf.} \cite[Ex. 5.7]{Hack-Schenkel-13} for the case with empty boundary.
	Consider half-Minkowski spacetime $\mathbb{R}^m_+:=\mathbb{R}^{m-1}\times\overline{\mathbb{R}_+}$ with flat metric and let $p\in\mathring{\mathbb{R}}^m_+$.
	Let $M\doteq\mathbb{R}^m_+\setminus J(p)$ be endowed with the restriction to $M$ of the Minkowski metric. This spacetime is globally hyperbolic with timelike boundary.
	Let $B_1, B_2$ be open balls in $\mathbb{R}^m_+$ centered at $p$ with $B_1\subset B_2$.
	
	We consider $\psi\in\Omega^0(M)$ such that $\psi|_{J(B_1\cap M)}=1$ and $\psi|_{M\setminus J(B_2\cap M)}=0$.
	In addition we introduce $\varphi\in\Omega^0(M)$ such that:
	(a) $\varphi=\hat{\varphi}|_{M}$ is the restriction to $M$ of an element $\hat{\varphi}\in\Omega_{\mathrm{tc}}^0(\mathbb{R}^m_+)$ such that for all $x\in \mathbb{R}_+^m$, $\hat{\varphi}(x)$ depends only on $\tau_{\mathbb{R}_+^m}(x)$ -- \textit{cf}. Theorem \ref{Thm: globally hyperbolic spacetime with time-like boundary} -- and $\hat{\varphi}(p)=1$;
	(b) $\chi:=\varphi\psi\in\Omega^0(M)$ is such that $\mathrm{t}\chi=\chi|_{\partial M}=0$;
	
	The existence of such function $\hat{\varphi}$ is guaranteed by the fact that $p\in\mathring{\mathbb{R}}^m_+$.
	Hence there exists a non-empty interval $I\subseteq\mathbb{R}$ such that $\tau_{\mathbb{R}_+^m}(p)\in I$ and $x\notin\partial M$ for all $x\in\operatorname{supp}(\psi)$ with $\tau_{\mathbb{R}_+^m}(x)\in I$.
	
	It follows that $\mathrm{d}\chi\in\Omega_{\mathrm{sc}}^1(M)$, though $\chi\in\Omega^0(M)$ is not spacelike compact. In addition there does not exist $\zeta\in\Omega_{\mathrm{sc}}^0(M)$ such that $\mathrm{d}\zeta=\mathrm{d}\chi$.
	Indeed, let us consider the spacelike curve $\gamma$ in $M$
	\begin{align*}
		\gamma(s)
		=(\tau(p),x_1(p),\dots,x_{m-2}(p),z(p)+s)\,,\qquad s>0\,,
	\end{align*}
	where $\tau=\tau_{\mathbb{R}_+^m},x_1,\ldots,x_{m-2},z$ are Cartesian coordinates on $\mathbb{R}^m_+$.
	Assuming the existence of $\zeta$ with the properties described above, integration along $\gamma$ would lead to the following contradiction
	\begin{align*}
		0
		=\int_\gamma\iota_\gamma^*\mathrm{d}\zeta
		=\int_\gamma\iota_\gamma^*\mathrm{d}\chi
		=-1\,.
	\end{align*}
\end{Example}

\subsection{The algebra of observables for $\operatorname{Sol}_{\mathrm{t}}(M)$ and for $\operatorname{Sol}_{\mathrm{nd}}(M)$}
\label{Sec: Algebra of observables for Sol(M)}

In this section we discuss an application of the previous results that we obtained.
Motivated by the algebraic approach to quantum field theory, we associate a unital $*$-algebra both to $\operatorname{Sol}_{\mathrm{t}}(M)$ and to $\operatorname{Sol}_{\mathrm{nd}}(M)$, whose elements are interpreted as the observables of the underlying quantum system.
Furthermore we study its key structural properties and we comment on their significance.
We recall that the corresponding question, when the underlying background $(M,g)$ is globally hyperbolic manifold with $\partial M=\emptyset$ has been thoroughly discussed in the literature -- \textit{cf.} \cite{Benini-16,Dappiaggi:2011cj,Hack-Schenkel-13,Sanders:2012sf}.


We now introduce the algebra of observables associated to the solution space $\operatorname{Sol}_{\mathrm{t}}(M)$ and $\operatorname{Sol}_{\mathrm{nd}}(M)$ and we discuss its main properties.
Our analysis follows closely that of \cite{Benini-16,Dappiaggi:2011cj,Hack-Schenkel-13,Sanders:2012sf} for globally hyperbolic spacetimes with empty boundary.
\\
Following \cite{Benini-16} we will identify a unital $*$-algebra $\mathcal{A}_{\mathrm{t}}(M)$ (\textit{resp.} $\mathcal{A}_{\mathrm{nd}}(M)$) built out of suitable linear functionals over $\operatorname{Sol}_{\mathrm{t}}(M)$ (\textit{resp.} $\operatorname{Sol}_{\mathrm{nd}}(M)$), whose collection is fixed so to contain enough elements to distinguish all configurations in $\operatorname{Sol}_{\mathrm{t}}(M)$ (\textit{resp.} $\operatorname{Sol}_{\mathrm{nd}}(M)$) -- \textit{cf.} Proposition \ref{Prop: sep and opt for the alg of obs with gauge bc}.

Taking into account the discussion in the preceding sections, particularly Equation \eqref{Eq: relations-bulk-to-boundary} and Definition \ref{Def: tangential and normal component} we introduce the following structures.

\begin{Definition}\label{Def: alg of obs for gauge bc}
	Let $(M,g)$ be a globally hyperbolic spacetime with timelike boundary.
	We call {\em algebra of observables} associated to $\operatorname{Sol}_{\mathrm{t}}(M)$, the associative, unital $*$-algebra 
	\begin{align}\label{Eq: alg of obs for deltad-tangential bc}
		\mathcal{A}_{\mathrm{t}}(M)\doteq\frac{\mathcal{T}[\mathcal{O}_{\mathrm{t}}(M)]}{\mathcal{I}[\mathcal{O}_{\mathrm{t}}(M)]}\,,\qquad
		\mathcal{O}_{\mathrm{t}}(M)\doteq\frac{\Omega_{\mathrm{c},\delta}^k(M)}{\delta\mathrm{d}\Omega_{\mathrm{c,t}}^k(M)}\,.
	\end{align}
	Here $\mathcal{T}[\mathcal{O}_{\mathrm{t}}(M)]\doteq\bigoplus_{n=0}^\infty\mathcal{O}_{\mathrm{t}}(M)^{\otimes n}$ is the universal tensor algebra with
	$\mathcal{O}_{\mathrm{t}}(M)^{\otimes 0}\equiv\mathbb{C}$, while the $*$-operation is
	{
		defined by $(\alpha_1\otimes\ldots\otimes\alpha_n)^*=(\overline{\alpha}_n\otimes\ldots\overline{\alpha}_1)$ for $\alpha_1\otimes\ldots\otimes\alpha_n\in\mathcal{O}_{\mathrm{t}}(M)^{\otimes n}$ and then extended by linearity -- here $\overline{\alpha}$ indicates complex conjugation.
	}
	In addition $\mathcal{I}[\mathcal{O}_{\mathrm{t}}(M)]$ is the $*$-ideal generated by elements of the form
	$[\alpha]\otimes[\beta]-[\beta]\otimes[\alpha]-i \widetilde{G}_\parallel([\alpha],[\beta])\mathbb{I}$, where $[\alpha],[\beta]\in\mathcal{O}_{\mathrm{t}}(M)$ while $\widetilde{G}_\parallel$ is defined in Proposition \ref{Prop: presymplectomorphism for spacelike solution spaces} and $\mathbb{I}$ is the identity of {$\mathcal{T}[\mathcal{O}_{\mathrm{t}}(M)]$}.

	Similarly, we call {\em algebra of observables} associated to $\operatorname{Sol}_{\mathrm{nd}}(M)$, the associative, unital $*$-algebra 
	\begin{align}\label{Eq: alg of obs for deltad-normal bc}
		\mathcal{A}_{\mathrm{nd}}(M)\doteq
		\frac{\mathcal{T}[\mathcal{O}_{\mathrm{nd}}(M)]}{\mathcal{I}[\mathcal{O}_{\mathrm{nd}}(M)]}\,,\qquad
		\mathcal{O}_{\mathrm{nd}}(M)\doteq
		\frac{\Omega_{\mathrm{c,n},\delta}^k(M)}{\delta\mathrm{d}\Omega_{\mathrm{c,nd}}^k(M)}\,.
	\end{align}
	where $\mathcal{T}[\mathcal{O}_{\mathrm{nd}}(M)]\doteq\bigoplus_{n=0}^\infty\mathcal{O}_{\mathrm{nd}}(M)^{\otimes n}$ is the universal tensor algebra with
	$\mathcal{O}_{\mathrm{nd}}(M)^{\otimes 0}\equiv\mathbb{C}$, while the $*$-operation is defined by $(\alpha_1\otimes\ldots\otimes\alpha_n)^*=(\overline{\alpha_n}\otimes\ldots\overline{\alpha_1})$ for $\alpha_1\otimes\ldots\otimes\alpha_n\in\mathcal{O}_{\mathrm{nd}}(M)^{\otimes n}$ and then extended by linearity -- here $\overline{\alpha}$ indicates complex conjugation.
	In addition
	$\mathcal{I}[\mathcal{O}_{\mathrm{nd}}(M)]$ is the $*$-ideal generated by elements of the form
	$[\alpha]\otimes[\beta]-[\beta]\otimes[\alpha]-i \widetilde{G}_\perp([\alpha],[\beta])\mathbb{I}$ , where $[\alpha],[\beta]\in\mathcal{O}_{\mathrm{nd}}(M)$ while $\widetilde{G}_\perp$ is defined in Proposition \ref{Prop: presymplectomorphism for spacelike solution spaces} and $\mathbb{I}$ is the identity of {$\mathcal{T}[\mathcal{O}_{\mathrm{nd}}(M)]$}.
\end{Definition}

\begin{remark}\label{Rmk: bc for obs associated to deltad-normal bc}
	Notice that, with respect to the definition of $\mathcal{O}_{\mathrm{t}}(M)$, the vector space $\mathcal{O}_{\mathrm{nd}}(M)$ introduced in Definition \ref{Def: alg of obs for gauge bc} contains equivalence classes built out of forms $\alpha\in\Omega_{\mathrm{c}}^k(M)$ such that $\delta\alpha=0$ as well as $\mathrm{n}\alpha=0$.
	The last condition is sufficient and necessary to have a well-defined pairing among $\mathcal{O}_{\mathrm{nd}}(M)$ and $\operatorname{Sol}_{\mathrm{nd}}(M)$.
	Indeed for all $A\in[A]\in\operatorname{Sol}_{\mathrm{nd}}(M)$ and for all $\alpha\in[\alpha]\in\mathcal{O}_{\mathrm{nd}}(M)$ we have that $(\alpha,A)$ is well-defined being $\alpha$ compactly supported.
	Moreover, for all $\chi\in\Omega^{k-1}(M)$ and $\eta\in\Omega_{\mathrm{nd}}^k(M)$ it holds
	\begin{align*}
	(\alpha,\mathrm{d}\chi)=
	(\delta\alpha,\chi)+
	(\mathrm{n}\alpha,\mathrm{t}\chi)_\partial=0\,,\qquad
	(\delta\mathrm{d}\eta,A)=
	(\eta,\delta\mathrm{d}A)+
	(\mathrm{nd}\eta,\mathrm{t}A)_\partial-
	(\mathrm{t}\eta,\mathrm{nd}A)_\partial=0\,.
	\end{align*}
	Notice that in the first equation we used the condition $\mathrm{n}\alpha=0$ since $\chi$ has no assigned boundary condition.
	This is opposite to the case of $\delta\mathrm{d}$-tangential boundary conditions, where $\chi$ is required to satisfy $\mathrm{t}\chi=0$ -- \textit{cf.} Definition \ref{Def: configuration space with deltad-tangential and deltad-normal bc} -- and therefore $\alpha$ is not forced to satisfy any boundary condition.
	Actually, the constraints $\delta\alpha=0$ and $\mathrm{n}\alpha=0$ are necessary to ensure gauge-invariance, namely $(\alpha,\mathrm{d}\chi)=0$ for all $\chi\in\Omega^k(M)$.
\end{remark}

\noindent We study the structural properties of the algebra of observables.
On account of its definition, it suffices to focus mainly on the properties of the generators $\mathcal{O}_{\mathrm{t}}(M)$ and $\mathcal{O}_{\mathrm{nd}}(M)$.
In particular, in the next proposition we follow the rationale advocated in \cite{Benini-16} proving that $\mathcal{O}_{\mathrm{t}}(M)$ and $\mathcal{O}_{\mathrm{nd}}(M)$ are {\em optimal}:

\begin{proposition}\label{Prop: sep and opt for the alg of obs with gauge bc}
	Let $\mathcal{O}_{\mathrm{t}}(M),\mathcal{O}_{\mathrm{nd}}(M)$ be as per Definition \ref{Def: alg of obs for gauge bc}.

	Then the pairing $(\;,\;)$ among $k$-forms -- \textit{cf.} equation \eqref{Eq: pairing between k-forms} -- descends to a well-defined pairing
	\begin{align*}
		&\mathcal{O}_{\mathrm{t}}(M)
		\times\operatorname{Sol}_{\mathrm{t}}(M)
		\to\mathbb{C}\,,\qquad
		[\alpha],[A]\mapsto
		(\alpha,A)\\
		&\mathcal{O}_{\mathrm{nd}}(M)
		\times\operatorname{Sol}_{\mathrm{nd}}(M)
		\to\mathbb{C}\,,\qquad
		[\alpha],[A]\mapsto
		(\alpha,A)\,.
	\end{align*}
	Moreover $\mathcal{O}_{\mathrm{t}}(M)$ (\textit{resp}. $\mathcal{O}_{\mathrm{nd}}(M)$) is {\bf optimal} with respect to $\operatorname{Sol}_{\mathrm{t}}(M)$ (\textit{resp.} $\operatorname{Sol}_{\mathrm{nd}}(M)$), namely:
	\begin{enumerate}
		\item 
		$\mathcal{O}_{\mathrm{t}}(M)$ (\textit{resp}. $\mathcal{O}_{\mathrm{nd}}(M)$) is {\em separating} with respect to $\operatorname{Sol}_{\mathrm{t}}(M)$ (\textit{resp.} $\operatorname{Sol}_{\mathrm{nd}}(M)$), that is
		\begin{align}
			\label{Eq: optimality for delta d-tangential bc algebra}
			\forall[A\in\operatorname{Sol}_{\mathrm{t}}(M)\colon
			&([\alpha],[A])
			=0\quad\forall [\alpha]\in\mathcal{O}_{\mathrm{t}}(M)
			\Longrightarrow [A]=[0]\in\operatorname{Sol}_{\mathrm{t}}(M)\,,\\
			\label{Eq: optimality for delta d-normal bc algebra}
			\forall[A]\in\operatorname{Sol}_{\mathrm{nd}}(M)\colon
			&([\alpha],[A])
			=0\quad\forall [\alpha]\in\mathcal{O}_{\mathrm{nd}}(M)
			\Longrightarrow [A]=[0]\in\operatorname{Sol}_{\mathrm{nd}}(M)\,.
		\end{align}
		\item
		$\mathcal{O}_{\mathrm{t}}(M)$ (\textit{resp}. $\mathcal{O}_{\mathrm{nd}}(M)$) is {\em non redundant} with respect to $\operatorname{Sol}_{\mathrm{t}}(M)$ (\textit{resp.} $\operatorname{Sol}_{\mathrm{nd}}(M)$), that is
		\begin{align}
			\label{Eq: non-redundancy for delta d-tangential bc algebra}
			\forall[\alpha]\in\mathcal{O}_{\mathrm{t}}(M)\colon
			&([\alpha],[A])=0
			\quad\forall [A]\in\operatorname{Sol}_{\mathrm{t}}(M)
			\Longrightarrow[\alpha]=[0]\in\mathcal{O}_{\mathrm{t}}(M)\,,\\
			\label{Eq: non-redundancy for delta d-normal bc algebra}
			\forall[\alpha]\in\mathcal{O}_{\mathrm{nd}}(M)\colon
			&([\alpha],[A])=0
			\quad\forall [A]\in\operatorname{Sol}_{\mathrm{nd}}(M)
			\Longrightarrow[\alpha]=[0]\in\mathcal{O}_{\mathrm{nd}}(M)\,,
		\end{align}
	\end{enumerate}
\end{proposition}

\begin{proof}
	Mutatis mutandis, the proof is similar both for the $\delta\mathrm{d}$-tangential and the $\delta\mathrm{d}$-normal boundary conditions.
	
	\paragraph{Proof for $\delta\mathrm{d}$-tangential boundary conditions.}
	As starting point observe that the pairing $([\alpha],[A]):=(\alpha,A)$ is well-defined.
	Indeed let us consider two representatives $A\in [A]\in\operatorname{Sol}_{\mathrm{t}}(M)$ and $\alpha\in[\alpha]\in\mathcal{O}_{\mathrm{t}}(M)$.
	The pairing $(\alpha,A)$ is finite being $\operatorname{supp}(\alpha)$ compact and there is no dependence on the choice of representative.
	As a matter of facts, if $\mathrm{d}\chi\in{\mathrm{d}\Omega_{\mathrm{t}}^{k-1}(M)}$ and $\eta\in\Omega^k_{\mathrm{c,t}}(M)$, it holds
	\begin{align*}
		(\alpha,\mathrm{d}\chi)=
		(\delta\alpha,\chi)+(\mathrm{n}\alpha,\mathrm{t}\chi)_\partial=0\,,\qquad
		(\delta\mathrm{d}\eta,A)=
		(\eta,\delta\mathrm{d}A)+
		(\mathrm{t}\eta,\mathrm{nd}A)_\partial-
		(\mathrm{nd}\eta,\mathrm{t}A)_\partial=0\,,
	\end{align*}
	where in the first equation we used that $\mathrm{t}\chi=0$ as well as $\delta\alpha=0$, while in the second equation we used $\delta\mathrm{d}A=0$ as well as $\mathrm{t}A=\mathrm{t}\eta=0$.
	
	Having established that the pairing between the equivalence classes is well-defined we prove equations \eqref{Eq: optimality for delta d-tangential bc algebra}-\eqref{Eq: non-redundancy for delta d-tangential bc algebra} separately.

	\paragraph{Proof of Equation \eqref{Eq: optimality for delta d-tangential bc algebra}.}
	Assume $\exists [A]\in\operatorname{Sol}_{\mathrm{t}}(M)$ such that $([\alpha],[A])=0$, $\forall [\alpha]\in\mathcal{O}_{\mathrm{t}}(M).$ Working at the level of representative, since $\alpha\in\Omega^k_{\mathrm{c},\delta}(M)$ we can choose $\alpha=\delta\beta$ with $\beta\in\Omega^{k+1}_{\mathrm{c}}(M)$.
	As a consequence $0=(\delta\beta, A)=(\beta,\mathrm{d}A)$ where we used implicitly \eqref{Eq: boundary terms for delta and d} and $\mathrm{t}A=0$.
	The arbitrariness of $\beta$ and the non-degeneracy of $(\;,\;)$ entails $\mathrm{d}A=0$.
	Hence $A$ individuates a de Rham cohomology class in $H^k_{\mathrm{t}}(M)$, {\it cf.} Appendix \ref{App: Poincare duality for manifold with boundary}.
	Furthermore, $([\alpha],[A])=0$ entails $\langle [\alpha],[A]\rangle=0$ where $\langle\;,\;\rangle$ denotes the pairing between $H_{k,\mathrm{c}}(M)$ and $H^k_{\mathrm{t}}(M)$ -- \textit{cf.} Appendix \ref{App: Poincare duality for manifold with boundary}.
	On account of Remark \ref{Rmk: consequence of Poincare--Lefschetz duality} it holds that $\langle\;,\;\rangle$ is non-degenerate and therefore $[A]=0$.
	
	\paragraph{Proof of Equation \eqref{Eq: non-redundancy for delta d-tangential bc algebra}.}
	Assume $\exists[\alpha]\in\mathcal{O}_{\mathrm{t}}(M)$ such that $([\alpha],[A])=0$ $\forall [A]\in\operatorname{Sol}_{\mathrm{t}}(M)$. Working at the level of representatives, we can consider
	$A=G_\parallel\omega$ with $\omega\in\Omega^k_{\mathrm{c},\delta}(M)$, while $\alpha\in\Omega^k_{\mathrm{c},\delta}(M)$. Hence, in view of Proposition \ref{Prop: exact sequence and duality relations}, $0=(\alpha,A)=(\alpha,G_\parallel\omega)=-(G_\parallel\alpha,\omega)$.
	Choosing $\omega=\delta\beta$, $\beta\in\Omega^{k+1}_{\mathrm{c}}(M)$ and using \eqref{Eq: boundary terms for delta and d}, it descends $(\mathrm{d}G_\parallel\alpha,\beta)=0$.
	Since $\beta$ is arbitrary and the pairing is non degenerate $\mathrm{d}G_\parallel\alpha=0$. Since $\mathrm{t}G_\parallel\alpha=0$, it turns out that $G_\parallel\alpha$ individuates an equivalence class $[G_\parallel\alpha]\in H^k_{\mathrm{t}}(M)$.
	Using the same argument of the previous item, $(G_\parallel\alpha,\beta)=0$ for all $\beta\in\Omega^k_{\mathrm{c},\delta}(M)$ entails that  $G_\parallel\alpha=\mathrm{d}\chi$ where $\chi\in\Omega^{k-1}_{\mathrm{t}}(M)$.
	Therefore $[G_\parallel\alpha]=[0]\in\operatorname{Sol}_{\mathrm{t}}(M)$: Proposition \ref{Prop: characterization of solution space in terms of test-forms} entails that $[\alpha]=[0]$.

	\paragraph{Proof for $\delta\mathrm{d}$-normal boundary conditions.}
	The fact that the pairing $([\alpha],[A])$ is well-defined for $[\alpha]\in\mathcal{O}_{\mathrm{nd}}(M)$ and $[A]\in\operatorname{Sol}_{\mathrm{nd}}(M)$ has already been discussed in Remark \ref{Rmk: bc for obs associated to deltad-normal bc}.
	It remains to discuss the proof of equations \eqref{Eq: optimality for delta d-normal bc algebra}-\eqref{Eq: non-redundancy for delta d-normal bc algebra}.

	\paragraph{Proof of Equation \eqref{Eq: optimality for delta d-normal bc algebra}.}
	Let $[A]\in\operatorname{Sol}_{\mathrm{nd}}(M)$ be such that $([\alpha],[A])=0$ for all $[\alpha]\in\mathcal{O}_{\mathrm{nd}}(M)$.
	This implies that $(\alpha,A)=0$ for all $A\in[A]$ and for all $\alpha\in\Omega_{\mathrm{c,n},\delta}^k(M)$. Taking in particular $\alpha=\delta\beta$ with $\beta\in\Omega_{\mathrm{c,n}}^k(M)$ it follows $(\mathrm{d}A,\beta)=0$.
	The non-degeneracy of $(\;,\;)$ implies $\mathrm{d}A=0$, that is $A$ defines an element in $H^k(M)$.
	The hypotheses on $A$ implies that $\langle A,[\eta]\rangle=0$ for all $[\eta]\in H_{k,\mathrm{c,n}}(M)$. The results in Appendix \ref{App: Poincare duality for manifold with boundary} -- \textit{cf.} Remark \ref{Rmk: consequence of Poincare--Lefschetz duality} -- ensure that $A=\mathrm{d}\chi$, therefore $[A]=[0]\in\operatorname{Sol}_{\mathrm{nd}}(M)$.
	
	\paragraph{Proof of Equation \eqref{Eq: non-redundancy for delta d-normal bc algebra}.}
	Let $[\alpha]\in\mathcal{O}_{\mathrm{nd}}(M)$ be such that $([\alpha],[A])=0$ for all $[A]\in\operatorname{Sol}_{\mathrm{nd}}(M)$.
	This implies in particular that, choosing $\alpha\in[\alpha]$ and $A=G_\perp\beta$ with $\beta\in\Omega_{\mathrm{c,n},\delta}^k(M)$,  $0=(\alpha,G_\perp\beta)=-(G_\perp\alpha,\beta)$.
	With the same argument of the first statement it follows that $G_\perp\alpha=\mathrm{d}\chi$ where $\chi\in\Omega^{k-1}(M)$ is such that $\mathrm{nd}\chi=0$.
	Therefore we found that $[G_\perp\alpha]=[0]\in\operatorname{Sol}_{\mathrm{nd}}(M)$: Proposition \ref{Prop: characterization of solution space in terms of test-forms} implies that $[\alpha]=[0]\in\mathcal{O}_{\mathrm{nd}}(M)$.	
\end{proof}

The following corollary translates at the level of algebra of observables the degeneracy of the presymplectic spaces discussed in Proposition \ref{Prop: presymplectomorphism for spacelike solution spaces} -- {\it cf.} Remark \ref{Rmk: on degeneracy on presymplectic structure}.
As a matter of facts, since $\widetilde{G}_\parallel$ (\textit{resp.} $\widetilde{G}_\perp$) can be degenerate, the algebra of observables $\mathcal{A}_{\mathrm{t}}(M)$ (\textit{resp.} $\mathcal{A}_{\mathrm{nd}}(M)$) will possess a non-trivial center.
In other words

\begin{corollary}\label{Cor: non-semi simple alg for gauge bc}
	If $\mathrm{d}\Omega_{\mathrm{sc,t}}^{k-1}(M)\subset\Omega_{\mathrm{sc}}^k(M)\cap\mathrm{d}\Omega_{\mathrm{t}}^{k-1}(M)$ is a strict inclusion, then the algebra $\mathcal{A}_{\mathrm{t}}(M)$ is not semi-simple.
	Similarly, if $\mathrm{d}\Omega_{\mathrm{sc}}^{k-1}(M)\subset\Omega_{\mathrm{sc}}^k(M)\cap\mathrm{d}\Omega^{k-1}(M)$ is a strict inclusion, then the algebra $\mathcal{A}_{\mathrm{nd}}(M)$ is not semi-simple.
\end{corollary}
\begin{proof}
	Since the proof is the same for either $\delta\mathrm{d}$-tangential and $\delta\mathrm{d}$-normal boundary conditions, we shall consider only the first case.

	With reference to Remark \ref{Rmk: on degeneracy on presymplectic structure}, if $\mathrm{d}\Omega_{\mathrm{sc,t}}^{k-1}(M)\subset\Omega_{\mathrm{sc}}^k(M)\cap\mathrm{d}\Omega_{\mathrm{t}}^{k-1}(M)$ is a strict inclusion then there exists an element $[A]\in\operatorname{Sol}_{\mathrm{t}}^{\mathrm{sc}}(M)$ such that $\sigma_{\mathrm{t}}([A],[B])=0$ for all $[B]\in\operatorname{Sol}_{\mathrm{t}}^{\mathrm{sc}}(M)$.
	On account of Proposition \ref{Prop: characterization of solution space in terms of test-forms} there exists $[\alpha]\in\mathcal{O}_{\mathrm{t}}(M)$ such that $[G_\parallel\alpha]=[A]$.
	Moreover, Proposition \ref{Prop: presymplectomorphism for spacelike solution spaces} ensures that $\widetilde{G}_\parallel([\alpha],[\beta])=0$ for all $[\beta]\in\mathcal{O}_{\mathrm{t}}(M)$.
	It follows from Definition \ref{Def: alg of obs for gauge bc} that $[\alpha]$ belongs to the center of $\mathcal{A}_{\mathrm{t}}(M)$, that is, $\mathcal{A}_{\mathrm{t}}(M)$ is not semi-simple.
\end{proof}

\begin{remark}
	Corollary \ref{Cor: non-semi simple alg for gauge bc} has established that the algebra of observables possesses a non trivial center. While from a mathematical viewpoint this feature might not appear of particular significance, it has far reaching consequences from the physical viewpoint. Most notably, the existence of Abelian ideals was first observed in the study of gauge theories in \cite{Dappiaggi:2011zs} leading to an obstruction in the interpretation of these models in the language of locally covariant quantum field theories as introduced in \cite{Brunetti-Fredenhagen-Verch-03}. This issue has been thoroughly studied in \cite{Benini:2013ita,Benini:2013tra,Sanders:2012sf,Wrochna-Zahn:17} turning out to be an intrinsic feature of Abelian gauge theories on globally hyperbolic spacetimes with empty boundary. Corollary \ref{Cor: non-semi simple alg for gauge bc} shows that the same conclusions can be drawn when the underlying manifold possesses a timelike boundary.
\end{remark}

\begin{remark}
	To conclude this section we observe that all algebras of observables that we have constructed obey to the so-called principle of {\em F-locality}.
	This concept was introduced for the first time in \cite{Kay:1992es} and it asserts that, given any globally hyperbolic region $\mathcal{O}\subset\mathring{M}$ the restriction to $\mathcal{O}$ of the algebra of observables built on $M$ is $*$-isomorphic to the one which one would construct intrinsically on $(\mathcal{O},g|_{\mathcal{O}})$.
	In our approach this property is implemented per construction and its proof is a direct generalization of the same argument given in \cite{Dappiaggi:2017wvj}.
	For this reason we omit the details.
\end{remark}

\appendix

\section{Existence of fundamental solutions on ultrastatic spacetimes}\label{App: Existence of fundamental solutions on ultrastatic spacetimes}

In this section we prove that Assumption \ref{Thm: assumption theorem} is verified in a large class of globally hyperbolic spacetimes $(M,g)$ with timelike boundary. These can be characterized by the following two additional hypotheses:
\begin{enumerate}
	\item $(M,g)$ is ultrastatic, that is, 
	with reference to Equation \eqref{eq:line_element}, we impose $\beta=1$ and $h_\tau=h_0$ for all $\tau\in\mathbb{R}$.
	Hence $\partial_\tau$ is a timelike Killing vector field.
	\item The Cauchy surface $(\Sigma,h_0)$ with $\partial\Sigma\neq\emptyset$ is of {\em bounded geometry}, that is there exists an $(m-1)$-dimensional {Riemannian} manifold $(\widehat{\Sigma},\widehat{h})$ of bounded geometry\footnote{Recall that a Riemannian manifold $(N,h)$ with $\partial N=\emptyset$ is called of {\em bounded geometry} if the injectivity radius $r_{\mathrm{inj}}(N)>0$ and $\|\nabla^k R\|_{L^\infty(N)}<\infty$ for all $k\in\mathbb{N}\cup\{0\}$ where $R$ is the scalar curvature while $\nabla$ is the Levi-Civita connection associated to $h$.} such that $\Sigma\subset\widehat{\Sigma}$ and $\widehat{h}|_\Sigma=h_0$.
	In addition, $\partial\Sigma$ is a smooth submanifold of bounded geometry in $\widehat{\Sigma}$.
\end{enumerate}

It is worth recalling that, whenever one considers a complex vector bundle $E$ over $(\Sigma,h_0)$ endowed with both a fiberwise Hermitian product $\langle\;,\;\rangle_E$ and a product preserving connection $\nabla^E$, one can define a suitable notion of Sobolev spaces.
Most notably, let $\Gamma_{\mathrm{me}}(E)$ denote the equivalence classes of measurable sections of $E$.
Then, for all $\ell\in\mathbb{N}\cup\{0\}$, we define
\begin{equation}\label{Eq: Sobolev space}
H^\ell(\Sigma;E)\equiv H^\ell(E)\doteq\{u\in\Gamma_{\mathrm{me}}(E)\;|\;\nabla^j u\in L^2(\Sigma;E\otimes T^*\Sigma^{\otimes j}),\;j\leq \ell\}\,,
\end{equation}
where we omitted the subscript $E$ on $\nabla$ for notational simplicity. The theory of these space has been thoroughly studied in the literature and for the case in hand we refer mainly to \cite{Grosse}.

In the following we study the existence of advanced and retarded fundamental solutions for the D'Alembert - de Rham wave operator $\Box=\mathrm{d}\delta+\delta \mathrm{d}$ acting on $k$-forms. We use a method, first employed in \cite{Dappiaggi-Drago-Ferreira-19} for the special case $k=0$, based on a functional analytic tool known as boundary triples, see for example \cite{Behrndt-Langer-12}. In order to be self-consistent, we will recall the necessary definitions and results from this paper, to which we refer for further details. The main ingredient is the following:

\begin{Definition}\label{Def: boundary triples}
	Let $\mathsf{H}$ be a separable Hilbert space over $\mathbb{C}$ and let $S:D(S)\subset\mathsf{H}\to\mathsf{H}$ be a closed, symmetric, linear operator. A {\bf boundary triple} for the adjoint operator $S^*$ is a triple $(\mathsf{h},\gamma_0,\gamma_1)$ consisting of a separable Hilbert space $\mathsf{h}$ over $\mathbb{C}$ and of two linear maps $\gamma_i:D(S^*)\to\mathsf{h}$, $i=0,1$ such that 
	$$(S^*f|f^\prime)_{\mathsf{H}}-(f|S^*f^\prime)_{\mathsf{H}}=(\gamma_1 f|\gamma_0 f^\prime)_{\mathsf{h}}-(\gamma_0 f|\gamma_1 f^\prime)_{\mathsf{h}}\,,
	\quad\forall f,f^\prime\in D(S^*)\,,$$
	In addition the map $\gamma:D(S^*)\to\mathsf{h}\times\mathsf{h}$ such that $f\mapsto (\gamma_0(f),\gamma_1(f))$ is surjective.
\end{Definition}

Boundary triples are a convenient tool to characterize the self-adjoint extensions of a large class of linear operators. Before discussing a few notable result we need to define the following additional structures \cite{Behrndt-Langer-12}.

\begin{Definition}\label{Def: Relation}
	Let $\mathsf{H}$ be a Hilbert space over $\mathbb{C}$. We call {\bf (linear) relation} over $\mathsf{H}$, a subspace $\Theta\subseteq\mathsf{H}\times\mathsf{H}$.
	{
		The domain of a linear relation $\Theta$ is $\operatorname{dom}(\Theta)\subseteq\mathsf{H}$ defined by
		\begin{align*}
			\operatorname{dom}(\Theta)
			:=\{
			f\in\mathsf{H}\,|\,
			\exists f'\in\mathsf{H}\,,\,(f,f')\in\Theta
			\}\,.
		\end{align*}
	}
	We indicate respectively with adjoint $\Theta^*$ and inverse $\Theta^{-1}$ of $\Theta$ as the sets
	\begin{align}\label{Eq: adjoint relation}
		\Theta^*
		\doteq\{(f,f^\prime)\in\mathsf{H}\times\mathsf{H}\;|\;(f^\prime,g)=(f,g^\prime)\;\forall (g,g^\prime)\in\Theta\}\,,
	\end{align}
	and
	$$\Theta^{-1}\doteq\{(f,f^\prime)\in\mathsf{H}\times\mathsf{H}\;|\;(f^\prime,f)\in\Theta\}.$$
	Consequently we say that $\Theta$ is self-adjoint if $\Theta=\Theta^*$
\end{Definition}

\noindent The proof of the following proposition can be found in \cite{Malamud}.

\begin{proposition}\label{Prop: Self-Adjoint Extensions}
	Let $S:D(S)\subseteq\mathsf{H}\to\mathsf{H}$ be a closed, symmetric operator.
	Then $S$ admits a boundary triple $(\mathsf{h},\gamma_0,\gamma_1)$ if and only if it admits self-adjoint extensions.
	If $\Theta$ is a closed, densely defined linear relation {on $\mathsf{h}$}, then $S_\Theta\doteq S^*|_{\ker(\gamma_1-\Theta\gamma_0)}$ is a closed extension of $S$
	where $\ker(\gamma_1-\Theta\gamma_0)\doteq\{\psi\in{\mathsf{H}}\;|\;(\gamma_0\psi,\gamma_1\psi)\in\Theta\}$.
	In addition the map $\Theta\to S_\Theta$ is one-to-one and $S^*_\Theta=S_{\Theta^*}$.
	Hence there is a one-to-one correspondence between self-adjoint relations $\Theta$ and self-adjoint extensions of $S$.
\end{proposition}

In order to apply these tools to the case in hand, first of all we need to recall that our goal is that of constructing advanced and retarded fundamental solutions for the D'Alembert de-Rham wave operator $\Box$ acting on $k$-forms. In other words, calling as $\Lambda^k T^*\mathring{M}$ the $k$-th exterior power of the cotangent bundle over $\mathring{M}$, $k\geq 1$, and with $\boxtimes$ the external tensor product, we look for $G^\pm\in\Gamma_{\mathrm{c}}(\Lambda^kT^* \mathring{M}\boxtimes\Lambda^k T^*\mathring{M})^\prime$ such that 
$$\Box\circ G^\pm = G^\pm\circ\Box= \operatorname{Id}|_{\Gamma_{\mathrm{c}}(\Lambda^kT^* \mathring{M}\boxtimes\Lambda^k T^*\mathring{M})}\,,$$
while $\textrm{supp}(G^\pm(\omega))\subseteq J^\pm(\textrm{supp}(\omega))$ for all $\omega\in\Gamma_{\mathrm{c}}(\Lambda^k T^*\mathring{M})$ -- \textit{cf.} Assumption \ref{Thm: assumption theorem}.
Working at the level of integral kernels and setting $G^\pm(\tau-\tau^\prime,x,x^\prime)=\theta[\pm(\tau-\tau^\prime)]G(\tau-\tau^\prime,x,x^\prime)$, this amounts to solving the following distributional, initial value problem 
\begin{equation}\label{Eq: system for G}
\left(\Box\otimes\mathbb{I}\right) G = \left(\mathbb{I}\otimes\Box\right) G = 0\,,\qquad
G|_{\tau=\tau^\prime}=0\,,\qquad
\partial_\tau G|_{\tau=\tau^\prime}=\delta_{\textrm{diag}(\mathring{M})}\,.
\end{equation}
where $\delta_{\textrm{diag}(\mathring{M})}$ stands for the bi-distribution yielding $\delta_{\textrm{diag}(\mathring{M})}(\omega_1\boxtimes\omega_2)=(\omega_1,\omega_2)$ for all $\omega_1,\omega_2\in\Gamma_{\mathrm{c}}(\Lambda^kT^*\mathring{M})$. Since we have assumed that the underlying spacetime $(M,g)$ is ultrastatic, Equation \eqref{eq:line_element} entails that, up to a global irrelevant sign depending on the convention used for the metric signature, \cite{Pfenning:2009nx} 
$$\Box=\partial_\tau^2+S,$$
where $S$ is a uniformly elliptic operator whose local form can be found in \cite{Pfenning:2009nx}. This entails that, in order to construct solutions of \eqref{Eq: system for G}, we can follow the rationale outlined in \cite{Dappiaggi-Drago-Ferreira-19}. 

To this end we start by focusing our attention on $S$ analysing it within the framework of boundary triples.
Our first observation consists of noticing, that being $(M,g)$ globally hyperbolic, Theorem \ref{Thm: globally hyperbolic spacetime with time-like boundary} ensures that $M$ is diffeomorphic to $\mathbb{R}\times\Sigma$.
Leaving implicit the identification $M\simeq\mathbb{R}\times\Sigma$ and recalling Theorem \ref{Thm: globally hyperbolic spacetime with time-like boundary}, let us indicate with $\iota_\tau\colon\Sigma\to M$ the (smooth one-parameter group of) embedding maps which realizes $\Sigma$ at time $\tau$ as $\iota_\tau\Sigma=\lbrace \tau\rbrace\times\Sigma\doteq\Sigma_\tau$. It holds $\Sigma_\tau\simeq\Sigma_{\tau^\prime}$ for all $\tau,\tau^\prime\in\mathbb{R}$.
	If follows that, on account of Theorem \ref{Thm: globally hyperbolic spacetime with time-like boundary}, for all $\omega\in\Omega^k(M)$ and $\tau\in\mathbb{R}$, $\omega|_{\Sigma_\tau}\in \Gamma(\iota_\tau^*\Lambda^kT^*M)$.
	Here $\iota_\tau^*(\Lambda^kT^*M)$ denotes the pull-back bundle over $\Sigma_\tau\simeq\Sigma$ built out of $\Lambda^kT^*M$ via $\iota_\tau$ -- \textit{cf.} \cite{Husemoller-94}.
	Moreover, recalling Definition \ref{Def: tangential and normal component}, it holds that $\omega|_{\Sigma_\tau}$ can be further decomposed as
	\begin{align*}
	\omega|_{\Sigma_\tau}:=
	(\star_{\Sigma_\tau}^{-1}\iota_\tau^*\star_M)\omega\wedge\mathrm{d}\tau+
	\iota_\tau^*\omega=
	\mathrm{n}_{\Sigma_\tau}\omega\wedge\mathrm{d}\tau+
	\mathrm{t}_{\Sigma_\tau}\omega\,.
	\end{align*}
	where $\mathrm{t}_{\Sigma_\tau}\omega\in\Omega^k(\Sigma_\tau)$ while $\mathrm{n}_{\Sigma_\tau}\omega\in\Omega^{k-1}(\Sigma_\tau)$ -- \textit{cf.} Definition \ref{Def: tangential and normal component}.
	Barring the identification between $\Sigma_\tau$ and $\Sigma_{\tau^\prime}$ the latter decomposition induces the isomorphisms
	\begin{align}
	&\Gamma(\iota_\tau^*\Lambda^kT^*M)\simeq \Omega^{k-1}(\Sigma)\oplus \Omega^k(\Sigma)\,,\qquad
	\omega\to(\omega_0\oplus\omega_1)\,\\
	\label{Eq: identification isomorphism for k-forms on ultrastatic spacetimes}
	&\Omega^k(M)\to C^\infty(\mathbb{R},\Omega^{k-1}(\Sigma))\oplus C^\infty(\mathbb{R},\Omega^k(\Sigma))\,,\qquad
	\omega\to(\tau\mapsto\mathrm{t}_{\Sigma_\tau}\omega)\oplus(\tau\mapsto\mathrm{n}_{\Sigma_\tau}\omega)\,.
	\end{align}
	Furthermore a direct computation shows that, for all $\omega\in\Omega^k(M)$, it holds that
	\begin{align*}
	S\omega|_{\Sigma_\tau}=(-\Delta_{k-1}\mathrm{t}_{\Sigma_\tau}\omega)\wedge \mathrm{d}\tau-\Delta_k\mathrm{n}_{\Sigma_\tau}\omega\,,
	\end{align*}
	where $\Delta_k$ is the Laplace-Beltrami operator acting on $k$-forms, built out of $h_0$.

Putting together all these data and working in the language of Definition \ref{Def: boundary triples}, we can consider the following building blocks:
\begin{enumerate}
	\item  As Hilbert space we set 
	\begin{align*}
	\mathsf{H}\equiv L^2(\Omega^{k-1}(\Sigma))\oplus L^2(\Omega^k(\Sigma))\,,
	\end{align*}
	where $L^2(\Omega^k(\Sigma))$ is the closure of $\Omega^k_{\mathrm{c}}(\Sigma)$ with respect to the pairing $(\;,\;)_\Sigma$ between $k$-forms, {\it i.e.}, $(\alpha,\beta)_{\Sigma}=\int\limits_\Sigma \alpha\wedge *_\Sigma\beta$ for all $\alpha,\beta\in\Omega^k_{\mathrm{c}}(\Sigma)$.	
	\item We identify with a slight abuse of notation $S$ with $(-\Delta_{k-1})\oplus(-\Delta_k)$ where  $\Delta_k$ is the Laplace-Beltrami operator built out of $h_0$ acting on $k$-forms.
\end{enumerate}
Observe that $S$ can be regarded as an Hermitian and densely defined operator on $H^2_0(\Lambda^{k-1}T^*\Sigma)\oplus H^2_0(\Lambda^kT^*\Sigma)$ where $H^2_0(\Lambda^kT^*\Sigma)$ is the closure of $\Omega^k_{\mathrm{c}}(\mathring{\Sigma})$ with respect to the $H^2(\Lambda^kT^*\Sigma)$-norm -- \textit{cf.} Equation \eqref{Eq: Sobolev space} with $E\equiv\Lambda^kT^*\Sigma$.
In this case both the inner product and the connection are those induced from the underlying metric $h_0$.
{
	Standard arguments entail that $S$ is a closed symmetric operator on $\mathsf{H}$ whose adjoint $S^*$ is defined on the maximal domain $D(S^*)\doteq\{(\omega_0\oplus\omega_1)\in\mathsf{H}\;|\;S(\omega_0\oplus\omega_1)\in\mathsf{H}\}$, with $S^*(\omega_0\oplus\omega_1)=S(\omega_0\oplus\omega_1)$ for all $\omega_0\oplus\omega_1\in D(S^*)$.
	In addition the deficiency indices of $S^*$ coincide. Therefore $S$ admits self-adjoint extensions, which can be described as per Proposition \ref{Prop: Self-Adjoint Extensions}. 
}
In order to realize explicitly {a boundary triple for $S^*$}, we start by observing that, since $\partial\Sigma\neq\emptyset$, we can introduce the standard trace map between Sobolev spaces, {\it i.e.}, for every $\ell\geq\frac{1}{2}$ there exists a continuous surjective map $\mathrm{res}_\ell:H^\ell(\Lambda^kT^*\Sigma)\to H^{\ell-\frac{1}{2}}(\iota_{\partial\Sigma}^*\Lambda^kT^*\Sigma)$ whose action on $\Omega^k_{\mathrm{c}}(\Sigma)$ coincides with the restriction to $\partial\Sigma$ for every $\ell$ {-- here $\iota_{\partial\Sigma}^*\Lambda^kT^*\Sigma$ denotes the pull-back bundle with respect to the inclusion $\iota_{\partial\Sigma}\colon\partial\Sigma\to\Sigma$.}
This last property allows us to better characterize the action of the restriction map, since, for every $\alpha\in\Omega^k_{\mathrm{c}}(\Sigma)$ a straightforward computations shows that, for all $\ell\geq\frac{1}{2}$ 
$$\alpha|_{\partial\Sigma}=\mathrm{res}_\ell\alpha=\alpha_0+\alpha_1\wedge \mathrm{d}x\,,$$
where, up to an irrelevant isomorphism, we can identify $\alpha_0\equiv\mathrm{t}_{\partial\Sigma}\alpha$ and $\alpha_1\equiv\mathrm{n}_{\partial\Sigma}\alpha$ -- \textit{cf.} Definition \ref{Def: tangential and normal component}.
Here, for every $p\in\partial\Sigma$, $\mathrm{d}x$ is the basis element of $T^*_pM$ such that $\mathrm{d}x(\nu_p)=1$ {while $\mathrm{d}x(X)|_p=0$ for all smooth vector fields $X\in\Gamma(T\Sigma)$ tangent to $\partial\Sigma$} -- here $\nu_p$ is the outward pointing, unit vector normal to $\partial\Sigma$ at $p$.
With this observation in mind and following mutatis mutandis the same analysis of \cite{Dappiaggi-Drago-Ferreira-19} for the scalar case, we can construct the following boundary triple for $S^*$
\begin{itemize}
	\item
	$\mathsf{h}=\mathsf{h}_0\oplus\mathsf{h}_1$ where {$\mathsf{h}_0\doteq L^2(\Omega^{k-2}(\partial\Sigma))\oplus L^2(\Omega^{k-1}(\partial\Sigma))$} while $\mathsf{h}_1=L^2(\Omega^{k-1}(\partial\Sigma))\oplus L^2(\Omega^{k}(\partial\Sigma))$; 
	\item
	the map $\gamma_0:D(S^*)\to\mathsf{h}$ such that, for all $\omega_0\oplus\omega_1\in D(S^*)$, 
	\begin{equation}\label{Eq: gamma0}
	\gamma_0(\omega_0\oplus\omega_1)= (\mathrm{n}_{\partial\Sigma}\omega_0\oplus
	\mathrm{t}_{\partial\Sigma}\omega_0)\oplus
	(\mathrm{n}_{\partial\Sigma}\omega_1\oplus
	\mathrm{t}_{\partial\Sigma}\omega_1)\,.
	\end{equation}
	\item
	the map $\gamma_1:D(S^*)\to\mathsf{h}$ such that, for all $\omega_0\oplus\omega_1\in D(S^*)$, 
	\begin{equation}\label{Eq: gamma1}
	\gamma_1(\omega_0\oplus\omega_1)=
	(\mathrm{t}_{\partial\Sigma}\delta_\Sigma\omega_0\oplus
	\mathrm{n}_{\partial\Sigma}\mathrm{d}_{\Sigma}\omega_0)\oplus
	(\mathrm{t}_{\partial\Sigma}\delta_\Sigma\omega_1\oplus
	\mathrm{n}_{\partial\Sigma}\mathrm{d}_{\Sigma}\omega_1)\,,
	\end{equation}
	where with a slight abuse of notation we denote still with $\mathrm{d}_\Sigma$ and $\delta_\Sigma$ the extension to the space of square-integrable $k$-forms of the action of the differential and of the codifferential on $\Omega^k_{\mathrm{c}}(\Sigma)$.
\end{itemize}

In view of Proposition \ref{Prop: Self-Adjoint Extensions} we can follow slavishly the proof of \cite[Th. 30]{Dappiaggi-Drago-Ferreira-19} to infer the following statement:

\begin{theorem}\label{Thm: existence of propagators}
	Let $(M,g)$ be an ultrastatic and globally hyperbolic spacetime with timelike boundary {and of bounded geometry}.
	Let $(\mathsf{h},\gamma_0,\gamma_1)$ be the boundary triple built as per Equation \eqref{Eq: gamma0} and \eqref{Eq: gamma1} associated to the operator $S^*$. Let $\Theta$ be a self-adjoint relation on $\mathsf{h}$ as per Definition \ref{Def: Relation} and let $S_\Theta\doteq S^*|_{D(S_\Theta)}$ where $D(S_\Theta)=\ker(\gamma_1-\Theta\gamma_0)$, {\it cf.} Proposition \ref{Prop: Self-Adjoint Extensions}.
	If the spectrum of $S_\Theta$ is bounded from below, then there exists unique advanced and retarded Green's operator $G^\pm_\Theta$ associated to $\partial_\tau^2+S_\Theta$.
	They are completely determined in terms of the bidistributions $G^\pm_\Theta=\theta[\pm(\tau-\tau^\prime)]G_\Theta$ where $G_\Theta\in\Gamma_{\mathrm{c}}(\Lambda^kT^*\mathring{M}\boxtimes\Lambda^kT^*\mathring{M})'$ is such that for $\omega_1,\omega_2\in\Gamma_{\mathrm{c}}(\Lambda^kT^*\mathring{M})$,
	\begin{align*}
	G_\Theta(\omega_1,\omega_2)=
	\int_{\mathbb{R}^2}
	\left(\omega_1|_{\Sigma},S^{-\frac{1}{2}}_{k,\Theta}\sin(S^{\frac{1}{2}}_{k,\Theta}(\tau-\tau^\prime))\omega_2|_{\Sigma}\right)_{\Sigma}
	\mathrm{d}\tau\mathrm{d}\tau'\,,
	\end{align*}
	where $(\;,\;)_\Sigma$ stands for the pairing between $k$-forms and where $\omega_2$ identifies an element in $D(S_\Theta)$ via the identifications \eqref{Eq: identification isomorphism for k-forms on ultrastatic spacetimes}.
	Moreover it holds that 
	\begin{align}
	G^\pm_\Theta\omega\in\ker(\gamma_1-\Theta\gamma_0)\,,\quad
	\forall\omega\in\Gamma_{\mathrm{c}}(\Lambda^kT^*\mathring{M})\,.
	\end{align}
\end{theorem}

The last step consists of proving that the boundary conditions introduced in Definition \ref{Def: Dirichlet, Box-tangential, Box-normal, Robin Box-tangential, Robin Box-normal boundary conditions} fall within the class considered in Theorem \ref{Thm: existence of propagators}. In the following proposition we adopt for simplicity the notation $\mathrm{t}=\mathrm{t}_{\partial\Sigma}$, $\mathrm{n}=\mathrm{n}_{\partial\Sigma}$, $\mathrm{nd}=\mathrm{n}_{\partial\Sigma}\mathrm{d}_{\partial\Sigma}$, $\mathrm{t}\delta=\mathrm{t}_{\partial\Sigma}\delta_\Sigma$.

\begin{proposition}\label{Prop: self-adjoint relation for parallel-, perp- and f,0- boundary conditions}
	The following relations on $\mathsf{h}$ are selfadjoint:
	\begin{align}
	\label{Eq: parallel-relation}
	\Theta_\parallel&\doteq\lbrace
	(\mathrm{n}\omega_0\oplus
	0\oplus
	\mathrm{n}\omega_1\oplus
	0
	\;;\;
	0\oplus
	\mathrm{nd}\omega_0\oplus
	0\oplus
	\mathrm{nd}\omega_1
	)\;|\;\omega_0\oplus\omega_1\in D(S^*)\rbrace
	\\
	\label{Eq: perp-relation}
	\Theta_\perp&\doteq\lbrace
	(0\oplus
	\mathrm{t}\omega_0\oplus
	0\oplus
	\mathrm{t}\omega_1
	\;;\;
	\mathrm{t}\delta\omega_0\oplus
	0\oplus
	\mathrm{t}\delta\omega_1\oplus
	0
	)\;|\;\omega_0\oplus\omega_1\in D(S^*)\rbrace
	\\
	\label{Eq: f-relation}
	\Theta_{f_\parallel}&\doteq\lbrace
	(\mathrm{n}\omega_0\oplus
	0\oplus
	\mathrm{n}\omega_1\oplus
	0
	\;;\;
	f\mathrm{n}\omega_0\oplus
	\mathrm{nd}\omega_0\oplus
	f\mathrm{n}\omega_1\oplus
	\mathrm{nd}\omega_1
	)\;|\;\omega_0\oplus\omega_1\in D(S^*)\rbrace,
	\quad f\in C^\infty(\partial\Sigma)\,, f\geq 0\,.\\
	\Theta_{f_\perp}&\doteq\lbrace
	(0\oplus
	\mathrm{t}\omega_0\oplus
	0\oplus
	\mathrm{t}\omega_1
	\;;\;
	\mathrm{t}\delta\omega_0\oplus
	f\mathrm{t}\omega_0\oplus
	\mathrm{t}\delta\omega_1\oplus
	f\mathrm{t}\omega_1
	)\;|\;\omega_0\oplus\omega_1\in D(S^*)\rbrace,
	\quad f\in C^\infty(\partial\Sigma)\,, f\leq 0\,.
	\end{align}
	Moreover the self-adjoint extension $S_{\Theta_\sharp}$ for $\sharp\in\lbrace\parallel,\perp,f_\parallel,f_\perp\rbrace$ abides to the hypotheses of Theorem \ref{Thm: existence of propagators}. The associated propagators $G_\sharp$, $\sharp\in\lbrace\parallel,\perp,f_\parallel,f_\perp\rbrace$, obey the boundary conditions as per Definition \ref{Def: Dirichlet, Box-tangential, Box-normal, Robin Box-tangential, Robin Box-normal boundary conditions}.
\end{proposition}
\begin{proof} We show that $\Theta_\parallel,\Theta_\perp,\Theta_{f_\parallel},\Theta_{f_\perp}$ are self-adjoint relations as per Definition \ref{Def: Relation}. Since the proofs for the different cases are very similar we shall consider only $\Theta_\parallel$.
	A short computation shows that $\Theta_\parallel\subseteq\Theta_\parallel^*$. We prove the converse inclusion.
	Let $\underline{\alpha}:=(\alpha_1\oplus\ldots\alpha_4\,;\,\alpha_5\oplus\ldots\alpha_8)\in\Theta_\parallel^*$.
	Considering equation {\eqref{Eq: adjoint relation}} we find
	\begin{align}
		(\mathrm{n}\omega_0,\alpha_5)+
		(\mathrm{n}\omega_1,\alpha_7)=
		(\mathrm{nd}\omega_0,\alpha_2)+
		(\alpha_4,\mathrm{nd}\omega_1,\alpha_4)\,,\qquad
		\forall\omega_0\oplus\omega_1\in D(S^*)\,.
	\end{align}
	Choosing $\omega_1$ and $\mathrm{n}\omega_0=0$ -- this does not affect the value $\mathrm{nd}\omega_0$ on account of Remark \ref{Rmk: surjectivity of t,n,tdelta,nd} -- it follows that $(\alpha_2,\mathrm{nd}\omega_0)=0$ for all $\omega_0\in\Omega_{\mathrm{c,n}}^{k-1}(\Sigma)$.
	Since $\mathrm{nd}$ is surjective it follows that $\alpha_2=0$.
	With a similar argument $\alpha_5=0$ as well as $\alpha_2=0$, $\alpha_4=0$.
	Finally, on account of Remark \ref{Rmk: surjectivity of t,n,tdelta,nd} there exists $\omega_0\oplus\omega_1\in D(S^*)$ such that
	\begin{align*}
		\mathrm{n}\omega_0=\alpha_1\,,\qquad
		\mathrm{n}\omega_1=\alpha_3\,,\qquad
		\mathrm{nd}\omega_0=\alpha_6\,,\qquad
		\mathrm{nd}\omega_1=\alpha_8\,.
	\end{align*}
	It follows that $\alpha\in\Theta_\parallel$, that is, $\Theta_\parallel=\Theta_\parallel^*$.
	
	In addition $S_{\Theta_\sharp}$ is positive definite for $\sharp\in\lbrace\parallel,\perp,f_\parallel,f_\perp\rbrace$. It descends from the following equality, which holds for all $\omega_0\otimes\omega_1\in D(S^*)$:
	\begin{align*}
		(\omega_0\oplus\omega_1,S_{\Theta_\sharp}(\omega_0\oplus\omega_1))_{\mathsf{H}}=
		\sum_{j=1}^2\big[
		\|\mathrm{d}\omega_i\|^2+
		\|\delta\omega_i\|^2+
		(\mathrm{n}\omega_i,\mathrm{t}\delta\omega_i)-
		(\mathrm{t}\omega_i,\mathrm{nd}\omega_i)
		\big]\,,
	\end{align*}
	where the last two terms are non-negative because of the boundary conditions and of the hypothesis on the sign of $f$.
	Therefore we can apply Theorem \ref{Thm: existence of propagators}.
	
	Finally we should prove that the propagators $G^\pm_{\Theta_\sharp}$ associated with the relations $\Theta_\sharp$ coincide with the propagators $G^\pm_\sharp$ introduced in Theorem \ref{Thm: assumption theorem}.
	The fulfilment of the appropriate boundary conditions is a consequence of Lemma \ref{Lem: equivalence between M-boundary conditions and Sigma-boundary conditions}.
\end{proof}

\begin{remark}
	It is worth mentioning that, although we have only considered test sections of compact support in $\mathring{M}$, such assumption can be relaxed allowing the support to intersect $\partial M$.
	In order to prove that this operation is legitimate, a rather natural strategy consists of realizing that the boundary conditions here considered fall in the (generalization of) those of Robin type. These were considered in \cite{Gannot:2018jkg} for the case of a real scalar field on an asymptotically anti de Sitter spacetime where, in between many results, it was proven the explicit form of the wavefront set of the advanced and retarded fundamental solutions. In particular it was shown that two points lie in the wave front set either if they are connected directly by a light geodesic or by one which is reflected at the boundary. A direct inspection of their approach suggests that the same result holds true if one considers also static globally hyperbolic spacetimes with timelike boundary and vector valued fields. A detailed proof of this statement would require a lengthy paper on its own and thus this question will be addressed explicitly in a future work.
\end{remark}

\section{An explicit decomposition}\label{App: an explicit decomposition}

\begin{lemma}\label{Lem: equivalence between M-boundary conditions and Sigma-boundary conditions}
	Let $M=\mathbb{R}\times\Sigma$ be a globally hyperbolic spacetime -- \textit{cf.} Theorem \ref{Thm: globally hyperbolic spacetime with time-like boundary}.
	Moreover, for all $\tau\in\mathbb{R}$, let $\mathrm{t}_{\Sigma_\tau}\colon\Omega^k(M)\to\Omega^k(\Sigma_\tau)$, $\mathrm{n}_{\Sigma_\tau}\colon\Omega^k(\Sigma_\tau)\to\Omega^{k-1}(\Sigma)$ be the tangential and normal maps on $\Sigma_\tau\doteq\{\tau\}\times\Sigma$, where $M=\mathbb{R}\times\Sigma$ -- \textit{cf.} Definition \ref{Def: tangential and normal component}.
	Moreover, let $\mathrm{t}_{\partial\Sigma_\tau}\colon\Omega^k(\Sigma_\tau)\to\Omega^k(\partial\Sigma_\tau)$ and let $\mathrm{n}_{\partial\Sigma_\tau}\colon\Omega^k(\Sigma_\tau)\to\Omega^{k-1}(\partial\Sigma_\tau)$ be the tangential and normal maps on $\partial\Sigma_\tau\doteq\{\tau\}\times\partial\Sigma$.
	Let $f\in C^\infty(\partial\Sigma)$ and set $f_\tau\doteq f|_{\partial\Sigma_\tau}$ .
	Then for $\sharp\in\lbrace\mathrm{D},\parallel,\perp,f_\parallel,f_\perp\rbrace$ it holds
	\begin{align}\label{Eq: equivalence between M-boundary conditions and Sigma-boundary conditions}
		\omega\in\Omega_\sharp^k(M)\Longleftrightarrow
		\mathrm{t}_{\Sigma_\tau}\omega\in\Omega_\sharp^k(\Sigma_\tau)\,,
		\mathrm{n}_{\Sigma_\tau}\omega\in\Omega_\sharp^{k-1}(\Sigma_\tau)\,,
		\quad\forall \tau\in\mathbb{R}\,.
	\end{align}
	More precisely this entails that 
	\begin{align*}
		\omega\in\ker\mathrm{t}_{\partial M}\cap\ker\mathrm{n}_{\partial M}&\Longleftrightarrow
		\mathrm{t}_{\Sigma_\tau}\omega,\mathrm{n}_{\Sigma_\tau}\omega\in
		\ker\mathrm{t}_{\partial\Sigma_\tau}\cap\ker\mathrm{n}_{\partial\Sigma_\tau}\,,\forall \tau\in\mathbb{R}\,;\\
		\omega\in\ker\mathrm{n}_{\partial M}\cap\ker\mathrm{n}_{\partial M}\mathrm{d}&\Longleftrightarrow
		\mathrm{t}_{\Sigma_\tau}\omega,\mathrm{n}_{\Sigma_\tau}\omega\in
		\ker\mathrm{n}_{\partial\Sigma_\tau}\cap\ker\mathrm{n}_{\partial\Sigma_\tau}\mathrm{d}_{\Sigma_\tau}\,,\forall \tau\in\mathbb{R}\,;\\
		\omega\in\ker\mathrm{t}_{\partial M}\cap\ker\mathrm{t}_{\partial M}\delta&\Longleftrightarrow
		\mathrm{t}_{\Sigma_\tau}\omega,\mathrm{n}_{\Sigma_\tau}\omega\in
		\ker\mathrm{t}_{\partial\Sigma_\tau}\cap\ker\mathrm{t}_{\partial\Sigma_\tau}\delta_{\Sigma_\tau}\,,\forall \tau\in\mathbb{R}\,;\\
		\omega\in\ker\mathrm{n}_{\partial M}\cap\ker(\mathrm{n}_{\partial M}\mathrm{d}-f\mathrm{t}_{\partial M})&\Longleftrightarrow
		\mathrm{t}_{\Sigma_\tau}\omega,\mathrm{n}_{\Sigma_\tau}\omega\in
		\ker\mathrm{n}_{\partial\Sigma_\tau}\cap\ker(\mathrm{n}_{\partial\Sigma_\tau}\mathrm{d}_{\Sigma_\tau}-f_t\mathrm{t}_{\partial\Sigma_\tau})\,,\forall \tau\in\mathbb{R}\,;\\
		\omega\in\ker\mathrm{t}_{\partial M}\cap\ker(\mathrm{t}_{\partial M}\delta-f\mathrm{n}_{\partial M})&\Longleftrightarrow
		\mathrm{t}_{\Sigma_\tau}\omega,\mathrm{n}_{\Sigma_\tau}\omega\in
		\ker\mathrm{t}_{\partial\Sigma_\tau}\cap\ker(\mathrm{t}_{\partial\Sigma_\tau}\delta_{\Sigma_\tau}-f_t\mathrm{n}_{\partial\Sigma_\tau})\,,\forall t \in\mathbb{R}\,.
	\end{align*}
\end{lemma}
\begin{proof}
	The equivalence \eqref{Eq: equivalence between M-boundary conditions and Sigma-boundary conditions} is shown for $\perp$-boundary condition.
	The proof for $\parallel$-boundary conditions follows per duality -- \textit{cf.} \eqref{Rmk: on nomenclature for Dirichlet and Neumann boundary conditions} -- while the one for $\mathrm{D}$-, $f_\parallel$-, $f_\perp$-boundary conditions can be carried out in a similar way.
	
	On account of Theorem \ref{Thm: globally hyperbolic spacetime with time-like boundary} it holds that, for all $\tau\in\mathbb{R}$, we can decompose any $\omega\in\Omega^k(M)$ as follows:
	\begin{align*}
		\omega|_{\Sigma_\tau}=
		\mathrm{t}_{\Sigma_\tau}\omega+
		\mathrm{n}_{\Sigma_\tau}\omega\wedge\beta^{\frac{1}{2}}\mathrm{d}\tau\,.
	\end{align*}
	Notice that, being $M$ isometric to $\mathbb{R}\times\Sigma$, it holds that $\tau\to\mathrm{t}_{\Sigma_\tau}\omega \in C^\infty(\mathbb{R},\Omega^k(\Sigma))$ while $\tau\to\mathrm{n}_{\Sigma_\tau}\omega \in C^\infty(\mathbb{R},\Omega^{k-1}(\Sigma))$. Here we have implicitly identified $\Sigma\simeq\Sigma_\tau$.
	
	A similar decomposition holds near the boundary of $\Sigma_\tau$.
	Indeed for all relatively compact open neighbourhood $U_{\partial\Sigma}\subseteq\partial\Sigma$ of $\partial\Sigma$, we consider a neighbourhood $U\subseteq\Sigma$ of the form $U=[0,\epsilon_\tau)\times U_{\partial\Sigma}$ built our of the exponential map of $M$.
	Let $U_x\doteq\{x\}\times U_{\partial\Sigma}$ for $x\in [0,\epsilon_\tau)$ and let $\mathrm{t}_{U_x}$, $\mathrm{n}_{U_x}$ be the corresponding tangential and normal maps -- \textit{cf.} Definition \ref{Def: tangential and normal component}.
	With this definition we can always split $\mathrm{t}_{\Sigma_\tau}\omega$ and $\mathrm{n}_{\Sigma_\tau}\omega$ as follows:
	\begin{align}
		\omega|_{\{\tau\}\times U_x}=
		\mathrm{t}_{U_x}\mathrm{t}_{\Sigma_\tau}\omega+
		\mathrm{n}_{U_x}\mathrm{t}_{\Sigma_\tau}\omega\wedge N^{\frac{1}{2}}\mathrm{d}x+
		\mathrm{t}_{U_x}\mathrm{n}_{\Sigma_\tau}\omega\wedge\beta^{\frac{1}{2}}\mathrm{d}\tau+
		\mathrm{n}_{U_x}\mathrm{n}_{\Sigma_\tau}\omega\wedge N^{\frac{1}{2}}\mathrm{d}x\wedge\beta^{\frac{1}{2}}\mathrm{d}\tau\,,
	\end{align}
	where $N=g(\partial_x,\partial_x)$.
	Since $U_{\partial\Sigma}$ is relatively compact
	it follows that $(\tau,x)\to\mathrm{t}_{U_x}\mathrm{t}_{\Sigma_\tau}\omega\in C^\infty(\mathbb{R}\times[0,\epsilon),\Omega^k(\partial\Sigma))$ and similarly $\mathrm{t}_{U_x}\mathrm{n}_{\Sigma_\tau}\omega$, $\mathrm{n}_{U_x}\mathrm{t}_{\Sigma_\tau}\omega$ and $\mathrm{n}_{U_x}\mathrm{n}_{\Sigma_\tau}\omega$. Once again we have implicitly identified $U_{\partial\Sigma}\simeq\{x\}\times U_{\partial\Sigma}=U_x$. According to this splitting it holds
	\begin{align*}
		\mathrm{t}_{\partial M}\omega|_{\{\tau\}\times U_{\partial\Sigma}}&
		=\mathrm{t}_{U_x}\mathrm{t}_{\Sigma_\tau}\omega|_{x=0}
		+\mathrm{t}_{U_x}\mathrm{n}_{\Sigma_\tau}\omega|_{x=0}\wedge\beta^{\frac{1}{2}}|_{\partial M}\mathrm{d}\tau
		=\mathrm{t}_{\partial\Sigma_\tau}\mathrm{t}_{\Sigma_\tau}\omega
		+\mathrm{t}_{\partial\Sigma_\tau}\mathrm{n}_{\Sigma_\tau}\omega\wedge\beta^{\frac{1}{2}}|_{\partial M}\mathrm{d}\tau\,,\\
		\mathrm{n}_{\partial M}\omega|_{\{\tau\}\times U_{\partial\Sigma}}&
		=\mathrm{n}_{U_x}\mathrm{t}_{\Sigma_{\tau}}\omega|_{x=0}
		+\mathrm{n}_{U_x}\mathrm{n}_{\Sigma_\tau}\omega|_{x=0}\wedge\beta^{\frac{1}{2}}|_{\partial M}\mathrm{d}\tau
		=\mathrm{n}_{\partial\Sigma_\tau}\mathrm{t}_{\Sigma_\tau}\omega
		+\mathrm{n}_{\partial\Sigma_\tau}\mathrm{n}_{\Sigma_\tau}\omega\wedge\beta^{\frac{1}{2}}|_{\partial M}\mathrm{d}\tau\,.
	\end{align*}
	It follows that
	\begin{align}
		\label{Eq: n omega vanishes iff n t omega and nn omega do}
		&\mathrm{n}_{\partial M}\omega=0
		\Longleftrightarrow
		\mathrm{n}_{\partial\Sigma_\tau}\mathrm{n}_{\Sigma_\tau}\omega=0\,,\quad
		\mathrm{n}_{\partial\Sigma_\tau}\mathrm{t}_{\Sigma_\tau}\omega=0\,,\\
		\nonumber
		&\mathrm{t}_{\partial M}\omega=0
		\Longleftrightarrow
		\mathrm{t}_{\partial\Sigma_\tau}\mathrm{n}_{\Sigma_\tau}\omega=0\,,\quad
		\mathrm{t}_{\partial\Sigma_\tau}\mathrm{t}_{\Sigma_\tau}\omega=0\,.
	\end{align}
	This proves the statement for Dirichlet boundary conditions.
	Assuming now $\mathrm{n}_{\partial M}\omega=0$, a similar computation yields
	\begin{align*}
		\mathrm{nd}\omega|_{\{\tau\}\times U_{\partial\Sigma}}
		&=N^{-\frac{1}{2}}\partial_x\mathrm{t}_{U_x}\mathrm{t}_{\Sigma_\tau}\omega|_{x=0}
		+N^{-\frac{1}{2}}\mathrm{d}_{\partial\Sigma}(N^{\frac{1}{2}}\mathrm{n}_{U_x}\mathrm{t}_{\Sigma_\tau}\omega)|_{x=0}
		\\&+N^{-\frac{1}{2}}\partial_x(\beta^{\frac{1}{2}}\mathrm{t}_{U_x}\mathrm{n}_{\Sigma_\tau}\omega)|_{x=0}\wedge\mathrm{d}\tau
		+N^{-\frac{1}{2}}\mathrm{d}_{\partial\Sigma}(N^{\frac{1}{2}}\beta^{\frac{1}{2}}\mathrm{n}_{U_x}\mathrm{n}_{\Sigma_\tau}\omega)|_{x=0}\wedge\mathrm{d}\tau
		\\&=N^{-\frac{1}{2}}\partial_x\mathrm{t}_{U_x}\mathrm{t}_{\Sigma_\tau}\omega|_{x=0}
		+N^{-\frac{1}{2}}\partial_x(\mathrm{t}_{U_x}\mathrm{n}_{\Sigma_\tau}\omega)|_{x=0}\wedge\beta^{\frac{1}{2}}\mathrm{d}\tau\,,
	\end{align*}
	where in the second equality we used the assumption $\mathrm{n}_{\partial\Sigma_\tau}\mathrm{n}_{\Sigma_\tau}\omega=0$, $\mathrm{n}_{\partial\Sigma_\tau}\mathrm{t}_{\Sigma_\tau}\omega=0$. Notice that the terms where either $N$ or $\beta$ are differentiated do not appear for the same reason. 
	
	On account of the hypotheses $\mathrm{n}_{\partial\Sigma_\tau}\mathrm{n}_{\Sigma_\tau}\omega=0$, $\mathrm{n}_{\partial\Sigma_\tau}\mathrm{t}_{\Sigma_\tau}\omega=0$ is follows that
	\begin{align*}
		\mathrm{nd}_{\partial\Sigma_\tau}\mathrm{n}_{\Sigma_t}\omega|_{\{\tau\}\times U_{\partial \Sigma}}
		&=N^{-\frac{1}{2}}\partial_x(\mathrm{t}_{U_x}\mathrm{n}_{\Sigma_\tau}\omega)|_{x=0}\wedge\beta^{\frac{1}{2}}\mathrm{d}\tau\,,\\
		\mathrm{nd}_{\partial\Sigma_\tau}\mathrm{t}_{\Sigma_t}\omega|_{\{\tau\}\times U_{\partial \Sigma}}
		&=N^{-\frac{1}{2}}\partial_x\mathrm{t}_{U_x}\mathrm{t}_{\Sigma_\tau}\omega|_{x=0}\,.
	\end{align*}
	It descends that $\mathrm{n}_{\partial M}\mathrm{d}\omega=0$ if and only if $\mathrm{nd}_{\partial\Sigma_\tau}\mathrm{n}_{\Sigma_t}\omega=0$ and $\mathrm{nd}_{\partial\Sigma_\tau}\mathrm{t}_{\Sigma_t}\omega|_{\{\tau\}\times U_{\partial \Sigma}}=0$.
	Together with Equation \eqref{Eq: n omega vanishes iff n t omega and nn omega do} this shows the thesis for $\perp$ boundary conditions.
\end{proof}

\section{Relative de Rham cohomology}\label{App: Poincare duality for manifold with boundary}

In this appendix we summarize a few definitions and results concerning de Rham cohomology and Poincar\'e duality, especially when the underlying manifold has a non-empty boundary. A reader interested in more details can refer to \cite{Bott-Tu-82,Schwarz-95}. 

For the purpose of this section $M$ refers to a smooth, oriented manifold of dimension $\dim M=m$ with a smooth boundary $\partial M$, together with an embedding map $\iota_{\partial M}:M\to\partial M$. In addition $\partial M$ comes endowed with orientation induced from $M$ via $\iota_{\partial M}$. We recall that $\Omega^\bullet(M)$ stands for the de Rham cochain complex which in degree  $k\in\mathbb{N}\cup\{0\}$ corresponds to $\Omega^k(M)$, the space of smooth $k$-forms. Observe that we shall need to work also with compactly supported forms and all definitions can be adapted accordingly.
To indicate this specific choice, we shall use a subscript $\mathrm{c}$, {\it e.g.} $\Omega^\bullet_{\mathrm{c}}(M)$. We denote instead the $k$-th de Rham cohomology group of $M$ as 
$$H^k(M)\doteq\frac{\ker(\mathrm{d}_k\colon\Omega^k(M)\to\Omega^{k+1}(M))}{\operatorname{Im} (\mathrm{d}_{k-1}\colon\Omega^{k-1}(M)\to\Omega^k(M))},$$
where we introduce the subscript $k$ to highlight that the differential operator $\mathrm{d}$ acts on $k$-forms. Equations \eqref{Eq: k-forms with vanishing tangential or normal component} and \eqref{Eq: relations-bulk-to-boundary} entail that we can define $\Omega^\bullet_{\mathrm{t}}(M)$, the subcomplex of $\Omega^\bullet(M)$, whose degree $k$ corresponds to $\Omega^k_{\mathrm{t}}(M)\subset\Omega^k(M)$. The associated de Rham cohomology groups will be denoted as $H^k_{\mathrm{t}}(M)$, $k\in\mathbb{N}\cup\{0\}$.

Similarly we can work with the codifferential $\delta$ in place of $\mathrm{d}$, hence identifying a chain complex $\Omega^\bullet(M)$ which in degree  $k\in\mathbb{N}\cup\{0\}$ corresponds to $\Omega^k(M)$, the space of smooth $k$-forms. The associated $k$-th homology groups will be denoted with 
$$H_k(M)\doteq\frac{\ker (\delta_k\colon\Omega^k(M)\to\Omega^{k-1}(M))}{\operatorname{Im} (\delta_{k+1}\colon\Omega^{k+1}(M)\to\Omega^k(M))}.$$
Equations \eqref{Eq: k-forms with vanishing tangential or normal component} and \eqref{Eq: relations-bulk-to-boundary} entail that we can define the $\Omega^\bullet_{\mathrm{n}}(M)$ (\textit{resp.} $\Omega^\bullet_{\mathrm{c}}(M)$, $\Omega^\bullet_{\mathrm{c,n}}(M)$), the subcomplex of $\Omega^\bullet(M)$, whose degree $k$ corresponds to $\Omega^k_{\mathrm{n}}(M)\subset\Omega^k(M)$ (\textit{resp.} $\Omega^k_{\mathrm{c}}(M),\Omega^k_{\mathrm{c,n}}(M)\subseteq\Omega^k(M)$).
The associated homology groups will be denoted as $H_{k,\mathrm{n}}(M)$ (\textit{resp.} $H_{k,\mathrm{c}}(M)$, $H_{k,\mathrm{c,n}}(M)$), $k\in\mathbb{N}\cup\{0\}$.
Observe that, in view of its definition and on account of equation \eqref{Eq: relations between d,delta,t,n}, the Hodge operator induces an isomorphism $H^k(M)\simeq H_{m-k}(M)$ which is realized as $H^k(M)\ni[\alpha]\mapsto [\star\alpha]\in H_{m-k}(M)$. Similarly, on account of Equation \eqref{Eq: relations-bulk-to-boundary}, it holds $H^k_{\mathrm{t}}(M)\simeq H_{m-k,\mathrm{n}}(M)$ and $H^k_{\mathrm{c,t}}(M)\simeq H_{m-k,\mathrm{c,n}}(M)$.

As last ingredient, we introduce the notion of relative cohomology, {\it cf.} \cite{Bott-Tu-82}. We start by defining the relative de Rham cochain complex $\Omega^\bullet(M;\partial M)$  which in degree  $k\in\mathbb{N}\cup\{0\}$ corresponds to 
\begin{align*}
	\Omega^k(M;\partial M)\doteq \Omega^k(M)\oplus\Omega^{k-1}(\partial M),
\end{align*}
endowed with the differential operator $\underline{\mathrm{d}}_k:\Omega^k(M;\partial M)\to\Omega^{k+1}(M;\partial M)$ such that for any $(\omega,\theta)\in\Omega^k(M;\partial M)$
\begin{equation}\label{Eq: relative-differential}
\underline{\mathrm{d}}_k(\omega,\theta)=(\mathrm{d}_k\omega,\mathrm{t}\omega-\mathrm{d}_{k-1}\theta)\,.
\end{equation}
Per construction, each $\Omega^k(M;\partial M)$ comes endowed naturally with the projections on each of the defining components, namely $\pi_1:\Omega^k(M;\partial M)\to\Omega^k(M)$ and $\pi_2:\Omega^k(M;\partial M)\to\Omega^k(\partial M)$. With a slight abuse of notation we make no explicit reference to $k$ in the symbol of these maps, since the domain of definition will always be clear from the context. The relative cohomology groups associated to $\underline{\mathrm{d}}_k$ will be denoted instead as $H^k(M;\partial M)$ and the following proposition characterizes the relation with the standard de Rham cohomology groups built on $M$ and on $\partial M$, {\it cf.} \cite[Prop. 6.49]{Bott-Tu-82}:

\begin{proposition}\label{Prop: long exact sequence}
Under the geometric assumptions specified at the beginning of the section, there exists an exact sequence
	\begin{equation}
	\ldots\to H^k(M;\partial M)\operatornamewithlimits{\longrightarrow}^{\pi_{1,*}} H^k(M) \operatornamewithlimits{\longrightarrow}^{\mathrm{t}_*} H^k(\partial M)\operatornamewithlimits{\longrightarrow}^{\pi_{2,*}} H^{k+1}(M;\partial M)\to\ldots,
	\end{equation}
	where $\pi_{1,*}$, $\pi_{2,*}$ and $\mathrm{t}_*$ indicate the natural counterpart of the maps $\pi_1$, $\pi_2$ and $\mathrm{t}$ at the level of cohomology groups.
\end{proposition}

\noindent The relevance of the relative cohomology groups in our analysis is highlighted by the following statement, of which we give a concise proof:

\begin{proposition}\label{Prop: equivalence description of relative cohomology}
Under the geometric assumptions specified at the beginning of the section, there exists an isomorphism between $H^k_{\mathrm{t}}(M)$ and $H^k(M;\partial M)$ for all $k\in\mathbb{N}\cup\{0\}$.
\end{proposition}

\begin{proof}
	Consider $\omega\in\Omega^k_{\mathrm{t}}(M)\cap\ker\mathrm{d}$ and let $(\omega,0)\in\Omega^k(M;\partial M)$, $k\in\mathbb{N}\cup\{0\}$. Equation \eqref{Eq: relative-differential} entails
	\begin{align*}
		\underline{\mathrm{d}}_k(\omega,0)=(\mathrm{d}_k\omega,\mathrm{t}\omega)=(0,0)\,.
	\end{align*}
	At the same time, if $\omega=\mathrm{d}_{k-1}\beta$ with $\beta\in\Omega_\mathrm{t}^{k-1}(M)$, then $(\mathrm{d}_{k-1}\beta,0)=\underline{\mathrm{d}}_{k-1}(\beta,0)$.
	Hence the embedding $\omega\mapsto (\omega,0)$ identifies a map $\rho: H^k_{\mathrm{t}}(M)\to  H^k(M;\partial M)$ such that $\rho([\omega])\doteq [(\omega,0)]$.
	
	To conclude, we need to prove that $\rho$ is surjective and injective.
	Let thus $[(\omega^\prime,\theta)]\in H^k(M;\partial M)$.
	It holds that $\mathrm{d}_k\omega^\prime=0$ and $\mathrm{t}\omega^\prime-\mathrm{d}_{k-1}\theta=0$.
	Recalling that $\mathrm{t}:\Omega^k(M)\to\Omega^k(\partial M)$ is surjective -- \textit{cf.} Remark \ref{Rmk: surjectivity of t,n,tdelta,nd} -- for all values of $k\in\mathbb{N}\cup\{0\}$, there must exists $\eta\in\Omega^{k-1}(M)$ such that $\mathrm{t}\eta=\theta$.
	Let $\omega\doteq\omega^\prime -\mathrm{d}_{k-1}\eta$.
	On account of \eqref{Eq: relations-bulk-to-boundary} $\omega\in\Omega^k_{\mathrm{t}}(M)\cap\ker\mathrm{d}_k$ and $(\omega, 0)$ is a representative if $[(\omega^\prime,\theta)]$ which entails that $\rho$ is surjective.
	
	Let $[\omega]\in H^k(M)$ be such that $\rho[\omega]=[0]\in H^k(M;\partial M)$.
	This implies that there exists $\beta\in\Omega^{k-1}(M)$, $\theta\in\Omega^{k-2}(\partial M)$ such that
	\begin{align*}
		(\omega,0)=
		\underline{\mathrm{d}}_{k-1}(\beta,\theta)=
		(\mathrm{d}_{k-1}\beta,\mathrm{t}\beta-\mathrm{d}_{k-2}\theta)\,.
	\end{align*}
	Let $\eta\in\Omega^{k-2}(M)$ be such that $\mathrm{t}\eta+\theta=0$.
	It follows that
	\begin{align*}
		(\omega,0)=
		\underline{\mathrm{d}}_{k-1}\left((\beta,\theta)+
		\underline{\mathrm{d}}_{k-2}(\eta,0)\right)=
		\underline{\mathrm{d}}_{k-1}(\beta+\mathrm{d}_{k-2}\eta,0)\,.
	\end{align*}
	This entails that $\omega=\mathrm{d}_{k-1}(\beta+\mathrm{d}_{k-2}\eta)$ where $\mathrm{t}(\beta+\mathrm{d}_{k-2}\eta)=0$.
	It follows that $[\omega]=0$ that is $\rho$ is injective.

\end{proof}

\noindent To conclude, we recall a notable result concerning the relative cohomology, which is a specialization to the case in hand of the Poincar\'e-Lefschetz duality, an account of which can be found in \cite{Maunder}:

\begin{proposition}\label{Prop: Poin-Lefs duality}
	Under the geometric assumptions specified at the beginning of the section and assuming in addition that $M$ admits a finite good cover, it holds that, for all $k\in\mathbb{N}\cup\{0\}$
   \begin{align}
   		H^{m-k}(M;\partial M)\simeq H^k_{\mathrm{c}}(M)^*\,,\qquad
   		[\alpha]\to\bigg(H^k_{\mathrm{c}}(M)\ni[\eta]\mapsto\int_M\overline{\alpha}\wedge\eta\in\mathbb{C}\bigg)\,.
   \end{align}
   where $m=\dim M$ and where on the right hand side we consider the dual of the $(m-k)$-th cohomology group built out compactly supported forms.
\end{proposition}

\begin{remark}\label{Rmk: consequence of Poincare--Lefschetz duality}
	On account of Propositions \ref{Prop: equivalence description of relative cohomology}-\ref{Prop: Poin-Lefs duality} and of the isomorphisms $H^k_{(\mathrm{c})}(M)\simeq H^{m-k}_{(\mathrm{c})}(M)$ the following are isomorphisms:
	\begin{align}\label{Eq: relative cohomology isomorphic to dual of compactly supported differential homology }
		H^k_{\mathrm{t}}(M)\simeq
		H^{m-k}_{\mathrm{c}}(M)^*\simeq
		H_{k,\mathrm{c}}(M)^*\,,\qquad
		H^k(M)\simeq H_{k,\mathrm{c,n}}(M)^*\,.
	\end{align}
\end{remark}

\subsection*{Acknowledgements}
We are grateful to Marco Benini, Jorma Louko, Sonia Mazzucchi, Valter Moretti, Ana Alonso Rodrigez, Alexander Schenkel and Alberto Valli for the useful discussions.
We are thankful to the anonymous referee who pointed out a flaw in one the proofs leading us as a consequence to a great improvement of our results.
R. L. is very grateful for the kind hospitality extended by the Department of Mathematics at the University of Trento where part of this work was carried out.
The work of C. D. was supported by the University of Pavia, while that of N. D. was supported in part by a research fellowship of the University of Trento. R. L. gratefully acknowledges financial support from by the Istituto Universitario di Studi Superiori - IUSS. This work is based partly on the MSc thesis of R. L. at the Department of Physics of the University of Pavia.

\end{document}